\documentclass[11pt,a4paper]{article}

\usepackage[margin=1.12in]{geometry}
\usepackage[utf8]{inputenc}
\usepackage[T1]{fontenc}
\usepackage{lmodern}
\usepackage{microtype}
\usepackage{amsmath}
\usepackage{amssymb}
\usepackage{amsthm}
\usepackage{aliascnt}
\usepackage{enumerate}
\usepackage{xcolor}
\usepackage{adjustbox}
\usepackage{tikz}
\usetikzlibrary{arrows.meta,positioning,fit,shapes.geometric}

\usepackage[backend=biber,bibstyle=numeric,citestyle=numeric-comp,maxbibnames=20,sorting=nyt]{biblatex}
\usepackage{hyperref}
\hypersetup{colorlinks=true,allcolors=blue,unicode=true,pdftitle={Kempe Reachability with k >= 3 Colors Is PSPACE-Complete on Claw-Free Graphs},pdfauthor={Duc A. Hoang},pdfkeywords={Kempe Reachability, Kempe changes, edge-colorings, reconfiguration, claw-free graphs, line graphs, PSPACE-completeness}}
\usepackage[capitalise,noabbrev]{cleveref}
\crefname{enumi}{part}{parts}
\Crefname{enumi}{Part}{Parts}

\newcommand{\KRfull}{\textnormal{\textsc{Kempe Reachability}}}
\newcommand{\KR}{\KRfull}

\newcommand{\PSPACE}{\mathsf{PSPACE}}
\newcommand{\NP}{\mathsf{NP}}
\newcommand{\PClass}{\mathsf{P}}

\theoremstyle{plain}
\newtheorem{theorem}{Theorem}[section]
\newaliascnt{lemma}{theorem}
\newtheorem{lemma}[lemma]{Lemma}
\aliascntresetthe{lemma}
\newaliascnt{proposition}{theorem}
\newtheorem{proposition}[proposition]{Proposition}
\aliascntresetthe{proposition}
\newaliascnt{corollary}{theorem}
\newtheorem{corollary}[corollary]{Corollary}
\aliascntresetthe{corollary}
\theoremstyle{definition}
\newaliascnt{definition}{theorem}
\newtheorem{definition}[definition]{Definition}
\aliascntresetthe{definition}
\theoremstyle{remark}
\newaliascnt{remark}{theorem}

\aliascntresetthe{remark}

\begin{document}

\begin{center}
{\LARGE\bfseries Kempe Reachability with $k\geq 3$ Colors Is $\mathsf{PSPACE}$-Complete on Claw-Free Graphs\par}
\vspace{0.75em}
Duc A.\ Hoang\\
VNU University of Science, Vietnam National University\\
Hanoi, Vietnam\\
\href{mailto:hoanganhduc@hus.edu.vn}{\texttt{hoanganhduc@hus.edu.vn}}
\end{center}

\begin{abstract}
A \emph{proper vertex coloring} assigns colors to vertices so that adjacent vertices have distinct colors.
A \emph{Kempe change} in a proper vertex coloring exchanges two colors on one component of the subgraph induced by the vertices of those colors.
For a fixed palette of $k$ colors, \textsc{Kempe Reachability} asks whether two given proper vertex colorings of the same graph are connected by a sequence of Kempe changes.
This problem is known to be $\mathsf{PSPACE}$-complete for every fixed $k\geq 3$.
A graph is \emph{claw-free} if no vertex has three pairwise nonadjacent neighbors.
The \emph{line graph} of a graph has one vertex for each edge, with two vertices adjacent exactly when the corresponding edges share an end.
Line graphs are claw-free, and their Kempe changes correspond to \emph{edge-Kempe changes} in the original graph, which exchange two colors on one component of the subgraph formed by the edges of those colors.
A claw-free graph with a proper $k$-coloring has maximum degree at most $2k-2$, whereas for $k\geq 4$ all proper $k$-colorings of a graph of maximum degree at most $k$ are connected by Kempe changes.
We show that, for every fixed $k\geq 3$, \textsc{Kempe Reachability} with $k$ colors is $\mathsf{PSPACE}$-complete even on line graphs of connected simple $k$-regular graphs, and hence even on connected $(2k-2)$-regular claw-free graphs.
Equivalently, deciding whether two proper $k$-edge-colorings are connected by edge-Kempe changes is $\mathsf{PSPACE}$-complete even on connected simple $k$-regular graphs.

ChatGPT 6 Pro and Claude Fable 5.1 were used to assist with developing and reviewing mathematical arguments and drafting and revising the paper.
The author takes full responsibility for the correctness and content of the paper.
\end{abstract}

\section{Introduction}\label{sec:intro}

Combinatorial reconfiguration studies how the feasible solutions of a problem can be transformed into one another by small changes.
A \emph{reconfiguration problem} fixes an underlying problem and a rule that changes one feasible solution into another, and it asks whether two given feasible solutions are connected under that rule.
For overviews of combinatorial reconfiguration, including coloring problems and parameterized complexity, we refer readers to the surveys~\cite{vandenHeuvel2013,Nishimura2018,MynhardtNasserasr2019,Bousquet2024survey}. 

\emph{Kempe changes} are a classical move for graph colorings.
Kempe~\cite{Kempe1879} used them in his attempted proof of the four color theorem. Exchanging two colors on a whole component of the subgraph induced by these colors preserves properness, whereas recoloring a single vertex need not do so.
A concrete instance is a cubic graph with a proper coloring of its edges by three colors, that is, a \emph{Tait coloring}.
Since every vertex is incident with exactly one edge of each color, the edges of any two colors form disjoint even cycles, the \emph{bichromatic cycles}, and exchanging the two colors along one of these cycles gives another Tait coloring.
For example, each Tait coloring of the complete graph $K_4$ gives its three perfect matchings three distinct colors, and the edges of any two colors form a Hamiltonian cycle of $K_4$.
An exchange along such a cycle therefore exchanges two color classes, and so all six Tait colorings of $K_4$ are connected by exchanges.
In general, one asks whether two given Tait colorings of a cubic graph are connected by exchanges along bichromatic cycles.
Kotzig~\cite{Kotzig1975} characterized the subgraphs formed by the edges on which two Tait colorings of a cubic graph differ, and he remarked that not every such subgraph arises from a sequence of exchanges along bichromatic cycles~\cite{Kotzig1975}. Mohar~\cite{Mohar2007} asked for which cubic bipartite graphs all Tait colorings are connected in this way, and he noted there that the Tait colorings of $K_{3,3}$ fall into two classes. 

Let $k\geq 1$ be a fixed positive integer.
Given a graph $G$ and two proper $k$-colorings $\alpha_G$ and $\beta_G$ of $G$, \KR{} asks whether $\alpha_G$ can be transformed into $\beta_G$ by a sequence of Kempe changes.
The equivalence classes of proper $k$-colorings under sequences of Kempe changes are the \emph{Kempe classes} of $G$, so \KR{} asks whether two given colorings lie in the same Kempe class.
The problem is in $\PSPACE$ by the standard reachability argument~\cite{Bonamy2020}. It was shown by Bonamy et al.~\cite{Bonamy2020} that the problem is $\PSPACE$-complete for every fixed $k\geq 3$, and for $k=3$ even on planar graphs of maximum degree six. 

The positive results depend on the relation between the number of colors and the maximum degree.
The \emph{maximum degree} $\Delta(G)$ of a graph $G$ is the largest degree of its vertices.
If $G$ is $d$-degenerate for an integer $d\geq 0$ and $k>d$, then all proper $k$-colorings of $G$ lie in one Kempe class, by a result of Las Vergnas and Meyniel recalled by Mohar~\cite{Mohar2007} and Feghali, Johnson, and Paulusma~\cite{Feghali2017}. For a connected graph $G$ with a proper $k$-coloring and $k\geq\Delta(G)\geq 3$, all proper $k$-colorings of $G$ lie in one Kempe class unless $k=3$ and $G$ is the triangular prism~\cite{Bonamy2019}. For three colors and connected cubic graphs, the same conclusion was proved earlier with the two exceptions $K_4$ and the triangular prism~\cite{Feghali2017}; the graph $K_4$ has no proper $3$-coloring, and the prism has two Kempe classes~\cite{vandenHeuvel2013}, \cite{Feghali2017}. Since a Kempe change acts inside one component of $G$, these results apply componentwise.
For $k\geq3$ and $\Delta(G)\leq k$, the result of Bonamy et al.~\cite{Bonamy2019} applies to every component of maximum degree at least three, with the triangular-prism exception when $k=3$. Every component of maximum degree at most two is $2$-degenerate, so the degeneracy result applies to it whenever $k\geq3$.
Hence, for $k\geq 4$ every instance with $\Delta(G)\leq k$ is a yes-instance, that is, its answer is yes, and for $k=3$ every instance with $\Delta(G)\leq 3$ and no component isomorphic to the triangular prism is a yes-instance.
So, apart from components isomorphic to the triangular prism when $k=3$, \KR{} with $k$ colors is interesting only on graphs of maximum degree greater than $k$.

For edge-colorings, the relevant parameter is the \emph{chromatic index} $\chi'(G)$ of a graph $G$, the least number of colors in a proper edge-coloring of $G$.
Holyer~\cite{Holyer1981} showed that it is $\NP$-complete to decide whether a cubic graph has chromatic index three, that is, whether it has a Tait coloring. McDonald, Mohar, and Scheide~\cite{McDonald2012} showed that all proper $(\Delta(G)+1)$-edge-colorings of a graph $G$ with $\Delta(G)\leq 3$ are connected by \emph{edge-Kempe changes}, which exchange two colors on one component of the subgraph formed by the edges of these two colors. For every graph $G$, the same holds for all proper $k$-edge-colorings with $k\geq\chi'(G)+2$, by a result of Mohar~\cite{Mohar2007}, and with $k=\chi'(G)+1$, by a result of Narboni~\cite{Narboni2023}. So reachability under edge-Kempe changes with $k$ colors is interesting only for $k=\chi'(G)$, since there is no proper $k$-edge-coloring for $k<\chi'(G)$; the case $k=\chi'(G)$ is the setting of \cref{thm:equivalence} and of the last statement of \cref{thm:main}.

\emph{Claw-free graphs}, the graphs with no induced $K_{1,3}$, are a natural class on which to study \KR{}, for three reasons.
First, the \emph{line graph} $L(G)$ of a graph $G$, whose vertices are the edges of $G$, two of them adjacent when they share an end, is claw-free, and the Kempe changes of $L(G)$ are the edge-Kempe changes of $G$ (\cref{lem:transfer}).
Hence, on line graphs, \KR{} is the reachability problem for edge-colorings, to which the questions on Tait colorings and the results on edge-Kempe changes recalled above belong.
Second, a claw-free graph with a proper $k$-coloring has maximum degree at most $2k-2$, because the neighborhood of a vertex $v$ meets each of the $k-1$ color classes that do not contain $v$ in at most two vertices.
Since for $k\geq 4$ every instance of maximum degree at most $k$ is a yes-instance, hard instances on claw-free graphs can only have maximum degree between $k+1$ and $2k-2$.
Third, Belavadi and Cameron~\cite{BelavadiCameron2025} characterized the graphs $H$ such that, for every graph $G$ with no induced subgraph isomorphic to $H$ and every $k$ at least the chromatic number of $G$, all proper $k$-colorings of $G$ lie in one Kempe class: these are the induced subgraphs of the path on four vertices. Since the claw is not one of them, some claw-free graphs have more than one Kempe class for some number of colors, so the question whether two given colorings lie in the same Kempe class is not trivial on claw-free graphs.

In this paper, we show that \KR{} is $\PSPACE$-complete on claw-free graphs, for every fixed number $k\geq 3$ of colors, with instances of the largest possible maximum degree $2k-2$.
The following theorem, which is proved in \cref{sec:closing}, states this result precisely.

\begin{theorem}\label{thm:main}
For every fixed integer $k\geq 3$, \KR{} with $k$ colors is $\PSPACE$-complete even when the input graph is the line graph of a connected simple $k$-regular graph.
In particular, it is $\PSPACE$-complete even when the input graph is a connected $(2k-2)$-regular claw-free graph.
Moreover, deciding whether two proper $k$-edge-colorings of a graph are connected by edge-Kempe changes is $\PSPACE$-complete even on connected simple $k$-regular graphs.
\end{theorem}

The degree $2k-2$ in \cref{thm:main} is the largest possible for a claw-free graph that admits a proper $k$-coloring.
For fixed $k\geq4$, the connectivity results above settle every instance with maximum degree at most $k$: any two proper $k$-colorings are reachable from one another.
Our hardness construction uses $(2k-2)$-regular graphs, so it settles the other end of the possible degree range.
Those hard instances do not satisfy a smaller maximum-degree bound, and hardness for a graph class does not automatically imply hardness for its subclasses.
Thus, the intermediate bounds $k+1\leq d\leq2k-3$ are not settled by these results.
For example, with four colors, the connectivity results recalled above~\cite{Mohar2007,Bonamy2019} cover maximum degree at most four. Our hard instances have degree six by \cref{thm:main}, so the bound five requires a different argument.

We now explain the ideas behind the reduction and, at the same time, outline the rest of the paper.
We will encode source configurations by edge-colorings of a connected simple $k$-regular graph.
Taking its line graph will then give the required claw-free vertex-coloring instance.
The five steps in \cref{fig:reduction-overview}(a) track the graphs and their two configurations through this construction.

We start with the nondeterministic constraint logic (NCL) problem of Hearn and Demaine~\cite{HearnDemaine2005}. For fixed $k\geq3$, our input is a nonempty connected simple cubic $\mathrm{AND}/\mathrm{OR}$ constraint graph $X_0$ with two legal orientations $x_s^0,x_t^0$.
Each source edge can be viewed as a switch with two positions, given by its direction.
An $\mathrm{AND}$ vertex has one edge of weight two and two of weight one.
An $\mathrm{OR}$ vertex has three edges of weight two.
A configuration is legal when every vertex has incoming weight at least two.
A move reverses one edge while keeping this rule.
Deciding whether $x_s^0$ can be changed into $x_t^0$ by these moves is $\PSPACE$-complete even under these restrictions (\cref{prop:source}).

We construct a connected $(2k-2)$-regular claw-free graph $H$ and two proper $k$-vertex-colorings $\alpha,\beta$ in polynomial time.
Here $\sim_{X_0}$ denotes reachability by legal reversals, and $\sim_H$ denotes Kempe reachability.
The required equivalence is
\[
x_s^0\sim_{X_0}x_t^0\quad\Longleftrightarrow\quad\alpha\sim_H\beta.
\]
Our main task is to build small graphs whose edge-colorings record the source directions and enforce the rules at the source vertices.
We must show both that legal source moves can be simulated and that other intermediate colorings cannot create connections absent from the source.

In \cref{sec:prelim}, we introduce the notation and prepare the source for this simulation (Step~1).
At an $\mathrm{OR}$ vertex, the inflow rule allows seven local patterns of edge directions: all except the one with three outward edges.
Our local simulation uses the simpler case in which two specified edges never both point inward, leaving only five patterns to represent.
We obtain this case by replacing each $\mathrm{OR}$ vertex with a small network of $\mathrm{AND}$ and $\mathrm{OR}$ vertices, using the construction of Hearn and Demaine~\cite{HearnDemaine2009}. At every $\mathrm{OR}$ vertex of the new graph, two designated edges form a \emph{protected pair}: the inflow rules of the whole graph prevent them from both pointing inward in any legal configuration.
This replacement preserves reachability between corresponding configurations and introduces no loops.

The pieces used to simulate the two ends of a source edge also need consistent color labels.
We therefore modify the source further to obtain a Tait coloring $\tau$, while preserving the protected pairs and reachability.
Its three colors give distinct labels to the three edges at each vertex and the same label to an edge at both ends.
These labels will let us join the local colorings into one proper coloring of the whole network.
The resulting constraint multigraph $X$ is connected, loopless and cubic, and its two legal orientations $x_s,x_t$ correspond to $x_s^0,x_t^0$.
The new orientations are reachable from one another exactly when the original ones are.
The coloring $\tau$ stays fixed; the source configurations still change only edge directions.

Step~2 begins in \cref{sec:construction}, where we build the network $F_X^k$ shown in \cref{fig:reduction-overview}(b).
Each source vertex is replaced by a small graph called a \emph{gate}, which helps enforce its incoming-weight rule.
Each source edge is replaced by a \emph{wire} joining the corresponding gates and carrying its direction.
We build these pieces with free edge ends so that they can be joined to one another.
At each end of a wire, a pair of edges acts as a two-state switch: its colors either agree or differ.
A fixed rule for that position in the gate interprets these two possibilities as inward and outward.
The wire must make its two ends describe one source direction: inward at one end and outward at the other.
We will choose proper $k$-edge-colorings $c_s,c_t$ of the network to represent $x_s,x_t$.

To make the two ends agree on this direction, we use a \emph{negator}, a small multigraph with two specified ordered pairs of free edge ends.
In every proper $k$-edge-coloring, exactly one of these pairs has equal colors.
We build negators from the inverting components of Leven and Galil~\cite{LevenGalil1983}. Along a wire, the same-color and different-color pairs alternate.
In the figure, the edge $f=uv$ becomes a wire of two negators, while $g=vw$ becomes a wire of three.
The choice accounts for how the two incident positions in the gates interpret the color pairs.

The gates also need to coordinate their three incident wires.
For this purpose, we add one triangle for each source vertex, as shown below the gates in \cref{fig:reduction-overview}(b).
The thick black bundles join each triangle vertex to a negator at the end of the corresponding wire.
Together, the gates and triangles enforce the source inflow rules, so every proper coloring of the network can be read as a legal source configuration.

To rule out extra connections, we must also control how a Kempe change alters the recorded directions.
A triangle has three distinct edge colors, so the two edges at each of its vertices give a different color pair.
When the two edges at a wire end have different colors, the attached bundle forces them to use the pair at the corresponding triangle vertex.
For a Kempe change to alter a recorded direction, its two exchanged colors must be the pair at the corresponding triangle vertex.
Thus, the triangle acts like a local lock: one exchange can change at most one of the three incident source directions.
Several source edges may change direction together, but they have no common endpoint and can be reversed one at a time through legal configurations.
Reading each coloring of a Kempe sequence in this way gives a sequence of legal source reversals.
This argument applies to arbitrary intermediate colorings, so extra freedom inside the gadgets cannot connect source configurations that were originally disconnected.

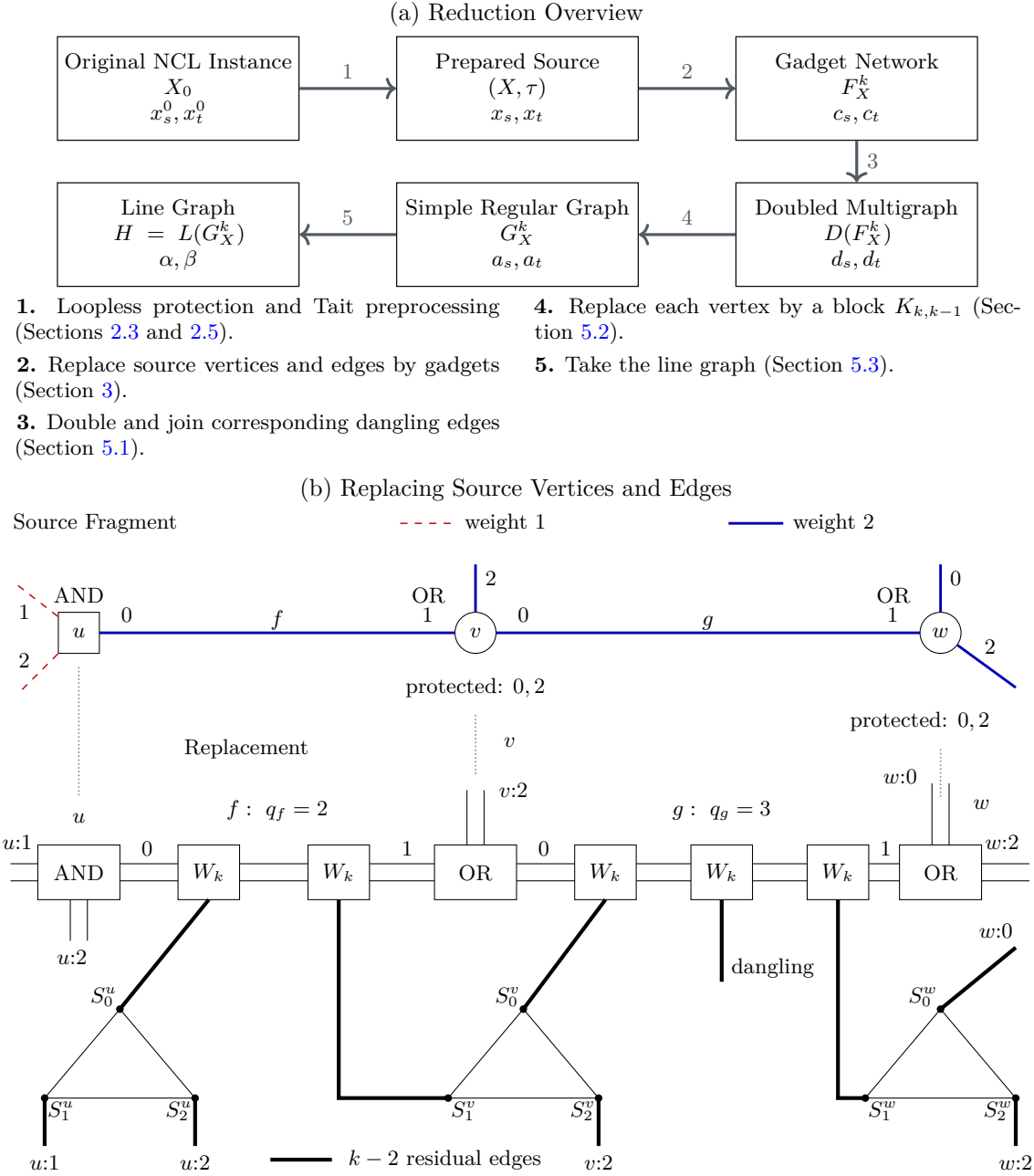
\begin{figure}[ht]
\centering
\begin{adjustbox}{max width=\textwidth}
\begin{minipage}{150mm}
\centering
{\small (a) Reduction Overview\par}
\smallskip
\begingroup
\definecolor{tikzdrawPrimary}{HTML}{1F4E79}
\definecolor{tikzdrawAccent}{HTML}{D97A00}
\definecolor{tikzdrawNeutral}{HTML}{505A64}
\definecolor{tikzdrawLight}{HTML}{E9EEF3}
\definecolor{tikzdrawSuccess}{HTML}{2E7D32}
\definecolor{tikzdrawDanger}{HTML}{B3261E}
\tikzset{
  box/.style={
    draw=tikzdrawPrimary,
    rounded corners=2pt,
    very thick,
    minimum height=8mm,
    text centered,
    align=center,
    inner sep=2mm
  },
  io/.style={
    box,
    trapezium,
    trapezium left angle=70,
    trapezium right angle=110
  },
  decision/.style={
    draw=tikzdrawAccent,
    diamond,
    aspect=2,
    text centered,
    inner sep=1.5mm
  },
  groupbox/.style={
    draw=tikzdrawNeutral,
    dashed,
    rounded corners=2pt,
    inner sep=3mm
  },
  edge/.style={
    ->,
    very thick,
    draw=tikzdrawNeutral
  },
  graphnode/.style={
    circle,
    draw=tikzdrawPrimary,
    fill=white,
    very thick,
    minimum size=5.5mm,
    inner sep=0pt,
    font=\scriptsize
  },
  graphedge/.style={
    very thick,
    draw=tikzdrawNeutral
  },
  note/.style={
    font=\small,
    text=tikzdrawNeutral
  }
}
\begin{tikzpicture}[node distance=10mm and 14mm]
\node[box, draw=black, rounded corners=0pt, line width=0.45pt, text width=33mm, inner sep=1mm, font=\fontsize{9}{10.5}\selectfont, node distance=6mm and 14mm, minimum width=35mm, minimum height=15mm] (original) {Original NCL Instance\\$X_0$\\$x_s^0,x_t^0$};
\node[box, draw=black, rounded corners=0pt, line width=0.45pt, text width=33mm, inner sep=1mm, font=\fontsize{9}{10.5}\selectfont, node distance=6mm and 14mm, minimum width=35mm, minimum height=15mm, right=of original] (source) {Prepared Source\\$(X,\tau)$\\$x_s,x_t$};
\node[box, draw=black, rounded corners=0pt, line width=0.45pt, text width=33mm, inner sep=1mm, font=\fontsize{9}{10.5}\selectfont, node distance=6mm and 14mm, minimum width=35mm, minimum height=15mm, right=of source] (network) {Gadget Network\\$F_X^k$\\$c_s,c_t$};
\node[box, draw=black, rounded corners=0pt, line width=0.45pt, text width=33mm, inner sep=1mm, font=\fontsize{9}{10.5}\selectfont, node distance=6mm and 14mm, minimum width=35mm, minimum height=15mm, below=of network] (doubled) {Doubled Multigraph\\$D(F_X^k)$\\$d_s,d_t$};
\node[box, draw=black, rounded corners=0pt, line width=0.45pt, text width=33mm, inner sep=1mm, font=\fontsize{9}{10.5}\selectfont, node distance=6mm and 14mm, minimum width=35mm, minimum height=15mm, left=of doubled] (simple) {Simple Regular Graph\\$G_X^k$\\$a_s,a_t$};
\node[box, draw=black, rounded corners=0pt, line width=0.45pt, text width=33mm, inner sep=1mm, font=\fontsize{9}{10.5}\selectfont, node distance=6mm and 14mm, minimum width=35mm, minimum height=15mm, left=of simple] (line) {Line Graph\\$H=L(G_X^k)$\\$\alpha,\beta$};
\draw[edge] (original) -- node[above, note] {{\fontsize{9}{10.5}\selectfont 1}} (source);
\draw[edge] (source) -- node[above, note] {{\fontsize{9}{10.5}\selectfont 2}} (network);
\draw[edge] (network) -- node[right, note] {{\fontsize{9}{10.5}\selectfont 3}} (doubled);
\draw[edge] (doubled) -- node[above, note] {{\fontsize{9}{10.5}\selectfont 4}} (simple);
\draw[edge] (simple) -- node[above, note] {{\fontsize{9}{10.5}\selectfont 5}} (line);
\end{tikzpicture}
\endgroup

\par\smallskip
\begin{minipage}{145mm}
\fontsize{9}{10.5}\selectfont
\setlength{\parindent}{0pt}
\begin{minipage}[t]{70mm}
\textbf{1.} Loopless protection and Tait preprocessing (\cref{sec:ncl,sec:tait-source}).\par
\smallskip
\textbf{2.} Replace source vertices and edges by gadgets (\cref{sec:construction}).\par
\smallskip
\textbf{3.} Double and join corresponding dangling edges (\cref{sec:double}).
\end{minipage}\hfill
\begin{minipage}[t]{70mm}
\textbf{4.} Replace each vertex by a block $K_{k,k-1}$ (\cref{sec:blocks}).\par
\smallskip
\textbf{5.} Take the line graph (\cref{sec:main-proof}).
\end{minipage}
\end{minipage}

\par\medskip
{\small (b) Replacing Source Vertices and Edges\par}
\smallskip
\begin{adjustbox}{max width=\linewidth}
\begin{tikzpicture}[x=1mm,y=1mm,
 font=\fontsize{9}{10.5}\selectfont,
 lb/.style={inner sep=1pt},
 srcand/.style={draw,rectangle,fill=white,minimum size=6mm,inner sep=0pt},
 srcor/.style={draw,circle,fill=white,minimum size=6mm,inner sep=0pt},
 gate/.style={draw,fill=white,minimum width=12mm,minimum height=8mm,inner sep=0pt},
 neg/.style={draw,fill=white,minimum width=9mm,minimum height=8mm,inner sep=0pt},
 tv/.style={circle,fill=black,minimum size=1.1mm,inner sep=0pt},
 wtone/.style={red!75!black,dashed,line width=.55pt},
 wttwo/.style={blue!70!black,line width=1.1pt},
 bundle/.style={line width=1.6pt},
 corr/.style={gray,densely dotted,line width=.6pt}]
\node[lb,anchor=west] at (0,94) {Source Fragment};
\draw[wtone] (57,94)--(65,94);
\node[lb,anchor=west] at (66,94) {weight 1};
\draw[wttwo] (105,94)--(113,94);
\node[lb,anchor=west] at (114,94) {weight 2};
\node[srcand] (srcu) at (10,78) {$u$};
\node[srcor] (srcv) at (68,78) {$v$};
\node[srcor] (srcw) at (136,78) {$w$};
\node[lb,anchor=south] at (10,82) {AND};
\node[lb,anchor=south east] at (64,82) {OR};
\node[lb,anchor=south east] at (132,82) {OR};
\coordinate (srcu1) at (1,85);
\coordinate (srcu2) at (1,69);
\coordinate (srcv2) at (68,88);
\coordinate (srcw0) at (136,88);
\coordinate (srcw2) at (147,70);
\draw[wtone] (srcu)--(srcu1);
\draw[wtone] (srcu)--(srcu2);
\coordinate (srcu0) at (13,78);
\coordinate (srcv1) at (65,78);
\coordinate (srcv0) at (71,78);
\coordinate (srcw1) at (133,78);
\draw[wttwo] (srcu0)--(srcv1);
\draw[wttwo] (srcv0)--(srcw1);
\draw[wttwo] (srcv)--(srcv2);
\draw[wttwo] (srcw)--(srcw0);
\draw[wttwo] (srcw)--(srcw2);
\node[lb,anchor=south] at (39,78) {$f$};
\node[lb,anchor=south] at (102,78) {$g$};
\node[lb,anchor=south] at (17,79.2) {$0$};
\node[lb,anchor=south] at (61,79.2) {$1$};
\node[lb,anchor=south] at (75,79.2) {$0$};
\node[lb,anchor=south] at (129,79.2) {$1$};
\node[lb] at (2,81) {$1$};
\node[lb] at (2,74) {$2$};
\node[lb,anchor=west] at (69,86) {$2$};
\node[lb,anchor=west] at (137,86) {$0$};
\node[lb,anchor=west] at (142,76) {$2$};
\node[lb] at (68,70) {protected: $0,2$};
\node[lb] at (133,65) {protected: $0,2$};
\draw[corr] (10,73)--(10,54);
\draw[corr] (68,66)--(68,57);
\draw[corr] (136,61)--(136,54);
\node[lb,anchor=west] at (25,61) {Replacement};
\node[gate] (gateu) at (10,43) {AND};
\node[gate] (gatev) at (68,43) {OR};
\node[gate] (gatew) at (136,43) {OR};
\node[lb] at (10,51) {$u$};
\node[lb] at (73,62) {$v$};
\node[lb] at (142,53) {$w$};
\node[neg] (wf0) at (29,43) {$W_k$};
\node[neg] (wf1) at (48,43) {$W_k$};
\node[neg] (wg0) at (87,43) {$W_k$};
\node[neg] (wg1) at (104,43) {$W_k$};
\node[neg] (wg2) at (121,43) {$W_k$};
\node[lb] at (39,52) {$f:\ q_f=2$};
\node[lb] at (104,52) {$g:\ q_g=3$};
\coordinate (u0c1) at (16,44.25);
\coordinate (u0c2) at (16,41.75);
\coordinate (v1c1) at (62,44.25);
\coordinate (v1c2) at (62,41.75);
\coordinate (v0c1) at (74,44.25);
\coordinate (v0c2) at (74,41.75);
\coordinate (w1c1) at (130,44.25);
\coordinate (w1c2) at (130,41.75);
\coordinate (f0ic1) at (24.5,44.25);
\coordinate (f0ic2) at (24.5,41.75);
\coordinate (f0oc1) at (33.5,44.25);
\coordinate (f0oc2) at (33.5,41.75);
\coordinate (f1ic1) at (43.5,44.25);
\coordinate (f1ic2) at (43.5,41.75);
\coordinate (f1oc1) at (52.5,44.25);
\coordinate (f1oc2) at (52.5,41.75);
\coordinate (g0ic1) at (82.5,44.25);
\coordinate (g0ic2) at (82.5,41.75);
\coordinate (g0oc1) at (91.5,44.25);
\coordinate (g0oc2) at (91.5,41.75);
\coordinate (g1ic1) at (99.5,44.25);
\coordinate (g1ic2) at (99.5,41.75);
\coordinate (g1oc1) at (108.5,44.25);
\coordinate (g1oc2) at (108.5,41.75);
\coordinate (g2ic1) at (116.5,44.25);
\coordinate (g2ic2) at (116.5,41.75);
\coordinate (g2oc1) at (125.5,44.25);
\coordinate (g2oc2) at (125.5,41.75);
\coordinate (f0r) at (29,39);
\coordinate (f1r) at (48,39);
\coordinate (g0r) at (87,39);
\coordinate (g1r) at (104,39);
\coordinate (g2r) at (121,39);
\coordinate (Su0) at (16,23);
\coordinate (Su1) at (5,10);
\coordinate (Su2) at (27,10);
\coordinate (Sv0) at (75,23);
\coordinate (Sv1) at (64,10);
\coordinate (Sv2) at (86,10);
\coordinate (Sw0) at (136,23);
\coordinate (Sw1) at (125,10);
\coordinate (Sw2) at (147,10);
\coordinate (cropu1c1) at (0,44.25);
\coordinate (cropu1c2) at (0,41.75);
\coordinate (u1c1) at (4,44.25);
\coordinate (u1c2) at (4,41.75);
\coordinate (u2c1) at (8.75,39);
\coordinate (u2c2) at (11.25,39);
\coordinate (cropu2c1) at (8.75,33);
\coordinate (cropu2c2) at (11.25,33);
\coordinate (v2c1) at (66.75,47);
\coordinate (v2c2) at (69.25,47);
\coordinate (cropv2c1) at (66.75,55);
\coordinate (cropv2c2) at (69.25,55);
\coordinate (w0c1) at (134.75,47);
\coordinate (w0c2) at (137.25,47);
\coordinate (cropw0c1) at (134.75,56);
\coordinate (cropw0c2) at (137.25,56);
\coordinate (w2c1) at (142,44.25);
\coordinate (w2c2) at (142,41.75);
\coordinate (cropw2c1) at (149,44.25);
\coordinate (cropw2c2) at (149,41.75);
\coordinate (cropSu1) at (5,3);
\coordinate (cropSu2) at (27,3);
\coordinate (cropSv2) at (86,3);
\coordinate (cropSw0) at (147,32);
\coordinate (cropSw2) at (147,3);
\coordinate (g1dangling) at (104,27);
\draw (u0c1)--(f0ic1);
\draw (u0c2)--(f0ic2);
\draw (f0oc1)--(f1ic1);
\draw (f0oc2)--(f1ic2);
\draw (f1oc1)--(v1c1);
\draw (f1oc2)--(v1c2);
\draw (v0c1)--(g0ic1);
\draw (v0c2)--(g0ic2);
\draw (g0oc1)--(g1ic1);
\draw (g0oc2)--(g1ic2);
\draw (g1oc1)--(g2ic1);
\draw (g1oc2)--(g2ic2);
\draw (g2oc1)--(w1c1);
\draw (g2oc2)--(w1c2);
\draw (u1c1)--(cropu1c1);
\draw (u1c2)--(cropu1c2);
\draw (u2c1)--(cropu2c1);
\draw (u2c2)--(cropu2c2);
\draw (v2c1)--(cropv2c1);
\draw (v2c2)--(cropv2c2);
\draw (w0c1)--(cropw0c1);
\draw (w0c2)--(cropw0c2);
\draw (w2c1)--(cropw2c1);
\draw (w2c2)--(cropw2c2);
\node[lb,anchor=south] at (20,45) {$0$};
\node[lb,anchor=south] at (58,45) {$1$};
\node[lb,anchor=south] at (78,45) {$0$};
\node[lb,anchor=south] at (128,45) {$1$};
\node[lb,anchor=south] at (1,46) {$u{:}1$};
\node[lb,anchor=north] at (9,32) {$u{:}2$};
\node[lb,anchor=west] at (71,55) {$v{:}2$};
\node[lb,anchor=east] at (133,57) {$w{:}0$};
\node[lb,anchor=south] at (145,46) {$w{:}2$};
\draw[bundle] (f0r)--(Su0);
\draw[bundle] (f1r)--(48,10)--(Sv1);
\draw[bundle] (g0r)--(Sv0);
\draw[bundle] (g2r)--(121,10)--(Sw1);
\draw[bundle] (g1r)--(g1dangling);
\node[lb,anchor=west] at (105,29) {dangling};
\draw (Su0)--(Su1)--(Su2)--(Su0);
\node[tv] at (Su0) {};
\node[lb,above left] at (Su0) {$S^u_0$};
\node[tv] at (Su1) {};
\node[lb,below right] at (Su1) {$S^u_1$};
\node[tv] at (Su2) {};
\node[lb,below left] at (Su2) {$S^u_2$};
\draw (Sv0)--(Sv1)--(Sv2)--(Sv0);
\node[tv] at (Sv0) {};
\node[lb,above left] at (Sv0) {$S^v_0$};
\node[tv] at (Sv1) {};
\node[lb,below right] at (Sv1) {$S^v_1$};
\node[tv] at (Sv2) {};
\node[lb,below left] at (Sv2) {$S^v_2$};
\draw (Sw0)--(Sw1)--(Sw2)--(Sw0);
\node[tv] at (Sw0) {};
\node[lb,above left] at (Sw0) {$S^w_0$};
\node[tv] at (Sw1) {};
\node[lb,below right] at (Sw1) {$S^w_1$};
\node[tv] at (Sw2) {};
\node[lb,below left] at (Sw2) {$S^w_2$};
\draw[bundle] (Su1)--(cropSu1);
\draw[bundle] (Su2)--(cropSu2);
\draw[bundle] (Sv2)--(cropSv2);
\draw[bundle] (Sw0)--(cropSw0);
\draw[bundle] (Sw2)--(cropSw2);
\node[lb,anchor=north] at (5,2) {$u{:}1$};
\node[lb,anchor=north] at (27,2) {$u{:}2$};
\node[lb,anchor=north] at (86,2) {$v{:}2$};
\node[lb,anchor=south east] at (147,33) {$w{:}0$};
\node[lb,anchor=north] at (147,2) {$w{:}2$};
\draw[bundle] (38,1)--(47,1);
\node[lb,anchor=west] at (49,1) {$k-2$ residual edges};
\end{tikzpicture}

\end{adjustbox}
\end{minipage}
\end{adjustbox}
\caption{Reduction for fixed $k\geq3$.
Boxes in (a) list graphs and their two endpoint states: orientations in the first two, edge-colorings in the next three, and vertex-colorings in the last.
The coloring $\tau$ is fixed auxiliary data; $G_X^k=B(D(F_X^k))$ and $H=L(G_X^k)$.
Arrows preserve reachability in both directions.
In (b), colors encode source weights; slots $0,1,2$ name the three incident positions at each vertex of the prepared source $X$.
Dotted lines show correspondence, and $q_f,q_g$ count the negators in the two wires.
The marked pairs are protected in the whole source.
Network boxes are gates or negators $W_k$.
The triangle vertex $S^z_j$ belongs to source vertex $z$ and slot $j$.
Paired lines show ordered connections, first above second horizontally.
Thick lines collect $k-2$ residual edges, outside the paired connections.
Labels like $u{:}1$ mark corresponding parts of one omitted wire, not identifications.
The middle bundle on $g$ remains dangling; gate residual edges are omitted.}
\label{fig:reduction-overview}
\end{figure}
 \clearpage

In \cref{sec:three}, we prove the other direction by choosing colorings that simulate the legal source moves.
For $k=3$ colors, we choose one particular coloring for each legal source configuration.
The fixed labels from $\tau$ make these choices agree along each wire.
When a legal move reverses the edge $f=uv$ in \cref{fig:reduction-overview}(b), we join paths inside the gates at $u$ and $v$ through its wire.
The joined paths form a whole component using two colors.
Exchanging the two colors on this component gives exactly the chosen coloring for the new configuration.
Thus, one source reversal uses one Kempe change at the cubic network stage.
For $k>3$ colors, we choose simulating sequences that keep the added colors fixed.
Deleting the edges with these colors leaves the three-color construction with short replacements along some edges.
The paths extend through those replacements, so legal reversals are still realized by sequences of Kempe changes.
Applying these choices to $x_s,x_t$ gives the pair $c_s,c_t$ and completes Step~2.

The network now has the required reachability behavior, but it still has free ends and may have parallel edges.
Its internal vertices already have degree $k$; each triangle vertex has two triangle edges and $k-2$ edges in its attached bundle.
The connected pieces follow the connections of $X$, so the network is connected.
In \cref{sec:closing}, we complete the graph while preserving reachability.
Step~3 (\cref{sec:double}) joins corresponding free ends of two copies, producing the multigraph $D(F_X^k)$.
Each internal vertex keeps its $k$ incident edges, and the joined edges connect the two copies.
The result is therefore connected and $k$-regular.
The new proper $k$-edge-colorings $d_s,d_t$ repeat $c_s,c_t$ on both copies; each joined edge keeps the color of its old free end.

Step~4 (\cref{sec:blocks}) replaces each vertex of this multigraph with a block $K_{k,k-1}$, producing $G_X^k=B(D(F_X^k))$.
Distinct incident edges receive distinct attachment vertices, removing parallel edges.
Each attachment vertex has $k-1$ edges inside its block and one outside; every other block vertex has $k$ internal edges.
The blocks are connected, so $G_X^k$ is connected, simple and $k$-regular.
The proper $k$-edge-colorings $a_s,a_t$ retain the colors of $d_s,d_t$ on the edges between blocks and extend them inside each block by a fixed rule.
Sequences extend through both steps.
Conversely, reading one copy or the edges between blocks turns an arbitrary Kempe sequence into a sequence in the preceding graph.
Reachability between the endpoint pairs is therefore preserved in both directions, although a single move may become several.

Step~5 (\cref{sec:main-proof}) gives the output graph $H=L(G_X^k)$.
Each edge of $G_X^k$ becomes a vertex of $H$; reading $a_s,a_t$ on these vertices gives the proper $k$-vertex-colorings $\alpha,\beta$.
An edge has $k-1$ other edges at each end, with no overlap because $G_X^k$ is simple, giving degree $2k-2$ in $H$.
The neighbors of each vertex of $H$ fall into two groups, each pairwise adjacent, so no three neighbors are pairwise nonadjacent and $H$ is claw-free.
Connectedness also passes to the line graph~(\cref{lem:line-graphs}).
Edge-Kempe changes now correspond exactly to vertex-Kempe changes~(\cref{lem:transfer}).
Together with membership in $\PSPACE$, these steps prove \cref{thm:main}; in \cref{sec:main-proof}, we also bound the order of the resulting graph.
We conclude in \cref{sec:conclusion} with the questions left by the construction.

\section{Preliminaries}\label{sec:prelim}

The introduction outlined a reduction through edge-colorings; we now set up the objects and moves needed to carry it out.
We recall graph colorings and Kempe changes in \cref{sec:prelim-graphs}, then describe graphs with free edge ends in \cref{sec:prelim-multipoles}.
We prepare the constraint-logic source in \cref{sec:ncl,sec:source-problem,sec:tait-source}, handling protection and the Tait coloring that coordinates the later gadgets.

\subsection{Graphs, Colorings, and Kempe Changes}\label{sec:prelim-graphs}

The reduction will move between edge-colorings and vertex-colorings of a line graph.
We first make this correspondence precise and record the graph properties it preserves.
For standard graph notation and concepts not defined here, we follow Diestel~\cite{Diestel2025}. 

All graphs considered here are finite.
We use $V(G)$ and $E(G)$ to denote the sets of vertices and edges of a (simple, undirected) graph $G$, respectively.
A \emph{multigraph} may have several edges with the same two ends, which are called \emph{parallel}, but it has no loops.
This loopless convention is narrower than Diestel's convention for multigraphs~\cite{Diestel2025}. A graph, also called a \emph{simple graph}, is a loopless multigraph without parallel edges.
The \emph{degree} of a vertex of a multigraph is the number of edges at the vertex.
For an integer $d\geq 0$, a multigraph is \emph{$d$-regular} if every vertex has degree $d$, and it is \emph{cubic} if it is $3$-regular.
A graph is claw-free if it has no induced subgraph isomorphic to the star $K_{1,3}$.
The line graph $L(G)$ of a graph $G$ has the vertex set $E(G)$, and two distinct edges of $G$ are adjacent in $L(G)$ exactly when they have a common end.

Let $k\geq 1$ be a fixed positive integer.
A \emph{proper $k$-coloring} of a graph $G$ is a map $c\colon V(G)\to\{0,\dots,k-1\}$ such that $c(u)\neq c(v)$ for every edge $uv\in E(G)$.
A \emph{proper $k$-edge-coloring} of a multigraph $G$ is a map $\varphi\colon E(G)\to\{0,\dots,k-1\}$ that gives distinct colors to any two edges with a common end.
For a graph $G$, a proper $k$-edge-coloring is equivalently a proper $k$-coloring of $L(G)$.
The chromatic index $\chi'(G)$ of a multigraph $G$ is the least integer $\ell$ such that $G$ has a proper $\ell$-edge-coloring.
A proper $3$-edge-coloring of a cubic multigraph is a Tait coloring; it uses the colors $0$, $1$, and $2$.
For distinct colors $a,b\in\{0,1,2\}$, the third color is denoted by $\overline{ab}$.

A \emph{color pair}~\cite{belcastroHaas2017} is a two-element set of colors. Let $c$ be a proper $k$-coloring of $G$, and let $\{a,b\}$ be a color pair.
A \emph{Kempe chain} of $c$ for the pair $\{a,b\}$ is the vertex set of a component of the subgraph of $G$ induced by the vertices $v$ with $c(v)\in\{a,b\}$.
A Kempe change replaces $c$ by the coloring obtained from $c$ by exchanging $a$ and $b$ on one Kempe chain for $\{a,b\}$, and the result is again a proper $k$-coloring.
We write $c\sim_G c'$ if $c'$ is obtained from $c$ by a finite, possibly empty, sequence of Kempe changes.
For a proper $k$-edge-coloring $\varphi$ of $G$ and distinct colors $a$ and $b$, an edge-Kempe change exchanges $a$ and $b$ on the edge set of one component of the subgraph of $G$ formed by the edges with color $a$ or $b$.
The edge set of such a component is also called a Kempe chain of $\varphi$ for $\{a,b\}$; by \cref{lem:transfer}, it is a Kempe chain of $\varphi$ read as a coloring of $L(G)$.
For proper $k$-edge-colorings $\varphi$ and $\varphi'$ of $G$, we also write $\varphi\sim_G\varphi'$ if $\varphi'$ is obtained from $\varphi$ by a finite, possibly empty, sequence of edge-Kempe changes; the kind of the colorings determines which relation is meant.
Proper $k$-colorings $c$ and $c'$ with $c\sim_G c'$ are \emph{Kempe equivalent}, and proper $k$-edge-colorings $\varphi$ and $\varphi'$ with $\varphi\sim_G\varphi'$ are \emph{edge-Kempe equivalent}.

We now state the problem and the upper bound that we use.

\begin{definition}[\KR{}]\label{def:problem}
For a fixed $k\geq 1$, an instance of \KR{} is a triple $(G,\alpha_G,\beta_G)$ where $G$ is a graph and $\alpha_G,\beta_G$ are proper $k$-colorings of $G$.
It is a \emph{yes-instance} if $\alpha_G\sim_G\beta_G$.
\end{definition}

The following upper bound is standard, and it is the only one we need.

\begin{proposition}[{Bonamy et al.~\cite{Bonamy2020}}]\label{prop:membership} For every fixed $k\geq 1$, \KR{} is in $\PSPACE$.
\end{proposition}

The next lemma records the standard correspondence between edge-Kempe changes and vertex-Kempe changes on the line graph~\cite{McDonald2012}. 

\begin{lemma}\label{lem:transfer}
Let $G$ be a graph, let $\varphi$ be a proper $k$-edge-coloring of $G$, regarded as a proper $k$-coloring of $L(G)$, and let $a$ and $b$ be distinct colors.
Then two edges of $G$ colored $a$ or $b$ lie in the same component of the subgraph of $G$ formed by the edges colored $a$ or $b$ if and only if the corresponding vertices of $L(G)$ lie in the same Kempe chain of $\varphi$ for the pair $\{a,b\}$.
Consequently, one edge-Kempe change on $G$ is one Kempe change on $L(G)$, and conversely, and two proper $k$-edge-colorings of $G$ are connected by edge-Kempe changes on $G$ if and only if, read as proper $k$-colorings of $L(G)$, they are connected by Kempe changes on $L(G)$.
\end{lemma}

The graphs of our reduction are line graphs, so we also recall three elementary consequences of the line-graph definition~\cite{Diestel2025}. 

\begin{lemma}\label{lem:line-graphs}
Let $G$ be a graph and let $k\geq 1$ be an integer.
\begin{enumerate}[(a)]
\item The line graph $L(G)$ is claw-free.
\item If $G$ is $k$-regular, then $L(G)$ is $(2k-2)$-regular.
\item If $G$ is connected and has at least one edge, then $L(G)$ is connected.
\end{enumerate}
\end{lemma}

\subsection{Multipoles}\label{sec:prelim-multipoles}

The line-graph correspondence lets us work with edge-colorings.
We now describe the small graphs with free ends that will be joined to form the reduction, together with the color restrictions at those ends.

A \emph{multipole of degree $k$} is a finite multigraph $F$ whose vertex set is partitioned into \emph{internal vertices} of degree $k$ and \emph{terminal vertices} of degree one, such that the neighbor of every terminal vertex is internal.
An edge at a terminal vertex is a \emph{dangling edge}, and the other edges of $F$ are \emph{internal edges}.
For $k=3$, these are the cubic multipoles of M{\'a}{\v{c}}ajov{\'a} and {\v{S}}koviera~\cite{MacajovaSkoviera2006} without loops and isolated edges, except that there the free end of a dangling edge is incident with no vertex; the terminal vertices take the place of these free ends. A $k$-regular multigraph is a multipole of degree $k$ without terminal vertices.
As in Leven and Galil~\cite{LevenGalil1983}, a \emph{multiple edge of multiplicity $i$} is a set of $i$ parallel internal edges, or a set of $i$ dangling edges at one internal vertex. A \emph{coloring} of $F$ is a proper $k$-edge-coloring of $F$.
Since a terminal vertex has only one edge, the colorings of $F$ are the maps from $E(F)$ to $\{0,\dots,k-1\}$ that give distinct colors to the $k$ edges at every internal vertex, and every internal vertex is then incident with exactly one edge of each color.
A multiple edge of multiplicity $1$ is a \emph{single edge}.
Let $c$ be a coloring of $F$, and let $P$ be a color pair.
For the multigraph $F$, we use the definitions of Kempe chains and edge-Kempe changes of \cref{sec:prelim-graphs}, without their line-graph interpretation, so a Kempe chain of $c$ for $P$ is the edge set of a component of the subgraph of $F$ formed by the edges colored in $P$.
Every internal vertex is incident with exactly two edges colored in $P$, while a terminal vertex is incident with at most one, so every Kempe chain of $c$ for $P$ is the edge set of a cycle through internal vertices or of a path whose two end edges are dangling edges, and we identify the Kempe chain with this cycle or path.
The result of an edge-Kempe change of $F$ with the color pair $P$ is again a coloring of $F$, because at every internal vertex on the Kempe chain its two edges colored in $P$ exchange their colors; we also call such a change a \emph{move} of $F$.
For colorings $c$ and $c'$ of $F$, we write $c\sim_F c'$ if $c'$ is obtained from $c$ by a finite, possibly empty, sequence of moves of $F$.

The Parity Lemma gives a common parity for the numbers of free edges of each color; Leven and Galil~\cite{LevenGalil1983} state it for graphs whose vertex degrees are $k$ or $1$ and attribute it to Izbicki. Izbicki proved this parity condition and a converse for finite connected loopless multigraphs whose vertex degrees are $k$ or $1$ and whose chromatic index is $k$, for $k=3$ in~\cite{Izbicki1962} and for every $k\geq3$ in~\cite{Izbicki1963}. The converse realizes any prescribed nonnegative counts with a common parity as the numbers of edges of each color incident with degree-$1$ vertices in some such multigraph.
We use the following form, which also identifies that parity with the number of internal vertices; the short counting argument below makes this extra conclusion explicit.

\begin{lemma}[Parity Lemma]\label{lem:parity}
Let $F$ be a multipole of degree $k$ with $n$ internal vertices, and let $c$ be a coloring of $F$.
Then, for every color $h$, the number of dangling edges of color $h$ has the same parity as $n$.
\end{lemma}

\begin{proof}
We count the edges of color $h$ at the internal vertices.
Every internal vertex has $k$ edges of $k$ distinct colors, so exactly one of its edges has the color $h$.
Let $m_h$ and $d_h$ be the numbers of internal edges and of dangling edges of color $h$.
Every internal edge has two distinct internal ends, and every dangling edge has one internal end.
Counting the pairs of an internal vertex and an edge of color $h$ at it therefore gives $n=2m_h+d_h$, so $d_h$ has the same parity as $n$.
\end{proof}

As in~\cite{MacajovaSkoviera2006}, some dangling edges of a multipole are grouped into connectors. A \emph{connector} of a multipole is an ordered pair of two distinct dangling edges, called its \emph{first edge} and its \emph{second edge}, and distinct connectors of one multipole have no common edge.
The \emph{state} of a connector $(f_1,f_2)$ under $c$ is the ordered pair $(c(f_1),c(f_2))$; the connector is \emph{equal} if $c(f_1)=c(f_2)$ and \emph{unequal} otherwise, and this property is its \emph{equality type}.
If a multipole has connectors, then a dangling edge that lies in no connector is a \emph{residual edge}, and if it has exactly two connectors, then we call one of them its \emph{input connector} and the other its \emph{output connector}.

Two dangling edges, at internal vertices $g$ and $g'$, are \emph{identified} by deleting their two terminal vertices and adding the edge $gg'$, as in the junction of~\cite{MacajovaSkoviera2006}, and we regard the new edge as each of them. The two dangling edges may belong to two multipoles, in which case the identification joins the multipoles, or to one multipole, in which case we require $g\neq g'$ so that no loop arises.
All identifications in this paper satisfy this condition.
Two ordered pairs of dangling edges, such as two connectors, are identified coordinate by coordinate, and we regard the resulting ordered pair of edges as each of them.
We also attach dangling edges to new vertices: a dangling edge is \emph{attached} to a new vertex $g$ if its terminal vertex is replaced by $g$.
In our constructions every new vertex ends with degree $k$, so the new vertices, the edges among them and one dangling edge for every edge at a new vertex that is not among them form a multipole of degree $k$, and attaching is the identification of dangling edges with dangling edges of this multipole.
Let a multipole $F$ be obtained from multipoles $F_1,\dots,F_m$ by identifications.
A multipole formed by some of the multipoles $F_i$ with the identifications of $F$ among them is a \emph{part} of $F$, and every edge of a part is an edge of $F$, since an identified dangling edge is regarded as the new edge.
For a coloring $c$ of $F$ and a part $F'$ of $F$, the \emph{restriction} $c|_{F'}$ gives every edge of $F'$ its color under $c$, and it is a coloring of $F'$, because every internal vertex of $F'$ has the same $k$ edges in $F'$ and in $F$.

\subsection{Protected Constraint Logic}\label{sec:ncl}

Multipoles describe the graphs that will carry our colorings; we now prepare the constraint-logic problem that those colorings must represent.
We use the loopless protected replacement from the book of Hearn and Demaine~\cite{HearnDemaine2009}, and then supply a Tait coloring in \cref{sec:tait-source}. Protection will hold for every legal source orientation, including those later recovered from arbitrary target colorings.

An \emph{$\mathrm{AND}/\mathrm{OR}$ constraint graph} has edges of weight $1$ or $2$ and two kinds of vertices~\cite{HearnDemaine2005}. An $\mathrm{AND}$ vertex has one incidence of weight $2$ and two of weight $1$, and an $\mathrm{OR}$ vertex has three incidences of weight $2$.
We call weight-$1$ edges \emph{red} and weight-$2$ edges \emph{blue}, following Hearn and Demaine~\cite{HearnDemaine2005}. All constraint graphs used below are loopless, and parallel edges remain distinct even when their ends agree.
The type of a vertex $z$ is denoted by $\theta(z)$.
A \emph{configuration} is an orientation of the edges, and it is \emph{legal} if the incoming weight at every vertex is at least $2$.
At an $\mathrm{AND}$ vertex, the local rule permits the blue edge to point outward exactly when both red edges point inward; it does not force the blue edge outward.
At an $\mathrm{OR}$ vertex, the local rule permits an edge to point outward exactly when at least one of the other two edges points inward. A \emph{legal reversal} reverses one edge and leaves the configuration legal.
For legal configurations $x,y$ of a constraint graph $X$, we write $x\sim_X y$ if a finite, possibly empty, sequence of legal reversals joins them.
The \emph{configuration-to-configuration problem} asks whether $x_s\sim_X x_t$ for two given legal configurations.
A pair of incidences at an $\mathrm{OR}$ vertex is \emph{protected} if no legal configuration of the whole graph directs both incidences toward that vertex.

We start from the configuration-to-configuration hardness result of Hearn and Demaine~\cite{HearnDemaine2005}. The later source constructions will preserve its endpoint reachability.

\begin{proposition}[{Hearn and Demaine~\cite{HearnDemaine2005}}]\label{prop:source} The configuration-to-configuration problem is $\PSPACE$-complete on connected simple loopless cubic $\mathrm{AND}/\mathrm{OR}$ constraint graphs.
\end{proposition}

Hearn and Demaine~\cite{HearnDemaine2005} give a polynomial-time transformation to an equivalent simple planar $3$-connected constraint graph, including corresponding endpoint configurations. We then replace each $\mathrm{OR}$ vertex by the protected-$\mathrm{OR}$ construction from their book~\cite{HearnDemaine2009}, shown in full in \cref{fig:hd-source}. The three exterior blue edges take the place of the old incident edges, and a legal orientation extends through the replacement exactly when at least one of them points inward.
The construction first uses a pair of degree-$2$ red--blue vertices, each requiring incoming weight at least $1$.
We replace this pair by the six-vertex graph in \cref{fig:hd-source}(a), using the conversion of Hearn and Demaine~\cite{HearnDemaine2009}. 

For each old endpoint configuration, we first choose a legal orientation of the intermediate replacement with its two red--blue vertices and the same exterior directions.
We then expand the pair, using the fixed directions shown in \cref{fig:hd-source}(a) for the added edges in both endpoints.
These directions supply incoming weight $1$ at each attachment vertex, so the same fixed pattern works with every legal orientation of the intermediate replacement~\cite{HearnDemaine2009}. The endpoint choices preserve and reflect reachability by \cref{lem:protected-endpoint-completion}.

\begin{figure}[ht]
\centering
\begin{adjustbox}{max width=\textwidth}
\begin{tikzpicture}[
 hdand/.style={rectangle,draw,fill=white,minimum width=7.5mm,minimum height=7.5mm,inner sep=0pt,font=\small},
 hdor/.style={circle,draw,fill=white,minimum size=7.5mm,inner sep=0pt,font=\small},
 hdred/.style={draw=red!70!black,line width=0.65pt,dash pattern=on 3pt off 2pt},
 hdblue/.style={draw=blue!70!black,line width=1.5pt},
 hdarrow/.style={->,>=stealth},
 hdl/.style={font=\small,inner sep=2.5pt,fill=white},
 hdnote/.style={font=\fontsize{9}{10.5}\selectfont,inner sep=1.5pt,align=center},
 hdgroup/.style={draw=black!55,densely dotted,line width=0.55pt,rounded corners=2pt},
 hdtitle/.style={font=\small,inner sep=1pt}]
\path[use as bounding box] (-2.6,-5.3) rectangle (12.4,3.7);
\node[hdand] (detail-p) at (-2,.5) {$p$};
\node[hdand] (detail-q) at (2,.5) {$q$};
\node[hdand] (detail-l) at (-.8,.5) {$\ell$};
\node[hdand] (detail-r) at (.8,.5) {$r$};
\node[hdor] (detail-t) at (0,1.65) {$t$};
\node[hdand] (detail-b) at (0,-.65) {$b$};
\coordinate (pblue) at (-2,2);
\coordinate (qblue) at (2,2);
\coordinate (pred) at (-2,-1);
\coordinate (qred) at (2,-1);
\draw[hdblue] (detail-p) -- node[hdl,left] {$2$} (pblue);
\draw[hdblue] (detail-q) -- node[hdl,right] {$2$} (qblue);
\draw[hdred] (detail-p) -- node[hdl,left] {$1$} (pred);
\draw[hdred] (detail-q) -- node[hdl,right] {$1$} (qred);
\draw[hdblue,hdarrow] (detail-t) -- (detail-l);
\draw[hdblue,hdarrow] (detail-t) -- (detail-r);
\draw[hdblue,hdarrow] (detail-b) -- (detail-t);
\draw[hdred,hdarrow] (detail-l) -- (detail-b);
\draw[hdred,hdarrow] (detail-r) -- (detail-b);
\draw[hdred,hdarrow] (detail-l) -- (detail-p);
\draw[hdred,hdarrow] (detail-r) -- (detail-q);
\node[hdtitle] at (0,3.1) {(a) Paired Red--Blue Converter};
\node[hdnote] at (0,-1.65) {Chosen Internal Orientation};
\node[hdnote] at (0,-2.4) {The same pattern is used\\at both designated endpoints.};
\begin{scope}[xshift=6.1cm]
\node[hdand] (tL) at (0,2.1) {$t_L$};
\node[hdand] (tM) at (1.8,2.1) {$t_M$};
\node[hdand] (tR) at (3.6,2.1) {$t_R$};
\node[hdand] (bL) at (0,0) {$b_L$};
\node[hdand] (bR) at (3.6,0) {$b_R$};
\node[hdor] (uL) at (.9,1.05) {$u_L$};
\node[hdor] (uR) at (2.7,1.05) {$u_R$};
\node[hdand] (aL) at (-1.45,0) {$a_L$};
\node[hdand] (aR) at (5.05,0) {$a_R$};
\node[hdor] (t) at (1.8,-1.6) {$t$};
\node[hdand] (l) at (.3,-2.75) {$\ell$};
\node[hdand] (r) at (3.3,-2.75) {$r$};
\node[hdand] (b) at (1.8,-3.9) {$b$};
\coordinate (extA) at (-2.5,0);
\coordinate (extB) at (6.1,0);
\coordinate (extC) at (1.8,2.9);
\draw[hdred] (tL) -- (tM);
\draw[hdred] (tM) -- node[hdl,above] {$H$} (tR);
\draw[hdred] (tL) -- (bL);
\draw[hdred] (tR) -- (bR);
\draw[hdblue] (tL) -- node[hdl,above right] {$I$} (uL);
\draw[hdblue] (bL) -- (uL);
\node[hdl] at (.83,.3) {$D$};
\draw[hdblue] (uL) -- node[hdl,above] {$E$} (uR);
\draw[hdblue] (tR) -- node[hdl,above left] {$G$} (uR);
\draw[hdblue] (bR) -- (uR);
\node[hdl] at (2.77,.3) {$F$};
\draw[hdred] (aL) -- (bL);
\draw[hdred] (aR) -- (bR);
\draw[hdblue] (t) -- (l);
\draw[hdblue] (t) -- (r);
\draw[hdblue] (b) -- (t);
\draw[hdred] (l) -- (b);
\draw[hdred] (r) -- (b);
\draw[hdred] (l) -- (aL);
\draw[hdred] (r) -- (aR);
\draw[hdblue] (aL) -- (extA);
\node[hdl] at (-2.32,.35) {$A$};
\draw[hdblue] (aR) -- (extB);
\node[hdl] at (5.92,.35) {$B$};
\draw[hdblue] (tM) -- node[hdl,right] {$C$} (extC);
\draw[hdgroup] (-2,.75) -- (-.95,.75) -- (-.95,-.75)
 -- (4.55,-.75) -- (4.55,.75) -- (5.6,.75)
 -- (5.6,-4.4) -- (-2,-4.4) -- (-2,.75);
\node[hdnote] at (-1.15,1.2) {$u_L:\{I,D\}$};
\node[hdnote] at (4.75,1.2) {$u_R:\{G,F\}$};
\node[hdnote] at (4.35,-3.85) {Protected at $t$:\\$\{tb,t\ell\}$};
\node[hdnote] at (-.75,-3.85) {Paired\\Converter};
\node[hdtitle] at (1.8,3.3) {(b) Expanded Protected-$\mathrm{OR}$ Replacement};
\end{scope}
\draw[hdred] (-1.95,-4.95) -- (-1.25,-4.95);
\node[hdnote,anchor=west] at (-1.15,-4.95) {Weight $1$};
\draw[hdblue] (1.35,-4.95) -- (2.05,-4.95);
\node[hdnote,anchor=west] at (2.15,-4.95) {Weight $2$};
\node[hdand] at (5.15,-4.95) {};
\node[hdnote,anchor=west] at (5.6,-4.95) {$\mathrm{AND}$};
\node[hdor] at (8.75,-4.95) {};
\node[hdnote,anchor=west] at (9.2,-4.95) {$\mathrm{OR}$};
\end{tikzpicture}
\end{adjustbox}
\caption{Source gadgets redrawn from Hearn and Demaine~\cite{HearnDemaine2009}. Squares are $\mathrm{AND}$ vertices and circles are $\mathrm{OR}$ vertices.
Thin dashed red edges have weight $1$, and thick solid blue edges have weight $2$.
Part~(a) shows the six-vertex paired converter and the internal orientation chosen at both designated endpoints; the four exterior ports are left unoriented.
These arrows describe the chosen endpoint pattern, not a restriction on every legal configuration.
Part~(b) has $13=7+6$ vertices, with the paired converter expanded inside the gray dotted outline, which is only an annotation.
The pairs $\{I,D\}$ at $u_L$, $\{G,F\}$ at $u_R$, and $\{tb,t\ell\}$ at $t$ are protected.
The three blue ports $A,B,C$ continue to the rest of the source graph; their unmarked ends are not vertices.}
\label{fig:hd-source}
\end{figure}
 
The fixed converter orientation lets us prepare the interior of each replacement before reversing an old edge, without imposing a restriction on intermediate configurations in the reverse implication.
The following argument makes the endpoint choices explicit for the construction of Hearn and Demaine~\cite{HearnDemaine2009}. In the intermediate replacement, the seven core vertices of \cref{fig:hd-source}(b) retain threshold $2$, while $a_L,a_R$ have only their exterior blue edge and their red edge to $b_L,b_R$, respectively, and have threshold $1$.
The \emph{boundary} records the directions of $A,B,C$, with inward meaning toward the replacement.
Let $X$ be a loopless cubic constraint graph, and let $X^+$ be obtained by expanding every $\mathrm{OR}$ vertex as described above.
For a legal configuration $x$ of $X$, let $x^+$ denote any completion obtained by choosing legal intermediate orientations with the old edge directions of $x$ and adding the fixed converter arrows from \cref{fig:hd-source}(a).

\begin{lemma}\label{lem:protected-endpoint-completion}
For every pair of legal configurations $x,y$ of $X$ and every independent choice of their completions $x^+,y^+$ as above, both completions are legal, and $x\sim_X y$ if and only if $x^+\sim_{X^+}y^+$.
\end{lemma}

\begin{proof}
We first join all legal intermediate orientations with the same boundary, then prepare each old reversal using an interior legal on both sides of that reversal.
An instruction to orient an edge leaves it unchanged when it already has the indicated direction.
The blue edges are $I=t_Lu_L$, $D=b_Lu_L$, $E=u_Lu_R$, $G=t_Ru_R$, and $F=b_Ru_R$.

\begin{itemize}
\item \textbf{Case 1: $A$ is inward.}
First, we orient $a_Lb_L$ toward $b_L$; the blue input supplies $a_L$.
We next arrange that $I$ enters $t_L$.
If $I$ instead enters $u_L$, both red edges enter $t_L$, so $C$ must be inward.
In this situation, if neither $G$ nor $F$ enters $u_R$, both enter their adjacent $\mathrm{AND}$ vertices; we may then orient $t_Mt_R$ and $b_Rt_R$ toward $t_R$, and reverse $G$ toward $u_R$.
These reversals are legal because $C$ supplies $t_M$, $F$ supplies $b_R$, and both prepared red edges supply $t_R$.
Still assuming that $I$ initially enters $u_L$, at least one of $G$ and $F$ now enters $u_R$, either from the start or after this preparation.
We therefore orient $E$ toward $u_L$ and then $I$ toward $t_L$.
Now we orient, in order, $t_Lb_L$ toward $b_L$, $D$ toward $u_L$, $E$ toward $u_R$, $G$ toward $t_R$, and $F$ toward $b_R$.
The reversal of $D$ is legal because both red edges now enter $b_L$, and each subsequent reversal leaves the relevant $\mathrm{OR}$ vertex supplied.
Finally, we orient the remaining red edges as $t_L\to t_M$, $t_R\to t_M$, $b_R\to t_R$, and $b_R\to a_R$.
Their tails have incoming blue edges, and their heads gain incoming weight.
All eleven interior directions are now fixed.

\item \textbf{Case 2: $A$ is outward and $B$ is inward.}
The reflection exchanging left and right gives the same argument, with $I,D$ exchanged with $G,F$ and the direction of $E$ reversed.

\item \textbf{Case 3: only $C$ is inward.}
The red attachment edges must point toward $a_L,a_R$, forcing $D,F$ toward $b_L,b_R$.
Since $C$ supplies $t_M$ and $D$ supplies $b_L$, we can orient both red edges at $t_L$ toward $t_L$, then orient $I$ toward $u_L$.
We next orient $E$ toward $u_R$ and $G$ toward $t_R$, and finish with $t_R\to t_M$ and $b_R\to t_R$.
Each reversal is legal by the indicated incoming blue edge or the two prepared red edges.
Again, every interior direction is fixed.
\end{itemize}

Each case reaches one fixed orientation, which is legal for every boundary in that case.
Reversing these sequences proves connectivity for each fixed legal boundary.
The converter arrows supply weight $1$ at both attachments and satisfy all four added vertices, so completion preserves legality and every intermediate reversal while keeping those arrows fixed.

\begin{itemize}
\item[($\Rightarrow$)]
For one old reversal, we regard its changing port as outward at each incident replacement and keep the other ports fixed.
This boundary still has an inward port, since the old configurations on both sides are legal and looplessness allows only one changing port in a replacement.
An interior legal for this boundary is also legal when that port is inward.
Fixed-boundary connectivity prepares this common interior independently at both ends, the old edge reverses once, and connectivity then reaches any chosen interior for the new boundary.
The fixed completion lifts every step.
Fixed-boundary connectivity also adjusts untouched replacements, including when no old edge reverses.
Concatenation therefore joins arbitrary independent choices at the two endpoints.

\item[($\Leftarrow$)]
Deleting the four added converter vertices from any legal expanded orientation leaves incoming weight at least $1$ at each attachment.
If all three ports were outward, the attachment edges would force $D,F$ away from $u_L,u_R$, and the two red inputs at $t_M$ would force $I,G$ away from them as well.
Then $E$ would have to enter both $u_L$ and $u_R$, which is impossible.
Thus, retaining the old edge directions gives a legal configuration of $X$.
An internal reversal projects to no change, and a reversal of an old edge projects to one legal reversal.
The projected sequence joins $x$ to $y$.
\end{itemize}
\end{proof}

The expanded replacement has seven core vertices and six vertices from the paired converter, for a total of thirteen.
Every edge joins distinct vertices, including those inside the converter, so the replacement is loopless.
The old graph is loopless and different old vertices receive disjoint replacements, so joining their exterior edges creates no loop either.
The connected replacement preserves connectedness.
All three $\mathrm{OR}$ vertices in the replacement have protected pairs, as indicated in \cref{fig:hd-source} and established by Hearn and Demaine~\cite{HearnDemaine2009}. Protection holds in every legal configuration of the whole graph, not only in the chosen endpoint configurations.
Thus, the source is already loopless and protected after this replacement.

If $n_0$ is the order of the simple source before protection, the replacement gives at most $13n_0$ vertices.
Indeed, each replaced $\mathrm{OR}$ vertex contributes $7+6=13$ vertices, while an old $\mathrm{AND}$ vertex remains one vertex.
This count does not include the earlier polynomial transformation to a simple $3$-connected source~\cite{HearnDemaine2005}.

\paragraph{Slots and Local States.}
The source is now loopless and protected, so we can give its three incidences fixed local names.
At an $\mathrm{AND}$ vertex, the incident edges are numbered by \emph{slots} $0,1,2$, with its blue edge in slot $0$.
At an $\mathrm{OR}$ vertex, a designated protected pair occupies slots $0$ and $2$.
The slot of an edge $f$ at its end $z$ is denoted by $s_z(f)$, and the \emph{local state} $x|_z$ is the word $\beta_0\beta_1\beta_2$ whose bit $\beta_j$ is $1$ exactly when the edge at slot $j$ points toward $z$.
The legal $\mathrm{AND}$ states are $011,100,101,110,111$.
Ordinary $\mathrm{OR}$ legality excludes only $000$, but protection further confines states occurring in the whole graph to
\[
001,\quad011,\quad010,\quad110,\quad100.
\]
The pairs of these five states that differ in one bit form the path in the displayed order.
Protection is used only after an orientation has been shown ordinarily legal.

\subsection{The Source Used by the Reduction}\label{sec:source-problem}

We have obtained a protected loopless source, but the two ends of each future wire still need a common color label.
The next proposition supplies a Tait coloring without changing reachability or losing protection.
We state the resulting source problem here and give the construction and proof in \cref{sec:tait-source}.

\begin{proposition}\label{prop:tait-source}
For every nonempty connected loopless cubic constraint multigraph $X_0$ with a protected pair at every $\mathrm{OR}$ vertex, there are such a multigraph $X$ with at most $10|V(X_0)|$ vertices, a Tait coloring $\tau$ of $X$, and a map $x\mapsto x'$ on legal configurations such that $x\sim_{X_0}y$ if and only if $x'\sim_X y'$.
The graph, the Tait coloring, the protected pairs and the endpoint map are computable in linear time.
\end{proposition}

Combining this proposition with the protected replacement of Hearn and Demaine~\cite{HearnDemaine2009} gives the source used throughout the remainder of the paper. 

\begin{corollary}\label{cor:tait-source}
The configuration-to-configuration problem is $\PSPACE$-complete on nonempty connected loopless cubic $\mathrm{AND}/\mathrm{OR}$ multigraphs supplied with a Tait coloring and a protected pair at every $\mathrm{OR}$ vertex.
\end{corollary}

\begin{proof}
Hardness follows from \cref{prop:source} by composing a completion map of \cref{lem:protected-endpoint-completion} with the map of \cref{prop:tait-source}.
There are at most $130n_0$ vertices after protection and Tait preprocessing, where $n_0$ is the order of the simple source before protection.
Membership follows by polynomial-space search through orientations; the supplied coloring and protected-pair designations do not alter legality or the move rule.
\end{proof}

From now on, a \emph{source graph} means a multigraph satisfying \cref{cor:tait-source}, with its slots fixed as above.
Parallel source edges remain distinct throughout the construction.
Each edge has an arbitrarily fixed first end and second end, which are distinct by looplessness.
For a vertex $z$, the permutation $\tau_z$ of $\{0,1,2\}$ sends slot $j$ to the Tait color of its edge.
 \subsection{Constructing a Tait-Colored Source}\label{sec:tait-source}

\Cref{sec:source-problem} stated the source we need; we now construct its Tait coloring while preserving legal reversals and protection.
The idea is to color the three incidences at each old vertex separately, then replace edges whose ends received different colors by short paths.
Small attached graphs force these paths to carry enough incoming weight to recover the old orientation.

We first give the three edges at each vertex different colors, allowing an edge to receive different colors at its ends.
Each edge with different end colors is stretched into a three-edge path through two new $\mathrm{OR}$ vertices.
Its middle edge takes the third color, and each new vertex gets one further edge to make its degree three.
Pairs of these extra edges attach to the core of Hearn and Demaine's constrained edge terminator~\cite{HearnDemaine2009}. This small graph forces both attaching edges inward, so each new vertex must receive inflow from the path.
We can therefore recover a legal configuration of the old graph from every legal configuration of the new graph.
A legal reversal of the old graph becomes one or three legal reversals of the new graph.
Conversely, a legal reversal of the new graph changes at most one edge of the recovered legal configuration.

\paragraph{The Core $\mathcal{Q}$.}
The extra path vertices need attachments that force their added edges outward.
The following four-vertex graph provides two such attachments at once.
Let $\mathcal{Q}$ be the constraint graph with four $\mathrm{AND}$ vertices $v_{\mathrm{L}}$, $v_{\mathrm{T}}$, $v_{\mathrm{R}}$ and $v_{\mathrm{B}}$, the edges $v_{\mathrm{L}}v_{\mathrm{T}}$, $v_{\mathrm{T}}v_{\mathrm{R}}$, $v_{\mathrm{R}}v_{\mathrm{B}}$ and $v_{\mathrm{B}}v_{\mathrm{L}}$ of weight $1$, the edge $v_{\mathrm{T}}v_{\mathrm{B}}$ of weight $2$, and two further edges of weight $2$, the \emph{ports}, one at $v_{\mathrm{L}}$ and one at $v_{\mathrm{R}}$, whose other ends lie outside $\mathcal{Q}$.
The indices stand for left, top, right and bottom, as in \cref{fig:q}.
Thus, the edges of $\mathcal{Q}$ other than the ports form the complete graph on its four vertices without the edge $v_{\mathrm{L}}v_{\mathrm{R}}$.
In a legal configuration of the constrained edge terminator~\cite{HearnDemaine2009}, the two edges that join its $\mathrm{OR}$ vertex to $\mathcal{Q}$ point into $\mathcal{Q}$, and the next lemma proves this fact for every orientation that satisfies the constraints of the four vertices of $\mathcal{Q}$. \usetikzlibrary{arrows.meta}
\begin{figure}[ht]
\centering
\begin{adjustbox}{max width=\textwidth}
\begin{tikzpicture}[
  vx/.style={circle, draw, fill=white, inner sep=0pt, minimum size=6.2mm, font=\footnotesize},
  lb/.style={font=\footnotesize, inner sep=2pt},
  hv/.style={line width=1.6pt},
  pt/.style={line width=1.6pt, -{Stealth[length=2.2mm]}, shorten >=1pt}
]
\usetikzlibrary{arrows.meta}
\node[vx] (vL) at (-1.6,0) {$v_{\mathrm{L}}$};
\node[vx] (vT) at (0,1.3) {$v_{\mathrm{T}}$};
\node[vx] (vR) at (1.6,0) {$v_{\mathrm{R}}$};
\node[vx] (vB) at (0,-1.3) {$v_{\mathrm{B}}$};
\node[vx] (g) at (-3.6,0) {$g$};
\node[vx] (gp) at (3.6,0) {$g'$};
\draw (vL) -- node[lb, above left] {$\zeta'$} (vT);
\draw (vT) -- node[lb, above right] {$\overline{\zeta\zeta'}$} (vR);
\draw (vR) -- node[lb, below right] {$\zeta'$} (vB);
\draw (vB) -- node[lb, below left] {$\overline{\zeta\zeta'}$} (vL);
\draw[hv] (vT) -- node[lb, right] {$\zeta$} (vB);
\draw[pt] (g) -- node[lb, above] {$\zeta$} (vL);
\draw[pt] (gp) -- node[lb, above] {$\zeta$} (vR);
\end{tikzpicture}
\end{adjustbox}
\caption{The drawing shows a copy of the constraint graph $\mathcal{Q}$ attached to two new $\mathrm{OR}$ vertices $g$ and $g'$ whose ports have the color $\zeta$.
The four $\mathrm{AND}$ vertices $v_{\mathrm{L}}$, $v_{\mathrm{T}}$, $v_{\mathrm{R}}$ and $v_{\mathrm{B}}$ are drawn as a diamond without the edge $v_{\mathrm{L}}v_{\mathrm{R}}$; thin lines are the edges of weight $1$ and thick lines the edges of weight $2$.
The labels give the colors of the Tait coloring $\tau$, where $\zeta'\ne\zeta$.
The arrows show the direction of the two ports in every legal configuration, which \cref{lem:q-core} proves.}
\label{fig:q}
\end{figure}
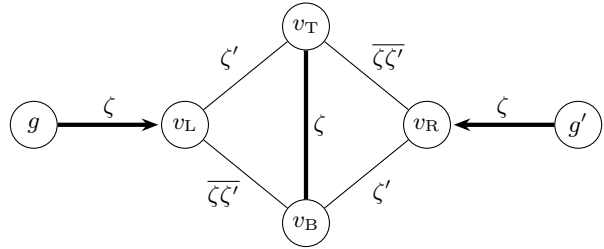
 
\begin{lemma}\label{lem:q-core}
In every orientation of the edges of $\mathcal{Q}$ and of its two ports in which the inflow at each of $v_{\mathrm{L}}$, $v_{\mathrm{T}}$, $v_{\mathrm{R}}$ and $v_{\mathrm{B}}$ is at least $2$, both ports are directed toward $\mathcal{Q}$.
\end{lemma}

\begin{proof}
We show that a port directed away from $\mathcal{Q}$ leaves the edge $v_{\mathrm{T}}v_{\mathrm{B}}$ no legal direction.
Suppose that the port at $v_{\mathrm{L}}$ is directed away from $v_{\mathrm{L}}$.
Then the inflow at $v_{\mathrm{L}}$ comes from the edges $v_{\mathrm{L}}v_{\mathrm{T}}$ and $v_{\mathrm{B}}v_{\mathrm{L}}$ of weight $1$ only, so both are directed toward $v_{\mathrm{L}}$.
The inflow at $v_{\mathrm{T}}$ is then at most the weight $1$ of $v_{\mathrm{T}}v_{\mathrm{R}}$ plus the weight of $v_{\mathrm{T}}v_{\mathrm{B}}$ if this edge is directed toward $v_{\mathrm{T}}$, so $v_{\mathrm{T}}v_{\mathrm{B}}$ is directed toward $v_{\mathrm{T}}$.
In the same way, the inflow at $v_{\mathrm{B}}$ is at most the weight $1$ of $v_{\mathrm{R}}v_{\mathrm{B}}$ plus the weight of $v_{\mathrm{T}}v_{\mathrm{B}}$ if this edge is directed toward $v_{\mathrm{B}}$, so $v_{\mathrm{T}}v_{\mathrm{B}}$ is directed toward $v_{\mathrm{B}}$, a contradiction.
Hence, the port at $v_{\mathrm{L}}$ is directed toward $\mathcal{Q}$.
The permutation that exchanges $v_{\mathrm{L}}$ and $v_{\mathrm{R}}$ and fixes $v_{\mathrm{T}}$ and $v_{\mathrm{B}}$ maps the edges of $\mathcal{Q}$ onto themselves with their weights, so by applying the above argument to $v_{\mathrm{R}}$ instead of $v_{\mathrm{L}}$, the port at $v_{\mathrm{R}}$ is also directed toward $\mathcal{Q}$.
\end{proof}

\paragraph{The Construction.}
We now build the graph $X$ and the Tait coloring $\tau$ of \cref{prop:tait-source}.
Let $X_0$ be a connected loopless cubic $\mathrm{AND}/\mathrm{OR}$ constraint multigraph with protected pairs, with slots at its vertices numbered as in \cref{sec:ncl}.
Every $\mathrm{AND}$ vertex is incident with exactly two edges of weight $1$ and every $\mathrm{OR}$ vertex with none, so each edge of weight $1$ shares an end with at most two other edges of weight $1$.
We color the edges of weight $1$ one at a time, each with a color that differs from the colors of the edges of weight $1$ that share an end with it and are already colored.
In the resulting coloring, the two edges of weight $1$ at every $\mathrm{AND}$ vertex have different colors.
For a vertex $v$ of $X_0$ and an edge $e$ at $v$, the \emph{color of $e$ at $v$}, denoted by $c_v(e)$, is defined as follows, and the notation $c_v(e)$ is local to this subsection.
\begin{itemize}
\item If $e$ has weight $1$, then $c_v(e)$ is the color of $e$.
\item If $v$ is an $\mathrm{AND}$ vertex and $e$ is its edge of weight $2$, then $c_v(e)$ is the color other than the colors of the two edges of weight $1$ at $v$.
\item If $v$ is an $\mathrm{OR}$ vertex, then $c_v(e)$ is the slot of $e$ at $v$.
\end{itemize}
The three colors $c_v(e)$ of the edges $e$ at a vertex $v$ are distinct.
An edge $e=uv$ is \emph{matched} if $c_u(e)=c_v(e)$, and \emph{mismatched} otherwise; every edge of weight $1$ is matched, so every mismatched edge has weight $2$.

The graph $X$ is obtained from $X_0$ in two steps.
First, every mismatched edge $e=uv$ is replaced by the path $u,g^e_u,g^e_v,v$ through two new $\mathrm{OR}$ vertices, with all three path edges of weight $2$.
Each replaced edge uses its own two new vertices, even when old edges are parallel.
Each new vertex receives one further edge of weight $2$, its port.
The port at $g^e_u$ has color $c_v(e)$, and the port at $g^e_v$ has color $c_u(e)$.
Second, for every color $\zeta$, the new vertices whose ports have the color $\zeta$ are partitioned into pairs, which is possible by \cref{lem:free-parity} below, and for every pair $\{g,g'\}$ a new copy of $\mathcal{Q}$ is added whose ports join $v_{\mathrm{L}}$ to $g$ and $v_{\mathrm{R}}$ to $g'$.
The vertices of $X_0$ keep their types and their slots in $X$, where the edge $ug^e_u$ takes the slot of $e$ at $u$ and the edge $g^e_vv$ takes the slot of $e$ at $v$.
At each new vertex $g^e_u$ or $g^e_v$, the port and one chosen path edge occupy slots $0$ and $2$, and at the vertices $v_{\mathrm{L}}$, $v_{\mathrm{T}}$, $v_{\mathrm{R}}$ and $v_{\mathrm{B}}$ of every copy of $\mathcal{Q}$ the incident edges are numbered so that slot $0$ holds the edge of weight $2$.

The Tait coloring $\tau$ of $X$ gives every matched edge $e=uv$ the color $c_u(e)=c_v(e)$.
For a mismatched edge $e=uv$, it gives the edges $ug^e_u$, $g^e_ug^e_v$ and $g^e_vv$ the colors $c_u(e)$, $\overline{c_u(e)c_v(e)}$ and $c_v(e)$, and it gives each port its color.
In a copy of $\mathcal{Q}$ whose ports have the color $\zeta$, it gives $v_{\mathrm{T}}v_{\mathrm{B}}$ the color $\zeta$, the edges $v_{\mathrm{L}}v_{\mathrm{T}}$ and $v_{\mathrm{R}}v_{\mathrm{B}}$ one color $\zeta'\ne\zeta$, and the edges $v_{\mathrm{T}}v_{\mathrm{R}}$ and $v_{\mathrm{B}}v_{\mathrm{L}}$ the color $\overline{\zeta\zeta'}$.

The partition into pairs in the second step needs an even number of ports of each color, which the next lemma provides.

\begin{lemma}\label{lem:free-parity}
For every color $\zeta$, the number of new vertices whose ports have the color $\zeta$ is even.
\end{lemma}

\begin{proof}
We count the pairs $(v,e)$ of a vertex $v$ of $X_0$ and an edge $e$ at $v$ with $c_v(e)=\zeta$.
Every vertex contributes exactly one such pair, so there are $|V(X_0)|$ of them, and $|V(X_0)|$ is even because $X_0$ is cubic.
A matched edge contributes two pairs or none, and a mismatched edge $e=uv$ contributes one pair if $\zeta\in\{c_u(e),c_v(e)\}$ and none otherwise.
Hence, the number of mismatched edges $e=uv$ with $\zeta\in\{c_u(e),c_v(e)\}$ is even.
Each such edge gives exactly one new vertex whose port has the color $\zeta$, namely $g^e_v$ if $c_u(e)=\zeta$ and $g^e_u$ if $c_v(e)=\zeta$, and every new vertex whose port has the color $\zeta$ arises in this way.
\end{proof}

\paragraph{The Maps $\mathrm{up}$ and $\mathrm{down}$.}
The proof of \cref{prop:tait-source} compares the legal configurations of $X_0$ and of $X$ through two maps between their orientations.
For an orientation $x$ of $E(X_0)$, let $\mathrm{up}(x)$ be the orientation of $E(X)$ defined as follows.
It directs every matched edge as $x$ does.
For every mismatched edge $e=uv$, it directs the three edges of the path $u,g^e_u,g^e_v,v$ toward $u$ if $x$ directs $e$ toward $u$, and toward $v$ otherwise.
It directs every port toward its copy of $\mathcal{Q}$, and in every copy of $\mathcal{Q}$ it directs $v_{\mathrm{L}}v_{\mathrm{T}}$, $v_{\mathrm{T}}v_{\mathrm{R}}$ and $v_{\mathrm{T}}v_{\mathrm{B}}$ toward $v_{\mathrm{T}}$ and $v_{\mathrm{R}}v_{\mathrm{B}}$ and $v_{\mathrm{B}}v_{\mathrm{L}}$ toward $v_{\mathrm{B}}$.
Before the next definition, one end of every mismatched edge $e$ is fixed as $u$ and the other end as $v$, in an arbitrary but fixed way.
For an orientation $x'$ of $E(X)$, let $\mathrm{down}(x')$ be the orientation of $E(X_0)$ that directs every matched edge as $x'$ does and every mismatched edge $e=uv$ toward $u$ if $x'$ directs $ug^e_u$ toward $u$, and toward $v$ otherwise.
The maps $\mathrm{up}$ and $\mathrm{down}$ are local to this subsection.

\begin{proof}[Proof of \cref{prop:tait-source}]
We first check the shape of $X$ and the coloring $\tau$, and then we compare the legal configurations through the maps $\mathrm{up}$ and $\mathrm{down}$, with $x'_s=\mathrm{up}(x_s)$ and $x'_t=\mathrm{up}(x_t)$.
\begin{enumerate}[(a)]
\item \textbf{The graph $X$.}
Every vertex of $X_0$ keeps three edges with their weights, because a mismatched edge $e$ at $u$ is replaced by $ug^e_u$ of the same weight $2$.
Every new vertex $g^e_u$ or $g^e_v$ is an $\mathrm{OR}$ vertex with two path edges and one port, all of weight $2$, and $v_{\mathrm{L}}$, $v_{\mathrm{T}}$, $v_{\mathrm{R}}$ and $v_{\mathrm{B}}$ are $\mathrm{AND}$ vertices with one edge of weight $2$ and two of weight $1$.
Every new edge has a new end, the two path edges at a vertex $g^e_u$ or $g^e_v$ lead to different vertices, and each of these vertices has exactly one port.
The two ports of a copy of $\mathcal{Q}$ lead to the two distinct vertices of a pair, and the edges of $\mathcal{Q}$ other than the ports form a simple graph.
Hence, $X$ has no loops.
Matched parallel old edges may remain; every new edge is indexed by its own old edge or copy of $\mathcal{Q}$, so these cause no collision.
Replacing an edge by a path keeps $X_0$ connected, and every copy of $\mathcal{Q}$ is joined to the rest of $X$ by its ports, so $X$ is connected.
\item \textbf{The coloring $\tau$.}
At a vertex $v$ of $X_0$, every edge $e$ of $X_0$ at $v$, or the path edge that replaces it, has the color $c_v(e)$, so the three colors at $v$ are distinct.
For a mismatched edge $e=uv$, the colors at $g^e_u$ are $c_u(e)$, $\overline{c_u(e)c_v(e)}$ and the port color $c_v(e)$, and at $g^e_v$ they are $\overline{c_u(e)c_v(e)}$, $c_v(e)$ and $c_u(e)$.
In a copy of $\mathcal{Q}$ with ports of color $\zeta$, each of its four vertices has one edge of each of the colors $\zeta$, $\zeta'$ and $\overline{\zeta\zeta'}$.
\item \textbf{Size and time.}
Let $r$ be the number of mismatched edges.
The first step adds $2r$ vertices, the second adds $r$ copies of $\mathcal{Q}$ with four vertices each, and $r\le|E(X_0)|=3|V(X_0)|/2$, so $|V(X)|=|V(X_0)|+6r\le 10|V(X_0)|$.
Every step handles each vertex and each edge a bounded number of times, and the pairs can be read off one list of new vertices for each color, so the algorithm runs in linear time.
The orientation $\mathrm{up}(x)$ is constructed in linear time as well, since every edge of $X$ is directed by reading the direction of at most one edge of $X_0$.
\item \textbf{The map $\mathrm{up}$.}
Let $x$ be a legal configuration of $X_0$.
Under $\mathrm{up}(x)$, every vertex of $X_0$ has the same inflow as under $x$, because the path edge at $u$ that replaces a mismatched edge $e=uv$ is directed toward $u$ exactly when $x$ directs $e$ toward $u$, and in the same way at $v$.
If the path of a mismatched edge $e=uv$ is directed toward $u$, then $g^e_u$ receives the inflow $2$ from $g^e_ug^e_v$ and $g^e_v$ receives the inflow $2$ from $g^e_vv$, and the other case is symmetric.
In every copy of $\mathcal{Q}$, the vertices $v_{\mathrm{L}}$ and $v_{\mathrm{R}}$ receive the inflow $2$ from their ports, $v_{\mathrm{T}}$ receives $1+1+2$, and $v_{\mathrm{B}}$ receives $1+1$.
Hence, $\mathrm{up}(x)$ is a legal configuration, and $\mathrm{down}(\mathrm{up}(x))=x$.
\item \textbf{The map $\mathrm{down}$.}
Let $x'$ be a legal configuration of $X$.
By \cref{lem:q-core}, every port is directed toward its copy of $\mathcal{Q}$, so each new vertex receives its inflow from its two path edges only.
We compare, at every vertex of $X_0$, the contribution of each of its edges to the inflow under $\mathrm{down}(x')$ with the contribution of the corresponding edge of $X$ under $x'$.
For a matched edge, the two contributions are equal.
Let $e=uv$ be a mismatched edge.
At $u$, the edge $ug^e_u$ contributes $2$ under $x'$ exactly when it is directed toward $u$, which is exactly when $\mathrm{down}(x')$ directs $e$ toward $u$, so the two contributions are equal.
At $v$, the edge $g^e_vv$ contributes $2$ under $x'$ only if it is directed toward $v$.
In that case, the inflow at $g^e_v$ comes from $g^e_ug^e_v$, which is therefore directed toward $g^e_v$, and then the inflow at $g^e_u$ comes from $ug^e_u$, which is therefore not directed toward $u$, so $\mathrm{down}(x')$ directs $e$ toward $v$ and $e$ also contributes $2$ at $v$.
Hence, the inflow under $\mathrm{down}(x')$ at every vertex of $X_0$ is at least its inflow under $x'$, which is at least $2$, and $\mathrm{down}(x')$ is a legal configuration of $X_0$.
\item \textbf{Protection.}
The proof for $\mathrm{down}$ shows more than the inflow inequality: every incoming transformed incidence at an old vertex remains incoming at that vertex under $\mathrm{down}$.
At the designated first end this is equality, and at the other end it follows from the implication proved in~(e).
Thus, if the transformed incidences of an old protected pair were both incoming, their projected incidences would violate protection in the old legal configuration.
Every new $\mathrm{OR}$ vertex has a port that always points outward by \cref{lem:q-core}; the port and its designated path edge therefore form a protected pair.
All vertices of the new cores are $\mathrm{AND}$ vertices, so this accounts for every $\mathrm{OR}$ vertex of $X$.
\item \textbf{Equivalence.} The map $x\mapsto x'$ of \cref{prop:tait-source} is $\mathrm{up}$, and $x_s$ and $x_t$ are arbitrary legal configurations of $X_0$.
\begin{itemize}
\item[($\Rightarrow$)] Let $x\to y$ be a legal reversal of an edge $e$ of $X_0$.
If $e$ is matched, then $\mathrm{up}(x)$ and $\mathrm{up}(y)$ differ only in $e$ and are legal by (d), so $\mathrm{up}(x)\to\mathrm{up}(y)$ is a legal reversal.
If $e=uv$ is mismatched and $x$ directs $e$ toward $u$, then reversing $ug^e_u$, $g^e_ug^e_v$ and $g^e_vv$ in this order leads from $\mathrm{up}(x)$ to $\mathrm{up}(y)$, and every vertex other than $u$, $g^e_u$, $g^e_v$ and $v$ keeps its inflow.
After the first reversal, $u$ has its inflow under $y$, which is at least $2$, the vertex $v$ keeps its inflow under $x$, the vertex $g^e_u$ receives the inflow $4$ from $ug^e_u$ and $g^e_ug^e_v$, and $g^e_v$ receives the inflow $2$ from $g^e_vv$.
After the second reversal, $v$ still has its inflow under $x$, $g^e_u$ receives the inflow $2$ from $ug^e_u$ and $g^e_v$ receives the inflow $4$, and the third reversal gives $\mathrm{up}(y)$.
The case in which $x$ directs $e$ toward $v$ is symmetric.
Hence, $x_s\sim_{X_0}x_t$ implies $\mathrm{up}(x_s)\sim_X\mathrm{up}(x_t)$.
\item[($\Leftarrow$)] Let $x'\to y'$ be a legal reversal of an edge of $X$.
The orientations $\mathrm{down}(x')$ and $\mathrm{down}(y')$ are legal configurations by (e), and they differ at most in the edge of $X_0$ whose direction $\mathrm{down}$ reads from the reversed edge, so $\mathrm{down}(x')=\mathrm{down}(y')$ or $\mathrm{down}(x')\to\mathrm{down}(y')$ is a legal reversal.
Applying this to every reversal of a sequence from $\mathrm{up}(x_s)$ to $\mathrm{up}(x_t)$ gives $\mathrm{down}(\mathrm{up}(x_s))\sim_{X_0}\mathrm{down}(\mathrm{up}(x_t))$, which is $x_s\sim_{X_0}x_t$ by (d).
\end{itemize}
\end{enumerate}
This completes the proof.
\end{proof}
 
\section{The Construction and Its Projection}\label{sec:construction}

With the protected Tait-colored source from \cref{sec:tait-source}, we now build a network whose colorings represent legal source configurations.
We construct its negators, gates, and wires in \cref{sec:negator,sec:gates,sec:network}, then check the local constraints and define the projection in \cref{sec:laws,sec:projection-map}.
Finally, \cref{sec:reflection} shows that one Kempe exchange changes only a matching of source edges, which can be reversed one at a time.

These facts apply to every proper coloring, without restricting the locations of any colors.
Throughout the section, $k\geq3$ is fixed and $X$ is the source graph of \cref{cor:tait-source}.

\subsection{The Negator}\label{sec:negator}

The preprocessing of \cref{sec:tait-source} supplied a protected source with a Tait coloring.
We first construct a negator whose two connectors always have opposite equality types, for every number $k\geq 3$ of colors.

We build it from the components of Leven and Galil~\cite{LevenGalil1983}. Their components compare a single edge with a multiple edge: at exactly one end, the single edge's color occurs among the multiple edges.
Joining the components through new vertices reverses this relation again, so the entire chain still inverts it.
At the two ends of the chain, the multiple edges have multiplicity one, giving the two required connectors.
The remaining dangling edges carry exactly the colors absent from the unequal connector, which will allow the triangle to control this connector.
We keep Leven and Galil's names for the multiple edges and give names to the vertices.

The statement of \cref{lem:inverting}(b) below compares the colors of two multiple edges, and the following terms make this precise for a coloring $c$ of a multipole.
A multiple edge $a$ \emph{contains} a color $h$ under $c$ if an edge of $a$ has the color $h$.
A single edge $b$ \emph{shares a color with} $a$ if $a$ contains the color of $b$; this is Leven and Galil's relation~\cite{LevenGalil1983} between two multiple edges, restricted to a single edge $b$. 

\begin{definition}\label{def:inverting}
Let $1\leq i\leq k-2$.
The \emph{inverting component} $H_i$ is the multipole of degree $k$ with the internal vertices $B$, $D$, $M$, $L$, $R$, $A$ and $C$, the multiple edges $u=BD$, $t=BL$, $v=DR$, $o=ML$, $p=MR$, $q=AR$, $s=LC$ and $r=AC$ of the multiplicities $k-2$, $1$, $1$, $k-1-i$, $i$, $k-1-i$, $i$ and $1$, respectively, and the multiple dangling edges $a$ at $A$, $b$ at $B$, $c$ at $C$, $d$ at $D$ and $e$ at $M$ of the multiplicities $i$, $1$, $k-1-i$, $1$ and $1$, respectively.
\end{definition}

Every internal vertex of $H_i$ has degree $k$; for example, $B$ is incident with the $k-2$ edges of $u$, the edge $t$ and the dangling edge $b$.
\Cref{fig:inverting} shows $H_i$.

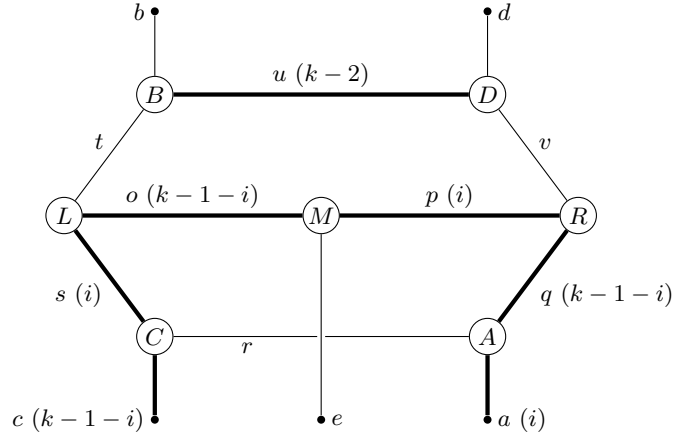
\begin{figure}[ht]
\centering
\begin{adjustbox}{max width=\textwidth}
\begin{tikzpicture}[
  vx/.style={circle, draw, fill=white, inner sep=0pt, minimum size=4.8mm, font=\footnotesize},
  tv/.style={circle, fill, inner sep=0pt, minimum size=3pt},
  lb/.style={font=\footnotesize, inner sep=2pt},
  me/.style={line width=1.6pt}
]
\node[vx] (B) at (-2.2,1.6) {$B$};
\node[vx] (D) at (2.2,1.6) {$D$};
\node[vx] (L) at (-3.4,0) {$L$};
\node[vx] (R) at (3.4,0) {$R$};
\node[vx] (C) at (-2.2,-1.6) {$C$};
\node[vx] (A) at (2.2,-1.6) {$A$};
\node[vx] (M) at (0,0) {$M$};
\node[tv] (xb) at (-2.2,2.7) {};
\node[tv] (xd) at (2.2,2.7) {};
\node[tv] (xa) at (2.2,-2.7) {};
\node[tv] (xc) at (-2.2,-2.7) {};
\node[tv] (xe) at (0,-2.7) {};
\draw[me] (B) -- node[lb, above] {$u$ ($k-2$)} (D);
\draw (B) -- node[lb, above left] {$t$} (L);
\draw (D) -- node[lb, above right] {$v$} (R);
\draw[me] (A) -- node[lb, below right] {$q$ ($k-1-i$)} (R);
\draw[me] (L) -- node[lb, below left] {$s$ ($i$)} (C);
\draw (A) -- node[lb, below, pos=0.75] {$r$} (C);
\draw[me] (M) -- node[lb, above] {$o$ ($k-1-i$)} (L);
\draw[me] (M) -- node[lb, above] {$p$ ($i$)} (R);
\draw (B) -- (xb);
\draw (D) -- (xd);
\draw[me] (A) -- (xa);
\draw[me] (C) -- (xc);
\draw[white, line width=3pt] (0,-1.42) -- (0,-1.78);
\draw (M) -- (xe);
\node[lb, left] at (xb.west) {$b$};
\node[lb, right] at (xd.east) {$d$};
\node[lb, right] at (xa.east) {$a$ ($i$)};
\node[lb, left] at (xc.west) {$c$ ($k-1-i$)};
\node[lb, right] at (xe.east) {$e$};
\end{tikzpicture}
\end{adjustbox}
\caption{The drawing shows the inverting component $H_i$, whose internal vertices $B$, $D$, $M$, $L$, $R$, $A$ and $C$ are drawn as circles and whose dangling edges end at small black dots.
A thin line is a single edge, labeled with its name, and a thick line is a multiple edge, labeled with its name and, in parentheses, its multiplicity.
The vertices $B$, $D$, $R$, $A$, $C$ and $L$ form a hexagon, and the vertex $M$ inside the hexagon is joined to $L$ by $o$ and to $R$ by $p$.
The dangling edges $b$, $d$, $a$ and $c$ leave the hexagon at $B$, $D$, $A$ and $C$.
The dangling edge $e$ at $M$ passes over the edge $r$ without meeting it.
The edges $o$, $p$, $q$, $s$, $a$ and $c$ are drawn thick for a general $i$, although for $i=1$ and for $i=k-2$ some of them are single edges, and $u$ is a single edge for $k=3$.}
\label{fig:inverting}
\end{figure}
 
The boundary counts and the exclusive alternative in the next lemma are due to Leven and Galil~\cite{LevenGalil1983}. We state them in our notation; they are the only properties of $H_i$ needed for the construction and its projection.

\begin{lemma}\label{lem:inverting}
Let $1\leq i\leq k-2$, and let $\varphi$ be a coloring of $H_i$.
\begin{enumerate}[(a)]
\item One color occurs on exactly three dangling edges of $H_i$, and every other color occurs on exactly one.
\item Exactly one of the following holds: $b$ shares a color with $a$, or $d$ shares a color with $c$.
\end{enumerate}
\end{lemma}

We now chain the inverting components into one multipole with two connectors.

\begin{definition}\label{def:negator}
The negator $W_k$ is obtained from $H_1,\dots,H_{k-2}$ as follows.
For $1\leq i\leq k-3$, the dangling edges of $c$ in $H_i$ and of $a$ in $H_{i+1}$ are attached to a new vertex $Z_i$, and the dangling edge $d$ of $H_i$ is identified with the dangling edge $b$ of $H_{i+1}$.
The input connector of $W_k$ is $(a,b)$ of $H_1$, its output connector is $(c,d)$ of $H_{k-2}$, and its residual edges are the edges $e$ of $H_1,\dots,H_{k-2}$.
\end{definition}

The vertex $Z_i$ has degree $(k-1-i)+(i+1)=k$, and the multiple edges $a$ of $H_1$ and $c$ of $H_{k-2}$ are single edges, so the two connectors of $W_k$ consist of single dangling edges.
For $k=3$, the negator $W_3$ is $H_1$, which is, up to isomorphism, the inverting component of Holyer~\cite{Holyer1981}. \Cref{fig:negator} shows the negator.

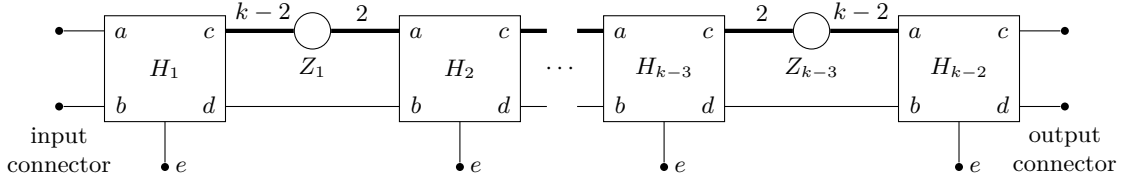
\begin{figure}[ht]
\centering
\begin{adjustbox}{max width=\textwidth}
\begin{tikzpicture}[
  vx/.style={circle, draw, fill=white, inner sep=0pt, minimum size=4.8mm, font=\footnotesize},
  tv/.style={circle, fill, inner sep=0pt, minimum size=3pt},
  lb/.style={font=\footnotesize, inner sep=2pt},
  bx/.style={draw, minimum width=16mm, minimum height=14mm, inner sep=0pt, font=\footnotesize},
  me/.style={line width=1.6pt}
]
\node[bx] (H1) at (0,0) {$H_1$};
\node[bx] (H2) at (3.9,0) {$H_2$};
\node[bx] (H3) at (6.6,0) {$H_{k-3}$};
\node[bx] (H4) at (10.5,0) {$H_{k-2}$};
\foreach \x/\n in {0/1,3.9/2,6.6/3,10.5/4} {
  \node[lb] at (\x-0.58,0.48) {$a$};
  \node[lb] at (\x-0.58,-0.48) {$b$};
  \node[lb] at (\x+0.58,0.48) {$c$};
  \node[lb] at (\x+0.58,-0.48) {$d$};
  \node[tv] (e\n) at (\x,-1.3) {};
  \draw (\x,-0.7) -- (e\n);
  \node[lb, right] at (e\n.east) {$e$};
}
\node[vx, label={[lb, label distance=1pt]below:$Z_1$}] (Z1) at (1.95,0.5) {};
\node[vx, label={[lb, label distance=1pt]below:$Z_{k-3}$}] (Z2) at (8.55,0.5) {};
\node[tv] (i1) at (-1.4,0.5) {};
\node[tv] (i2) at (-1.4,-0.5) {};
\node[tv] (o1) at (11.9,0.5) {};
\node[tv] (o2) at (11.9,-0.5) {};
\draw (-0.8,0.5) -- (i1);
\draw (-0.8,-0.5) -- (i2);
\draw (11.3,0.5) -- (o1);
\draw (11.3,-0.5) -- (o2);
\draw[me] (0.8,0.5) -- node[lb, above, inner ysep=3pt, pos=0.55] {$k-2$} (Z1);
\draw[me] (Z1) -- node[lb, above, inner ysep=3pt, pos=0.45] {$2$} (3.1,0.5);
\draw (0.8,-0.5) -- (3.1,-0.5);
\draw[me] (4.7,0.5) -- (5.05,0.5);
\draw (4.7,-0.5) -- (5.05,-0.5);
\node[lb] at (5.25,0) {$\cdots$};
\draw[me] (5.45,0.5) -- (5.8,0.5);
\draw (5.45,-0.5) -- (5.8,-0.5);
\draw[me] (7.4,0.5) -- node[lb, above, inner ysep=3pt, pos=0.55] {$2$} (Z2);
\draw[me] (Z2) -- node[lb, above, inner ysep=3pt, pos=0.45] {$k-2$} (9.7,0.5);
\draw (7.4,-0.5) -- (9.7,-0.5);
\node[lb, align=center, below] at (-1.4,-0.7) {input\\connector};
\node[lb, align=center, below] at (11.9,-0.7) {output\\connector};
\end{tikzpicture}
\end{adjustbox}
\caption{The drawing shows the negator $W_k$ as a chain of the inverting components $H_1,\dots,H_{k-2}$, each drawn as a box whose corner letters name its dangling edges $a$, $b$, $c$ and $d$ of \cref{fig:inverting}.
The dots stand for the components $H_3,\dots,H_{k-4}$, the vertices $Z_2,\dots,Z_{k-4}$ and the identified edges between $H_2$ and $H_{k-3}$, so that for $k=6$ they stand for $Z_2$ and one identified edge.
For $k\leq 5$, the chain is shorter, and for $k=3$ it consists of $H_1$ alone.
A thick line is a multiple edge, and a thin line is a single edge.
Between $H_i$ and $H_{i+1}$, the dangling edge $c$ of $H_i$ and the dangling edge $a$ of $H_{i+1}$ are attached to the new vertex $Z_i$, of which $Z_1$ and $Z_{k-3}$ are drawn as circles, and the dangling edge $d$ of $H_i$ is identified with the dangling edge $b$ of $H_{i+1}$.
The vertex $Z_i$ is incident with $k-1-i$ edges of $H_i$ and $i+1$ edges of $H_{i+1}$, and the labels at $Z_1$ and $Z_{k-3}$ give these multiplicities.
The two dangling edges on the far left form the input connector $(a,b)$ of $W_k$, the two on the far right form its output connector $(c,d)$, and the dangling edges $e$ below the boxes are its residual edges.}
\label{fig:negator}
\end{figure}
 
The next lemma shows that $W_k$ reverses equality types and identifies the set of residual colors.
These two properties let the gates read connector states and the triangles control later Kempe changes.

\begin{lemma}\label{lem:negator-law}
The negator $W_k$ is a multipole of degree $k$ with $8k-17$ internal vertices and $k+2$ dangling edges, and every coloring of $W_k$ has the following two properties.
\begin{enumerate}[(a)]
\item Exactly one of the input and the output connector is equal.
\item The $k-2$ residual edges have pairwise distinct colors, and these colors are the colors that do not occur on the unequal connector.
\end{enumerate}
\end{lemma}

\begin{proof}
We count the vertices first, then chain \cref{lem:inverting}(b) along the negator for~(a), and finally apply the Parity Lemma for~(b).
The negator consists of $k-2$ inverting components with seven internal vertices each and of the $k-3$ vertices $Z_i$, so it has $7(k-2)+(k-3)=8k-17$ internal vertices.
Its dangling edges are the two edges of each connector and the $k-2$ residual edges, since the other dangling edges of the components are attached to the vertices $Z_i$ or identified.
\begin{enumerate}[(a)]
\item For $1\leq i\leq k-2$, let $\delta_i$ be the statement that $b$ shares a color with $a$ in $H_i$, and let $\delta'_i$ be the statement that $d$ shares a color with $c$ in $H_i$.
By \cref{lem:inverting}(b), $\delta_i$ holds exactly when $\delta'_i$ fails.
For $1\leq i\leq k-3$, the edges of $c$ in $H_i$ and of $a$ in $H_{i+1}$ are the $k$ edges at $Z_i$, so their color sets are complementary, and the single edge $d$ of $H_i$, which is the edge $b$ of $H_{i+1}$, shares a color with exactly one of them.
Hence, $\delta'_i$ holds exactly when $\delta_{i+1}$ fails, and therefore $\delta_i$ holds exactly when $\delta_{i+1}$ holds.
Consequently, $\delta_1$ holds exactly when $\delta'_{k-2}$ fails.
Since $a$ in $H_1$ and $c$ in $H_{k-2}$ are single edges, $\delta_1$ says that the input connector is equal and $\delta'_{k-2}$ says that the output connector is equal.
\item The number $8k-17$ is odd, so by \cref{lem:parity} and the argument of \cref{lem:inverting}(a), one color $h$ occurs on three of the $k+2$ dangling edges and every other color on one.
By~(a), one connector is equal, and its two edges have a common color, which is therefore $h$.
Removing these two edges leaves every color exactly once on the unequal connector and the residual edges.
\end{enumerate}
\end{proof}

The name comes from M{\'a}{\v{c}}ajov{\'a} and {\v{S}}koviera~\cite{MacajovaSkoviera2006}, who call a colorable cubic multipole with two connectors of two edges each and one residual edge a negator if every coloring makes exactly one of the two connectors equal. By \cref{lem:negator-law}(a), together with the colorings of $W_3$ given in \cref{lem:w3}(a), the negator $W_3$ is a negator in their sense, and we keep their name for every $k$.

\subsection{The Gates}\label{sec:gates}

The negator law from \cref{sec:negator} reverses equality types; we now combine negators into gates for the source vertices.
Each gate has three ordered connectors, called its \emph{slots}, that represent the three incident source directions.

The $\mathrm{OR}$ gate joins three negators around a new vertex; its slot at zero reverses the physical output order of its negator.
The $\mathrm{AND}$ gate uses two negators and three new vertices.
Together with a triangle, this gate forces both light edges inward whenever the heavy edge points outward.
For a copy $g$ of $W_k$, its ordered input and output are denoted by $(a(g),b(g))$ and $(c(g),d(g))$, and its residual edges by $e_1(g),\ldots,e_{k-2}(g)$, in the order of their inverting components.
A new vertex \emph{receives} a dangling edge when a new terminal vertex is joined to it.

\begin{definition}\label{def:gates}
The \emph{$\mathrm{OR}$ gate} consists of three copies $g_0,g_1,g_2$ of $W_k$ and a new vertex $T$.
The edge $b(g_0)$ is identified with $a(g_1)$, and $b(g_1)$ with $a(g_2)$.
The edges $a(g_0)$ and $b(g_2)$ attach to $T$, which receives $k-2$ further dangling edges.
Its ordered slots $0,1,2$ are
\[
(d(g_0),c(g_0)),\qquad(c(g_1),d(g_1)),\qquad(c(g_2),d(g_2)).
\]
The \emph{$\mathrm{AND}$ gate} consists of two copies $g_0,g_1$ of $W_k$ and new vertices $x_0,x_1,x_2$.
The edge $a(g_0)$ attaches to $x_0$, which receives a new dangling edge $s_0$ and $k-2$ further dangling edges, denoted collectively by $r_0$.
The edge $d(g_0)$ and all residual edges of $g_1$ attach to $x_1$, which receives a new dangling edge $s_1$.
The edge $b(g_1)$ attaches to $x_2$, which receives a new dangling edge $r_2$ and $k-2$ further dangling edges, denoted collectively by $r'_2$.
Its ordered slots $0,1,2$ are
\[
(s_0,b(g_0)),\qquad(s_1,c(g_0)),\qquad(c(g_1),d(g_1)).
\]
All dangling edges outside these slots are residual edges of their gate.
\end{definition}

Each new vertex has degree $k$ by its listed incidences.
The two gates are connected multipoles of degree $k$, with $24k-50$ and $16k-31$ internal vertices, respectively.
The $\mathrm{AND}$ gate is connected because $x_1$ joins the two negators through $d(g_0)$ and the nonempty residual bundle of $g_1$; the other new vertices attach to one negator each.
\Cref{fig:gates} shows the identifications and ordered slots.

\usetikzlibrary{arrows.meta}
\begin{figure}[ht]
\centering
\begin{adjustbox}{max width=\textwidth}
\begin{tikzpicture}[
 vx/.style={circle,draw,fill=white,inner sep=0pt,minimum size=5mm,font=\footnotesize},
 tv/.style={circle,fill,inner sep=0pt,minimum size=3pt},
 lb/.style={font=\footnotesize,inner sep=2pt},
 bx/.style={draw,minimum width=16mm,minimum height=11mm,inner sep=0pt,font=\footnotesize},
 me/.style={line width=1.6pt},
 pr/.style={densely dotted,-{Stealth[length=4pt]},shorten <=2pt,shorten >=2pt}]
\usetikzlibrary{arrows.meta}
\begin{scope}
\node[bx] (g0) at (0,2.4) {$g_0$};
\node[bx] (g1) at (0,0) {$g_1$};
\node[bx] (g2) at (0,-2.4) {$g_2$};
\node[vx] (T) at (-2.5,0) {$T$};
\node[tv] (xT) at (-3.7,0) {};
\draw[me] (T) -- node[lb,above] {$k-2$} (xT);
\draw (T.north) -- (-2.5,2.7) -- (-0.8,2.7);
\draw (T.south) -- (-2.5,-2.7) -- (-0.8,-2.7);
\draw (-0.8,2.1) -- (-1.6,2.1) -- (-1.6,0.3) -- (-0.8,0.3);
\draw (-0.8,-0.3) -- (-1.6,-0.3) -- (-1.6,-2.1) -- (-0.8,-2.1);
\foreach \y/\j in {2.4/0,0/1,-2.4/2} {
 \node[tv] (r\j) at (0,\y-1.15) {};
 \draw[me] (0,\y-0.55) -- (r\j);
 \node[lb,right] at (r\j.east) {$k-2$};
 \node[tv] (sa\j) at (1.7,\y+0.3) {};
 \node[tv] (sb\j) at (1.7,\y-0.3) {};
 \draw (0.8,\y+0.3) -- (sa\j);
 \draw (0.8,\y-0.3) -- (sb\j);
}
\draw[pr] (sb0) to[bend right=75] node[lb,right] {slot $0$} (sa0);
\draw[pr] (sa1) to[bend left=75] node[lb,right] {slot $1$} (sb1);
\draw[pr] (sa2) to[bend left=75] node[lb,right] {slot $2$} (sb2);
\node[lb] at (-0.5,-4.2) {(a) $\mathrm{OR}$ gate};
\end{scope}
\begin{scope}[xshift=7.7cm]
\node[bx] (h0) at (0,1.8) {$g_0$};
\node[bx] (h1) at (0,-1.8) {$g_1$};
\node[vx] (x0) at (-2,2.1) {$x_0$};
\node[vx] (x1) at (3,1.5) {$x_1$};
\node[vx] (x2) at (-2,-2.1) {$x_2$};
\node[tv] (s0) at (-2,3.25) {};
\node[tv] (r0) at (-3.2,2.1) {};
\node[tv] (b0) at (-1.2,0.6) {};
\node[tv] (c0) at (1.65,2.1) {};
\node[tv] (s1) at (3,3.25) {};
\node[tv] (e0) at (0,0.6) {};
\node[tv] (a1) at (-1.5,-1.5) {};
\node[tv] (c1) at (1.65,-1.5) {};
\node[tv] (d1) at (1.65,-2.1) {};
\node[tv] (r2) at (-2,-3.25) {};
\node[tv] (rp2) at (-3.2,-2.1) {};
\draw (x0) -- node[lb,above] {$a(g_0)$} (-0.8,2.1);
\draw (x0) -- node[lb,left] {$s_0$} (s0);
\draw[me] (x0) -- node[lb,above] {$r_0$} (r0);
\draw (-0.8,1.5) -- (-1.2,1.5) -- node[lb,left] {$b(g_0)$} (b0);
\draw (0.8,2.1) -- node[lb,above] {$c(g_0)$} (c0);
\draw (0.8,1.5) -- node[lb,above] {$d(g_0)$} (x1);
\draw (x1) -- node[lb,right] {$s_1$} (s1);
\draw[me] (0,1.25) -- node[lb,right] {$k-2$} (e0);
\draw (-0.8,-1.5) -- node[lb,above] {$a(g_1)$} (a1);
\draw (-0.8,-2.1) -- node[lb,above] {$b(g_1)$} (x2);
\draw (0.8,-1.5) -- node[lb,above] {$c(g_1)$} (c1);
\draw (0.8,-2.1) -- node[lb,below] {$d(g_1)$} (d1);
\draw (x2) -- node[lb,left] {$r_2$} (r2);
\draw[me] (x2) -- node[lb,above] {$r'_2$} (rp2);
\draw[me] (0,-2.35) -- (0,-3.25) -- (3,-3.25) -- node[lb,right] {$e_1(g_1),\ldots,e_{k-2}(g_1)$} (x1);
\node[lb] at (0,-4.2) {(b) $\mathrm{AND}$ gate};
\end{scope}
\end{tikzpicture}
\end{adjustbox}
\caption{The two gates, with each box representing an entire copy of $W_k$.
Its input edges $a,b$ leave the left side and its output edges $c,d$ the right side, in each case first above second; its residual bundle leaves below.
Circles are the new internal vertices, small black dots are terminal vertices, thin lines are single edges, and thick lines represent $k-2$ edges.
Dotted arrows in part~(a) annotate slot order and are not graph edges: slot zero runs from $d(g_0)$ to $c(g_0)$, the reverse of the other outputs.
In part~(b), the ordered slots are $(s_0,b(g_0))$, $(s_1,c(g_0))$, and $(c(g_1),d(g_1))$.
All residual edges of $g_1$ attach to $x_1$; the edge $a(g_1)$ remains dangling.}
\label{fig:gates}
\end{figure}
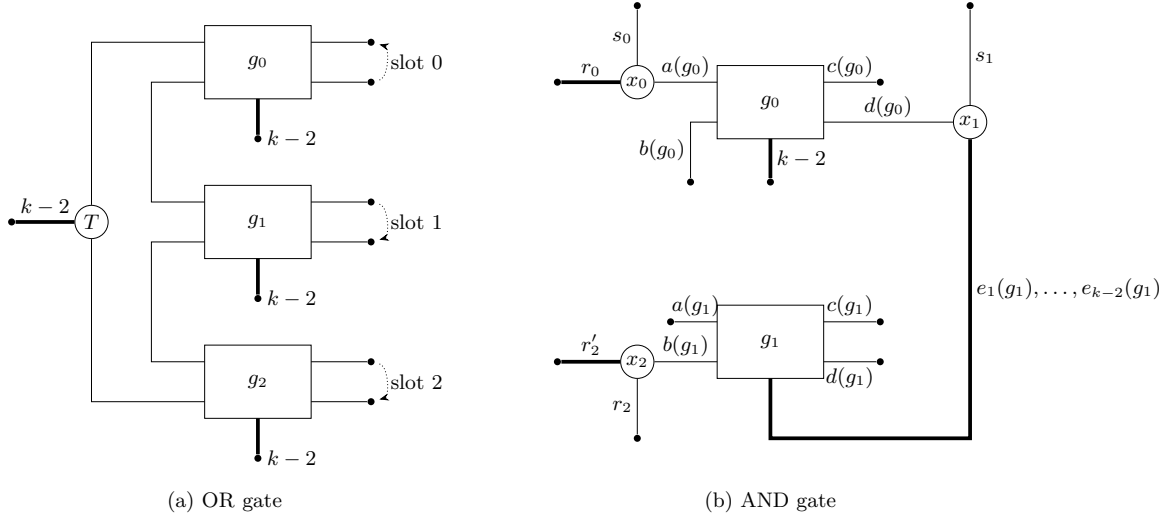
 
The \emph{equality indicator} of a connector is $1$ when it is equal and $0$ otherwise.
The \emph{polarity vectors}, indexed by slots $0,1,2$, are $p_{\mathrm{OR}}=(0,0,0)$ and $p_{\mathrm{AND}}=(1,1,0)$.
If slot $j$ of a gate of type $\theta$ has equality indicator $\varepsilon$, its \emph{represented bit} is $\varepsilon\mathbin\oplus p_\theta(j)$, where $\oplus$ is addition modulo two.
Thus, an equal $\mathrm{OR}$ slot represents one; in the $\mathrm{AND}$ gate, an unequal slot represents one at slots zero and one, and an equal slot represents one at slot two.

\subsection{The Network}\label{sec:network}

The gates of \cref{sec:gates} represent the local source states, and we now connect them along the source edges.
The wire lengths account for slot polarities, while attached triangles will control which slots can change together.

A \emph{wire} joins the two slots belonging to one source edge.
Its connector equality types alternate, and its length compensates for the two endpoint polarities.
Each wire uses its own copies of the negator, even when source edges are parallel.
Its end negators are distinct and attach to triangles at the two distinct source ends.

\begin{definition}\label{def:network}
The \emph{network} $F^k_X$ has one gate of type $\theta(z)$ and one triangle with vertices $S^z_0,S^z_1,S^z_2$ for each source vertex $z$.
For every source edge $f$ with first end $u$ and second end $v$, its wire has copies $w_0,\ldots,w_{q_f-1}$ of $W_k$, where
\[
q_f=\begin{cases}
3,&p_{\theta(u)}(s_u(f))=p_{\theta(v)}(s_v(f)),\\
2,&p_{\theta(u)}(s_u(f))\ne p_{\theta(v)}(s_v(f)).
\end{cases}
\]
The output of $w_i$ is identified coordinate by coordinate with the input of $w_{i+1}$.
The input of $w_0$ is identified with slot $s_u(f)$, and the output of $w_{q_f-1}$ with slot $s_v(f)$, in their ordered coordinates.
The resulting connectors along the wire are $\kappa_0,\ldots,\kappa_{q_f}$, with $\kappa_0$ at $u$ and $\kappa_{q_f}$ at $v$.
The residual edges of $w_0$ attach to $S^u_{s_u(f)}$, and those of $w_{q_f-1}$ attach to $S^v_{s_v(f)}$.
Every other residual edge of a wire or gate remains dangling.
\end{definition}

Every triangle vertex has its two triangle edges and $k-2$ attached residual edges, so its degree is $k$.
Every identification joins internal vertices of distinct copies, and every attachment joins a negator to a separate new vertex.
Hence, the network has no loops and is a multipole of degree $k$.
A pair of parallel source edges has separate wires and separate gate slots.
\Cref{fig:wire} records the short wires and their endpoint attachments.

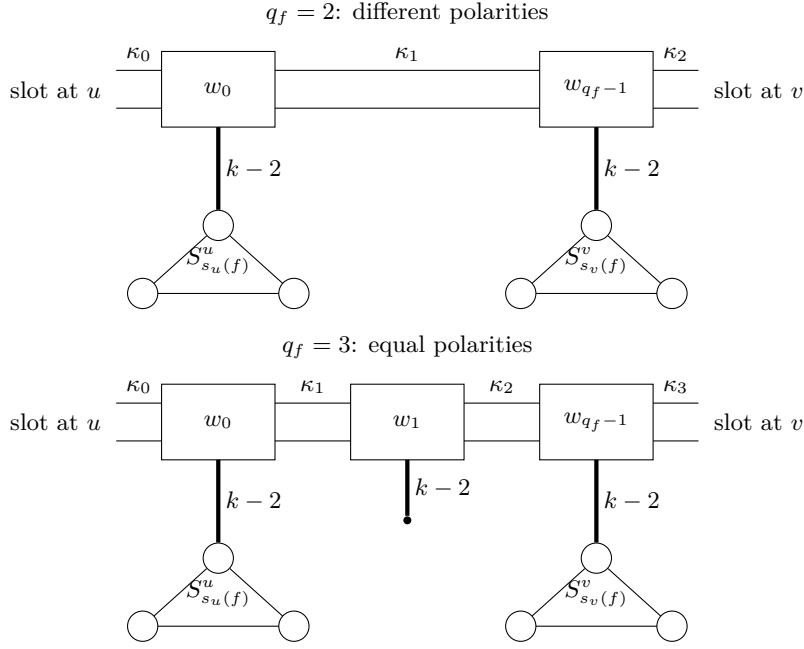
\begin{figure}[ht]
\centering
\begin{adjustbox}{max width=\textwidth}
\begin{tikzpicture}[
 bx/.style={draw,minimum width=15mm,minimum height=10mm,font=\footnotesize},
 vx/.style={circle,draw,fill=white,minimum size=4mm,inner sep=0pt},
 tv/.style={circle,fill,minimum size=3pt,inner sep=0pt},
 lb/.style={font=\footnotesize,inner sep=2pt},me/.style={line width=1.6pt}]
\foreach \yy/\qq in {0/2,-4.4/3} {
 \begin{scope}[yshift=\yy cm]
 \node[lb] at (-1.15,0) {slot at $u$};
 \node[lb] at (8.15,0) {slot at $v$};
 \node[bx] (L\qq) at (1,0) {$w_0$};
 \node[bx] (R\qq) at (6,0) {$w_{q_f-1}$};
 \draw (-0.35,0.25) -- (0.25,0.25);
 \draw (-0.35,-0.25) -- (0.25,-0.25);
 \draw (6.75,0.25) -- (7.35,0.25);
 \draw (6.75,-0.25) -- (7.35,-0.25);
 \node[lb,above] at (-0.05,0.3) {$\kappa_0$};
 \node[lb,above] at (7.05,0.3) {$\kappa_{\qq}$};
 \node[vx,label={[lb]below:$S^u_{s_u(f)}$}] (U\qq) at (1,-1.8) {};
 \node[vx,label={[lb]below:$S^v_{s_v(f)}$}] (V\qq) at (6,-1.8) {};
 \draw[me] (1,-0.5) -- node[lb,right] {$k-2$} (U\qq);
 \draw[me] (6,-0.5) -- node[lb,right] {$k-2$} (V\qq);
 \draw (U\qq) -- (0,-2.7) -- (2,-2.7) -- (U\qq);
 \draw (V\qq) -- (5,-2.7) -- (7,-2.7) -- (V\qq);
 \node[vx] at (0,-2.7) {};
 \node[vx] at (2,-2.7) {};
 \node[vx] at (5,-2.7) {};
 \node[vx] at (7,-2.7) {};
 \end{scope}
}
\draw (1.75,0.25) -- (5.25,0.25);
\draw (1.75,-0.25) -- (5.25,-0.25);
\node[lb,above] at (3.5,0.3) {$\kappa_1$};
\node[lb] at (3.5,1) {$q_f=2$: different polarities};
\begin{scope}[yshift=-4.4cm]
\node[bx] (M) at (3.5,0) {$w_1$};
\draw (1.75,0.25) -- (2.75,0.25);
\draw (1.75,-0.25) -- (2.75,-0.25);
\draw (4.25,0.25) -- (5.25,0.25);
\draw (4.25,-0.25) -- (5.25,-0.25);
\node[lb,above] at (2.25,0.3) {$\kappa_1$};
\node[lb,above] at (4.75,0.3) {$\kappa_2$};
\node[tv] (D) at (3.5,-1.3) {};
\draw[me] (3.5,-0.5) -- node[lb,right] {$k-2$} (D);
\node[lb] at (3.5,1) {$q_f=3$: equal polarities};
\end{scope}
\end{tikzpicture}
\end{adjustbox}
\caption{The two wire lengths in \cref{def:network}.
Each box represents $W_k$, and each pair of horizontal edges is identified coordinate by coordinate, first above second.
The endpoint residual bundles attach to the displayed vertices of the two triangles; the other triangle attachments and the endpoint gates are omitted.
For $q_f=3$, the middle residual bundle remains dangling.
The two endpoint negators are distinct in both cases.
A separate copy of this construction is used for each source edge, including parallel edges.}
\label{fig:wire}
\end{figure}
  \subsection{The Triangle and the Gates}\label{sec:laws}

The network of \cref{sec:network} is now assembled, but we still need to show that its colorings represent legal source configurations.
We first use the triangle colors to control unequal slots, then check the $\mathrm{OR}$ and $\mathrm{AND}$ laws.

The $\mathrm{OR}$ gate cannot have all three connectors unequal: the negators would force two edges at one vertex to share a color.
The $\mathrm{AND}$ gate needs the triangle to exclude one further possibility.
The three triangle edges have distinct colors, so their three vertices have different pairs of triangle colors.
The residual edges carry this information to the slots: an unequal connector has the pair at its triangle vertex.
This relation forces both light slots to represent inward edges whenever the heavy slot represents an outward edge.
It will also keep one Kempe change from reversing two source edges at the same vertex.
For a coloring $c$ of $F^k_X$, a vertex $z$ of $X$ and a slot $j$, let $\Pi^z_j$ be the color pair formed by the colors under $c$ of the two triangle edges at $S^z_j$; when several colorings are compared, we say under which coloring $\Pi^z_j$ is read.

\begin{lemma}\label{lem:triangle}
Let $c$ be a coloring of $F^k_X$, and let $z$ be a vertex of $X$.
\begin{enumerate}[(a)]
\item The color pairs $\Pi^z_0$, $\Pi^z_1$ and $\Pi^z_2$ are pairwise distinct, and, for each slot $j$, the $k-2$ residual edges at $S^z_j$ have the $k-2$ colors outside $\Pi^z_j$.
\item If the connector at slot $j$ of the gate at $z$ is unequal, then the colors of its two edges form the color pair $\Pi^z_j$.
\end{enumerate}
\end{lemma}

\begin{proof}
We use properness at the vertices of the triangle for~(a), and \cref{lem:negator-law} for the negator whose residual edges are attached to $S^z_j$ for~(b).
\begin{enumerate}[(a)]
\item Any two edges of the triangle have a common end, so the three triangle edges have three distinct colors.
For each slot $j$, let $h_j$ be the color of the triangle edge that is not incident with $S^z_j$.
Then $\Pi^z_j=\{h_0,h_1,h_2\}\setminus\{h_j\}$, and these three color pairs are distinct.
The $k$ edges at $S^z_j$ are its two triangle edges and $k-2$ residual edges, and they have distinct colors, so the residual edges have the $k-2$ colors outside $\Pi^z_j$.
\item By \cref{def:network}, the connector at slot $j$ is also a connector of a negator $w$ of a wire, and the residual edges of $w$ are attached to $S^z_j$.
This connector is unequal, so by \cref{lem:negator-law}(a) it is the unequal connector of $w$, and by \cref{lem:negator-law}(b) the residual edges of $w$ have the $k-2$ colors that do not occur on it.
By~(a), these are the colors outside $\Pi^z_j$, so the two colors of the connector form the color pair $\Pi^z_j$.
\end{enumerate}
\end{proof}

With the triangle, the gates obey the constraints of their vertices.

\begin{lemma}\label{lem:gate-laws}
Let $c$ be a coloring of $F^k_X$, and let $z$ be a vertex of $X$.
\begin{enumerate}[(a)]
\item If $z$ is an $\mathrm{OR}$ vertex, then some slot of the gate at $z$ represents $1$.
\item If $z$ is an $\mathrm{AND}$ vertex and slot $0$ of its gate represents $0$, then slots $1$ and $2$ represent $1$.
\end{enumerate}
\end{lemma}

\begin{proof}
We use a contradiction in each part, using \cref{lem:negator-law}(a) for the negators of the gate and properness at its new vertices.
\begin{enumerate}[(a)]
\item Suppose to the contrary that no slot of the gate at $z$ represents $1$.
Then the output connectors of $g_0$, $g_1$ and $g_2$ are unequal, so by \cref{lem:negator-law}(a) their input connectors are equal.
The second edge of the input connector of $g_0$ is the first edge of the input connector of $g_1$, and the second edge of the input connector of $g_1$ is the first edge of the input connector of $g_2$, so all edges of the three input connectors have the same color.
Among them are the first edge of the input connector of $g_0$ and the second edge of the input connector of $g_2$, which are two edges at $T$ and therefore have distinct colors, a contradiction.
\item Suppose that slot $0$ represents $0$, so $c(s_0)=c(b(g_0))$.
Properness at $x_0$ gives $c(a(g_0))\ne c(b(g_0))$.
The input of $g_0$ is unequal, so its output is equal by \cref{lem:negator-law}(a).
Properness at $x_1$ then gives $c(s_1)\ne c(d(g_0))=c(c(g_0))$.
Thus, slot $1$ is unequal and represents $1$.
The residual edges of $g_1$ at $x_1$ have exactly the colors outside the pair on $(s_1,c(g_0))$, which is $\Pi^z_1$ by \cref{lem:triangle}(b).
If slot $2$ were unequal, \cref{lem:negator-law}(b) for $g_1$ and \cref{lem:triangle}(b) would give these same residual edges exactly the colors outside $\Pi^z_2$.
These two descriptions would imply $\Pi^z_1=\Pi^z_2$, contrary to \cref{lem:triangle}(a).
Hence, slot $2$ is equal and also represents $1$.
\end{enumerate}
\end{proof}

\subsection{The Legal Configuration of a Coloring}\label{sec:projection-map}

The local gate laws of \cref{sec:laws} are now in place.
We read each source direction from the bits at its two wire ends and show that the resulting orientation is legal.

Each negator reverses the equality type, and the wire length accounts for the endpoint polarities.
The two ends therefore read opposite bits.
Directing the source edge toward the end that reads $1$ gives an orientation whose local states satisfy the gate laws.
Hence, the orientation is a legal configuration.

\begin{lemma}\label{lem:projection}
Let $c$ be a coloring of $F^k_X$.
For every edge $f$ of $X$ with ends $u$ and $v$, slot $s_u(f)$ of the gate at $u$ represents $1$ exactly when slot $s_v(f)$ of the gate at $v$ represents $0$.
Consequently, directing every edge $f$ toward the end whose slot represents $1$ defines an orientation $x_c$ of $X$, and $x_c$ is a legal configuration.
\end{lemma}

\begin{proof}
We follow the equality types of the connectors along the wire of $f$, and then we use \cref{lem:gate-laws} at every vertex.
Let $u$ be the first end and $v$ the second end of $f$, let $q=q_f$, and let $w_0,\dots,w_{q-1}$ and $\kappa_0,\dots,\kappa_q$ be the negators and the connectors of the wire of $f$ as in \cref{def:network}.
Then $\kappa_0$ is the connector at slot $s_u(f)$ of the gate at $u$, and $\kappa_q$ is the connector at slot $s_v(f)$ of the gate at $v$.
For $0\leq i\leq q$, let $\varepsilon_i=1$ if $\kappa_i$ is equal and $\varepsilon_i=0$ otherwise.
By \cref{lem:negator-law}(a), applied to the restriction of $c$ to $w_i$, we have $\varepsilon_{i+1}=1-\varepsilon_i$ for $0\leq i<q$, so $\varepsilon_q\equiv\varepsilon_0+q\pmod 2$.
By the definition of the represented bits, slot $s_u(f)$ represents the bit congruent to $\varepsilon_0+p_{\theta(u)}(s_u(f))$ modulo $2$, and slot $s_v(f)$ represents the bit congruent to $\varepsilon_q+p_{\theta(v)}(s_v(f))$ modulo $2$.
The sum of the two bits is therefore congruent to $q+p_{\theta(u)}(s_u(f))+p_{\theta(v)}(s_v(f))$ modulo $2$, and this number is odd by the choice of $q_f$ in \cref{def:network}.
Hence, exactly one of the two slots represents $1$.

The orientation $x_c$ is therefore well defined, and its local state at every vertex $z$ is the local state represented by the gate at $z$.
If $z$ is an $\mathrm{OR}$ vertex, then this state is not $000$ by \cref{lem:gate-laws}(a).
If $z$ is an $\mathrm{AND}$ vertex, then its first bit is $1$, or the state is $011$ by \cref{lem:gate-laws}(b).
In both cases, the local state is ordinarily legal for the type of $z$, so $x_c$ is a legal configuration.
Only now do we apply the global protection promise: at every $\mathrm{OR}$ vertex, $x_c$ belongs to the five states listed in \cref{sec:ncl}.
\end{proof}

We call $x_c$ the \emph{projection} of $c$.

\subsection{One Kempe Change}\label{sec:reflection}

By \cref{sec:projection-map}, every coloring has a legal projection.
We now show that a Kempe exchange changes that projection only by reversing source edges with no common end.
These edges can be reversed one at a time through legal configurations.

We allow an exchange on several Kempe chains of the same color pair, as the final graph construction will need.
If a slot changes its bit, exactly one of its two edges is recolored.
Its connector therefore has the exchanged pair before or after the change, when it is unequal.
The residual edges at its triangle vertex avoid this pair and keep their colors.
Thus, the triangle vertex has the exchanged pair in both colorings, and no other slot at that gate can change its bit.
The reversed source edges consequently have no common end.
Reversing them one at a time leaves each vertex in either its old or its new legal local state.

\begin{lemma}\label{lem:guard}
Let $c$ be a coloring of $F^k_X$, let $P$ be a color pair, and let $c'$ be obtained from $c$ by exchanging the colors of $P$ on a set of Kempe chains of $c$ for $P$.
Then, at every vertex of $X$, at most one slot of the gate represents different bits under $c$ and under $c'$.
\end{lemma}

\begin{proof}
We show that every slot whose bit changes has the color pair $P$ at its triangle vertex under $c$, and then we apply \cref{lem:triangle}(a).
The exchange maps the colors of $P$ to colors of $P$ and keeps the other colors, so the edges colored in $P$ are the same under $c$ and under $c'$, and the Kempe chains of $c$ and of $c'$ for $P$ are the same.
Hence, $c'$ is obtained from $c$ by moves, one for each Kempe chain of the set, and it is a coloring of $F^k_X$.
Since this exchange is an involution on the same Kempe chains, $c$ is obtained from $c'$ in the same way, so $c$ and $c'$ play symmetric roles below.

Let $z$ be a vertex of $X$, and let $j$ be a slot at which the gate at $z$ represents different bits under $c$ and under $c'$.
The bit of a slot depends only on the equality type of its connector $(f_1,f_2)$, so this connector is unequal under one of $c$ and $c'$, say under $d$, and equal under the other, say under $d'$.
The exchange changes the color of an edge by the transposition of the two colors of $P$ or not at all, and this transposition is a bijection, so if both $f_1$ and $f_2$ or neither of them changed their colors, then the connector would remain unequal.
Hence, exactly one of them changes its color; by symmetry, let it be $f_1$.
Then $d(f_1)\in P$, and $d'(f_1)$ is the other color of $P$, which is $d'(f_2)=d(f_2)$, so $\{d(f_1),d(f_2)\}=P$.
By \cref{lem:triangle}(b) applied to $d$, the set $\Pi^z_j$ under $d$ is $P$, and by \cref{lem:triangle}(a) the residual edges at $S^z_j$ have the colors outside $P$ under $d$.
These residual edges lie in no Kempe chain for $P$, so they have the same colors under $c$ and under $c'$, and by \cref{lem:triangle}(a) applied to $c$, the set $\Pi^z_j$ under $c$ is also $P$.
If two distinct slots $j$ and $j'$ of the gate at $z$ changed their bits, then $\Pi^z_j$ and $\Pi^z_{j'}$ under $c$ would both be $P$, which contradicts \cref{lem:triangle}(a).
\end{proof}

By \cref{lem:guard}, the source edges whose directions change form a matching, and the next lemma reverses them one at a time.

\begin{lemma}\label{lem:reflection}
Let $c$ and $c'$ be as in \cref{lem:guard}.
Then no two edges of $X$ that are directed differently by $x_c$ and $x_{c'}$ have a common end, and reversing these edges one at a time, in any order, leads from $x_c$ to $x_{c'}$ through legal configurations.
Consequently, $x_d\sim_X x_{d'}$ for all colorings $d$ and $d'$ of $F^k_X$ with $d\sim_{F^k_X}d'$.
\end{lemma}

\begin{proof}
We show that the edges directed differently form a matching, and then we reverse them one at a time.
Let $R$ be the set of edges of $X$ that are directed differently by $x_c$ and $x_{c'}$.
By \cref{lem:projection}, an edge $f$ with an end $z$ lies in $R$ exactly when slot $s_z(f)$ of the gate at $z$ represents different bits under $c$ and under $c'$.
By \cref{lem:guard}, every vertex of $X$ is therefore an end of at most one edge of $R$.
This matching property also excludes two parallel edges from $R$; the argument uses their distinct slots, not simplicity.
Let $f_1,\dots,f_m$ be the edges of $R$ in any order, let $x^{(0)}=x_c$, and, for $1\leq i\leq m$, let $x^{(i)}$ be obtained from $x^{(i-1)}$ by reversing $f_i$, so that $x^{(m)}=x_{c'}$.
Since at most one edge of $R$ is incident with a vertex $z$ of $X$, the local state of $x^{(i)}$ at $z$ is $x_c|_z$ if no edge of $R$ at $z$ lies among $f_1,\dots,f_i$, and it is $x_{c'}|_z$ otherwise.
Both states are legal by \cref{lem:projection}, so every $x^{(i)}$ is a legal configuration, and every step from $x^{(i-1)}$ to $x^{(i)}$ is a legal reversal.

Finally, if $d\sim_{F^k_X}d'$, then $d'$ is obtained from $d$ by a sequence of moves, and every move exchanges the colors of a color pair on one Kempe chain.
By the first part, the projections of any two consecutive colorings of this sequence are related by $\sim_X$, and hence $x_d\sim_X x_{d'}$.
\end{proof}
 \section{Simulation of Legal Reversals}\label{sec:three}

\Cref{sec:reflection} shows that Kempe reachability implies source reachability; we now prove the reverse implication.
We analyze cubic negators in \cref{sec:w3}, choose canonical colorings in \cref{sec:encoder}, and join the gate paths of \cref{sec:direct-paths} into one global path in \cref{sec:composition}.
Finally, in \cref{sec:colours,sec:lift}, we add fixed colors and lift the moves through digons for every $k\geq3$.

There are four possible local transitions at a protected $\mathrm{OR}$ vertex and five at an $\mathrm{AND}$ vertex; seven explicit path patterns cover all nine.
Until \cref{sec:colours}, the palette is $\{0,1,2\}$.
The source graph $X$ satisfies \cref{cor:tait-source} throughout.

\subsection{Three Colors and the Negator}\label{sec:w3}

The projection argument of \cref{sec:reflection} handled arbitrary colorings; to construct moves, we now need exact routes through $W_3$.
We determine each coloring from its connector states and list the Kempe paths that meet its dangling edges.

At a vertex, the colors of two edges force the color of the third.
Applying this rule along the internal paths gives the coloring and its Kempe chains.
When the connectors use only two colors, two paths join their corresponding edges.
Exchanging along the path joining their second edges changes both equality types.
For every other color pair, we also determine which dangling edges a Kempe chain joins.
Finally, turning the negator around exchanges its connectors, so a wire can be read from either end.

For a color pair $P=\{a,b\}\subseteq\{0,1,2\}$, we write $\overline{P}=\overline{ab}$ for its third color.
For an unequal state $(a,b)$ of a connector, the color $\overline{ab}$ is the color \emph{missing} from the connector, in analogy with a color missing at a vertex~\cite{Mohar2007}. For $k=3$, a negator has exactly one residual edge, and its color under a coloring is the \emph{residual color} of the negator; by \cref{lem:negator-law}(b), it is the color missing from the unequal connector.
For a copy $w$ of $W_3$ and an edge name $x$ of \cref{def:inverting}, we write $x(w)$ for the edge $x$ of $w$.
In \cref{lem:w3,lem:reversal}, which concern a single negator, we write $x$ for this edge; elsewhere we write $x(w)$.
This convention distinguishes the edge names $s$, $u$, $v$, $p$ and $q$ from a slot $s$, the ends $u$ and $v$ of an edge of $X$, the map $p_\theta$ of \cref{sec:network} and the number $q_f$.

The following rule describes how a coloring passes along a path through two internal vertices, each of which is incident with one further edge; it is the basic step in the analysis of $W_3$.

\begin{lemma}\label{lem:path-rule}
Let $c$ be a coloring of a multipole of degree three, let $g_0,g_1,g_2,g_3$ be a path in which $g_1$ and $g_2$ are internal vertices, and let $\xi$ and $\xi'$ be the colors of the third edges at $g_1$ and at $g_2$, that is, of the edges at $g_1$ and $g_2$ that do not lie on the path.
If $\xi\neq\xi'$, then the edges $g_0g_1$, $g_1g_2$ and $g_2g_3$ have the colors $\xi'$, $\overline{\xi\xi'}$ and $\xi$, respectively.
If $\xi=\xi'$, then the edges $g_0g_1$ and $g_2g_3$ have a common color $\zeta\neq\xi$, and the edge $g_1g_2$ has the color $\overline{\xi\zeta}$.
Conversely, every coloring of the three path edges and the two third edges as in one of these two cases gives distinct colors to the three edges at $g_1$ and to the three edges at $g_2$.
\end{lemma}

\begin{proof}
In order to prove both cases and the converse, we compare the two color sets that properness gives at $g_1$ and at $g_2$.
At $g_1$, the colors of $g_0g_1$ and of the third edge are the two colors other than the color of $g_1g_2$, and at $g_2$ the same holds for $g_2g_3$ and the third edge at $g_2$.
Hence, the color sets $\{c(g_0g_1),\xi\}$ and $\{c(g_2g_3),\xi'\}$ are equal.
If $\xi\neq\xi'$, then $\xi'$ lies in the first set and differs from $\xi$, so $c(g_0g_1)=\xi'$, and in the same way $c(g_2g_3)=\xi$, while $c(g_1g_2)$ differs from both and is $\overline{\xi\xi'}$.
If $\xi=\xi'$, then $c(g_0g_1)=c(g_2g_3)$, this color $\zeta$ differs from $\xi$, and $c(g_1g_2)$, which differs from $\xi$ and $\zeta$, is $\overline{\xi\zeta}$.
For the converse, in the first case the three edges at $g_1$ have the colors $\xi'$, $\overline{\xi\xi'}$ and $\xi$, and so do the three edges at $g_2$.
In the second case, the three edges at $g_1$ and the three edges at $g_2$ have the colors $\zeta$, $\overline{\xi\zeta}$ and $\xi$.
\end{proof}

The next lemma gives every coloring of $W_3$ explicitly, in the names of \cref{def:inverting}, and it describes, for every color pair, the Kempe chains inside the negator that end at its dangling edges.
In its formulas, the name of an edge stands for the color of that edge.

\begin{lemma}\label{lem:w3}
Let $(a,b)$ and $(c,d)$ be states of the input and output connectors of $W_3$ of opposite equality types, and let $h$ be the color missing from the unequal one.
\begin{enumerate}[(a)]
\item The negator $W_3$ has exactly one coloring with these connector states; its internal edges have the colors
\begin{equation*}
(t,v,u)=\begin{cases}(d,b,\overline{bd})&\text{if } b\neq d,\\(h,h,\overline{bh})&\text{if } b=d,\end{cases}\qquad
(s,q,r)=\begin{cases}(a,c,\overline{ac})&\text{if } a\neq c,\\(h,h,\overline{ah})&\text{if } a=c,\end{cases}
\end{equation*}
and $o=\overline{ts}$, $p=\overline{vq}$, while the residual edge $e$ has the color $h$.
\item If the four connector colors lie in a color pair $P$, so that $h=\overline{P}$, then the edges colored in $P$ form two paths, the \emph{first path} joining the first edges $a$ and $c$ of the two connectors and the \emph{second path} joining their second edges $b$ and $d$.
Exchanging the colors of $P$ on the second path gives the coloring of~(a) for the states $(a,b')$ and $(c,d')$, where $b'$ and $d'$ are the colors of $P$ other than $b$ and $d$.
\item If $P$ is a color pair that does not contain all four connector colors, let $\mu$ be the color of the two edges of the equal connector.
Every Kempe chain for $P$ that contains a dangling edge is a path between two dangling edges colored in $P$, and these paths are as follows.
If $h\notin P$, then $\mu=h$, and one path joins the two edges of the unequal connector.
If $h\in P$, then one path joins the edge of the unequal connector colored in $P$ to the residual edge $e$, and, if $\mu\in P$, another path joins the two edges of the equal connector.
\end{enumerate}
\end{lemma}

\begin{proof}
We derive the colors from \cref{lem:path-rule} on the two paths $L,B,D,R$ and $L,C,A,R$, and we check the result at the seven internal vertices.
By \cref{def:inverting} with $k=3$, the path $L,B,D,R$ consists of the edges $t$, $u$ and $v$, and the third edges at its internal vertices $B$ and $D$ are $b$ and $d$; the path $L,C,A,R$ consists of the edges $s$, $r$ and $q$, and the third edges at $C$ and $A$ are $c$ and $a$.
\begin{enumerate}[(a)]
\item Let $\varphi$ be a coloring of $W_3$ with the given connector states.
By \cref{lem:path-rule} on the path $L,B,D,R$, $(t,v,u)=(d,b,\overline{bd})$ if $b\neq d$, while $t=v=\zeta$ and $u=\overline{b\zeta}$ for some color $\zeta\neq b$ if $b=d$; on the path $L,C,A,R$, $(s,q,r)=(a,c,\overline{ac})$ if $a\neq c$, while $s=q=\zeta'$ and $r=\overline{a\zeta'}$ for some color $\zeta'\neq a$ if $a=c$.
Properness at $L$, $R$ and $M$ gives $o=\overline{ts}$, $p=\overline{vq}$ and $e=\overline{op}$.
The two connectors have opposite equality types, so $a=c$ and $b=d$ do not both hold, and three cases remain.
\begin{itemize}
\item \textbf{Case 1: $b\neq d$ and $a\neq c$.}
The formulas for $t,v,u,s,q,r$ are those of the statement, and it remains to show that $e=h$.
If $(a,b)$ is equal, then $a=b$ differs from $c$ and from $d$, so $a=b=h$, and therefore $o=\overline{dh}=c$, $p=\overline{hc}=d$ and $e=\overline{cd}=h$.
If $(c,d)$ is equal, then $c=d=h$, and therefore $o=\overline{ha}=b$, $p=\overline{bh}=a$ and $e=\overline{ba}=h$.
\item \textbf{Case 2: $b=d$ and $a\neq c$.}
Properness at $L$ and at $R$ gives $\zeta\neq s=a$ and $\zeta\neq q=c$, so $\zeta=\overline{ac}$.
Exactly one of the states $(a,b)$ and $(c,b)$ is equal, so $b\in\{a,c\}$, and the color missing from the unequal one is $\overline{ac}$ in both cases; hence $\zeta=h$, as in the statement.
Then $o=\overline{ha}=c$, $p=\overline{hc}=a$ and $e=\overline{ca}=h$.
\item \textbf{Case 3: $a=c$ and $b\neq d$.}
In the same way, $\zeta'=\overline{bd}=h$, and then $o=\overline{dh}=b$, $p=\overline{bh}=d$ and $e=\overline{bd}=h$.
\end{itemize}
Hence, every coloring with the given connector states is the one displayed in the statement, which proves uniqueness.
Conversely, the displayed coloring is proper at $B$, $D$, $A$ and $C$ by the converse part of \cref{lem:path-rule}, since $\zeta=h\neq b$ in Case~2 and $\zeta'=h\neq a$ in Case~3, and it is proper at $L$, $R$ and $M$, because the lists above show $t\neq s$, $v\neq q$ and $o\neq p$ in every case.
\item Let $P$ be a color pair that contains $a$, $b$, $c$ and $d$.
The unequal connector uses both colors of $P$, so $h=\overline{P}$.
In Case~1 of the proof of~(a), $h$ is one of the four connector colors, so Case~2 or Case~3 applies.
In Case~2, the edges with a color other than $h$ are $a$, $b$, $c$, $d$, $u$, $s$, $q$, $o$ and $p$, and they form the path $a,q,p,o,s,c$ through $A$, $R$, $M$, $L$ and $C$ and the path $b,u,d$ through $B$ and $D$.
In Case~3, they are $a$, $b$, $c$, $d$, $t$, $v$, $r$, $o$ and $p$, and they form the path $a,r,c$ through $A$ and $C$ and the path $b,t,o,p,v,d$ through $B$, $L$, $M$, $R$ and $D$.
In both cases, the first path joins $a$ to $c$ and the second joins $b$ to $d$.
Exchanging the colors of $P$ on the second path replaces the connector states by $(a,b')$ and $(c,d')$, which have opposite equality types, because $(a,b)$ is equal exactly when $(a,b')$ is unequal, and the unequal one of them is colored in $P$ and misses $h$.
The result is a coloring of $W_3$ with these states, so by the uniqueness in~(a) it is the coloring of~(a) for them.
\item Let $P$ be a color pair that does not contain all four connector colors; we use~(a) to find the dangling edges colored in $P$, and we follow the Kempe chains from them, according to whether $h$ and $\mu$ lie in $P$.
As noted in \cref{sec:prelim-multipoles}, a Kempe chain that contains a dangling edge is a path whose two end edges are dangling edges colored in $P$.
By~(a), the residual edge $e$ has the color $h$, and the two edges of the unequal connector have the two colors other than $h$.
If $h\notin P$, then $P$ consists of the colors of the unequal connector, so $\mu\notin P$ by the assumption on $P$, and hence $\mu=h$; the only dangling edges colored in $P$ are the two edges of the unequal connector, which are therefore the ends of one path.
If $h\in P$ and $\mu\notin P$, the only dangling edges colored in $P$ are the edge of the unequal connector colored in $P$ and the edge $e$, which are again the ends of one path.
It remains to consider the case in which $h\in P$ and $\mu\in P$, where four dangling edges are colored in $P$.
The map $\rho$ of \cref{lem:reversal} below, whose proof uses only \cref{def:inverting}, exchanges the two connectors, keeps the order of their edges and fixes $e$, so we may assume that the input connector is the equal one, that is, $a=b=\mu$.
In the three cases of the proof of~(a), we list the path through $a$ and the path through the edge of the unequal connector colored in $P$, with the colors that the formulas of~(a) give for $a=b=\mu$.
\begin{itemize}
\item \textbf{Case 1: $b\neq d$ and $a\neq c$.}
Here $\mu=h$, and $t=d$, $v=h$, $u=c$, $s=h$, $q=c$, $r=d$, $o=c$ and $p=d$.
If $P=\{h,c\}$, the two paths are $a,q,v,u,b$ through $A$, $R$, $D$ and $B$, and $c,s,o,e$ through $C$, $L$ and $M$.
If $P=\{h,d\}$, they are $a,r,s,t,b$ through $A$, $C$, $L$ and $B$, and $d,v,p,e$ through $D$, $R$ and $M$.
\item \textbf{Case 2: $b=d$ and $a\neq c$.}
Here $\mu=d$, so $P=\{h,d\}$, and $t=v=h$, $u=c$, $s=d$, $q=c$, $r=h$, $o=c$ and $p=d$; the two paths are $a,r,s,t,b$ through $A$, $C$, $L$ and $B$, and $d,v,p,e$ through $D$, $R$ and $M$.
\item \textbf{Case 3: $a=c$ and $b\neq d$.}
Here $\mu=c$, so $P=\{h,c\}$, and $t=d$, $v=c$, $u=h$, $s=q=h$, $r=d$, $o=c$ and $p=d$; the two paths are $a,q,v,u,b$ through $A$, $R$, $D$ and $B$, and $c,s,o,e$ through $C$, $L$ and $M$.
\end{itemize}
In each list, the edges are colored in $P$, consecutive edges share the vertices named with the list, in this order, and the third edge at each of these vertices is not colored in $P$.
Hence, each list is a Kempe chain for $P$, and the two lists join the dangling edges as claimed.
\end{enumerate}
\end{proof}

Parts~(b) and~(c) say, for every color pair, which dangling edges of a colored negator are joined by Kempe chains inside it, and the gate paths below use precisely these routes.

The symmetry used in the preceding proof is an isomorphism of multipoles, which we define next.
\label{sec:iso}
An \emph{isomorphism} of multipoles is an isomorphism of the underlying multigraphs, that is, a pair of bijections between the vertex sets and between the edge sets that preserves incidence, which maps internal vertices onto internal vertices, and an \emph{automorphism} of a multipole $F$ is an isomorphism from $F$ to itself.
An isomorphism therefore maps terminal vertices onto terminal vertices and dangling edges onto dangling edges.
For an isomorphism $\Theta$ from a multipole $F$ to a multipole $F'$ and a coloring $c$ of $F$, the coloring $\Theta_*c$ of $F'$ gives each edge $\Theta(f)$ the color $c(f)$, and it is proper, since $\Theta$ maps the $k$ edges at an internal vertex to the $k$ edges at its image.
For a color pair $P$, an edge $f$ is colored in $P$ under $c$ exactly when $\Theta(f)$ is colored in $P$ under $\Theta_*c$, so $\Theta$ maps the Kempe chains of $c$ for $P$ onto the Kempe chains of $\Theta_*c$ for $P$.
Moreover, if the exchange of $P$ on a Kempe chain $Y$ turns $c$ into $c'$, then the exchange of $P$ on $\Theta(Y)$ turns $\Theta_*c$ into $\Theta_*c'$.

The negator is symmetric under the exchange of its two connectors, which lets the negators of a wire be read from either end of the wire.

\begin{lemma}\label{lem:reversal}
The map $\rho$ that exchanges $B$ and $D$, $L$ and $R$, $A$ and $C$ and fixes $M$, extended to the terminal vertices, is an automorphism of $W_3$ that exchanges the input and the output connector, keeps the order of the edges in each connector and fixes the residual edge.
\end{lemma}

\begin{proof}
We read the action of $\rho$ off the list of edges in \cref{def:inverting}.
The map $\rho$ sends the internal edges $u=BD$, $t=BL$, $v=DR$, $o=ML$, $p=MR$, $q=AR$, $s=LC$ and $r=AC$ to $u$, $v$, $t$, $p$, $o$, $s$, $q$ and $r$, and the dangling edges $a$, $b$, $c$, $d$ and $e$ to $c$, $d$, $a$, $b$ and $e$, together with their terminal vertices, so it is an automorphism, it maps $(a,b)$ onto $(c,d)$ and $(c,d)$ onto $(a,b)$ coordinate by coordinate, and it fixes $e$.
\end{proof}

\subsection{Canonical Colorings}\label{sec:encoder}

The negator analysis of \cref{sec:w3} gives both unique extensions and exact Kempe routes.
We now choose one coloring for every legal source configuration, using its supplied Tait coloring to make the two ends of each wire agree.
These canonical colorings will be the endpoints of the simulation.

The canonical slot colors use the two colors other than the slot number.
For each gate type, its first-coordinate vector is
\[
\gamma_{\mathrm{OR}}=(1,0,0),\qquad
\gamma_{\mathrm{AND}}=(1,0,1).
\]
For a permitted local state $\beta$ and slot $j$, the slot state has first coordinate $\gamma_\theta(j)$.
Its second coordinate is $\gamma_\theta(j)$ when $\beta_j\ne p_\theta(j)$, and $3-j-\gamma_\theta(j)$ otherwise.
Thus, its equality indicator is $\beta_j\mathbin\oplus p_\theta(j)$ and both colors differ from $j$.
At $k=3$, each residual bundle is one edge; we write $e(g)$ for the residual edge of a negator $g$, and $e(T)$ for the dangling edge at the $\mathrm{OR}$ vertex $T$.

The next definition specifies all other edges outside the negators.
The uniqueness in \cref{lem:w3}(a) then colors every negator internally.

\begin{definition}\label{def:gate-colouring}
For a local state $\beta$ occurring at a source vertex of type $\theta$, the \emph{gate coloring} $\Sigma_\theta(\beta)$ has the slot states above and the following remaining colors.
For the $\mathrm{OR}$ gate, the colors on $a(g_0)$, $b(g_0)=a(g_1)$, $b(g_1)=a(g_2)$ and $b(g_2)$ are the four digits in the ring-color entry under $\beta$ in \cref{tab:canonical}, in that order; the edge $e(T)$ has the color absent from the two ring edges at $T$.
For the $\mathrm{AND}$ gate, the colors are
\begin{align*}
a(g_0)&=2,& b(g_0)&=1+\beta_0,& c(g_0)&=2\beta_1,& d(g_0)&=2-\beta_0,\\
a(g_1)&=2-\beta_0,& b(g_1)&=1-\beta_2,& c(g_1)&=1,& d(g_1)&=\beta_2,\\
s_0&=1,&s_1&=0,&r_0&=0,&r_2&=\beta_2,&r'_2&=2.
\end{align*}
Here an edge name stands for its color, and $r_0,r'_2$ are the single edges of their bundles.
Each negator receives the extension of \cref{lem:w3}(a).
\end{definition}

\begin{table}[ht]
\centering
\begin{tabular}{c|ccccc}
local state $\beta$ & $001$ & $011$ & $010$ & $110$ & $100$\\\hline
ring colors & $2221$ & $2201$ & $2200$ & $1200$ & $1222$
\end{tabular}
\caption{The canonical colors on the four consecutive ring edges of the protected $\mathrm{OR}$ gate.
Adjacent columns are exactly its four allowed local transitions.
The ordered slot zero is $(d(g_0),c(g_0))$, so its second edge is $c(g_0)$.}
\label{tab:canonical}
\end{table}

These prescriptions must define proper colorings, rather than merely the desired equality types.
The following verification also identifies the residual colors used by the direct paths.

\begin{lemma}\label{lem:gate-canonical}
For every local state occurring in the protected source, $\Sigma_\theta(\beta)$ is a proper coloring of its gate with the prescribed ordered slot states.
For the $\mathrm{AND}$ gate its negator residual colors are $e(g_0)=2-2\beta_1$ and $e(g_1)=1+\beta_0$.
\end{lemma}

\begin{proof}
We check opposite connector equality in every negator and distinct colors at the additional vertices.
For the $\mathrm{OR}$ gate, the ring-color entry under $\beta$ in \cref{tab:canonical} has equal consecutive digits exactly at those positions $j$ for which $\beta_j=0$.
Its output connector is equal exactly when $\beta_j=1$, so \cref{lem:w3}(a) applies to all three negators.
The first and last digits of every ring-color entry differ, and the third color on $e(T)$ makes the coloring proper at $T$.
Reversing the order of slot zero does not change its equality type.

For the $\mathrm{AND}$ gate, the five legal states are $011,100,101,110,111$.
If $\beta_0=0$, the state is $011$; the input of $g_0$ is then $(2,1)$ and its output is $(2,2)$, while the input of $g_1$ is $(2,0)$ and its output is $(1,1)$.
If $\beta_0=1$, the input of $g_0$ is $(2,2)$ and its output is $(2\beta_1,1)$, and the input of $g_1$ is $(1,1-\beta_2)$ while its output is $(1,\beta_2)$.
In every case the two connectors have opposite equality types, and \cref{lem:w3}(a) gives the stated residual colors.
The colors at $x_0$ are $2,1,0$, those at $x_1$ are $2-\beta_0,0,1+\beta_0$, and those at $x_2$ are $1-\beta_2,\beta_2,2$.
Each list contains three distinct colors.
The three slot pairs agree with the prescribed states by substitution into the formulas.
\end{proof}

For a source edge $f$, let $P_f=\{0,1,2\}\setminus\{\tau(f)\}$.
Color permutations preserve properness, Kempe chains and exchanges: an edge is colored in $P$ before a permutation $\psi$ exactly when it is colored in $\psi(P)$ afterward.
We use the permutation $\tau_z$ to transfer each local gate coloring to the actual Tait colors at $z$.

\begin{definition}\label{def:encoder}
For a legal source configuration $x$, the \emph{canonical coloring} $C_x$ of $F^3_X$ is prescribed as follows.
\begin{enumerate}[(a)]
\item The gate at $z$ has the coloring $\tau_z\circ\Sigma_{\theta(z)}(x|_z)$, and the triangle edge $S^z_iS^z_j$ has color $\tau_z(\ell)$, where $\{i,j,\ell\}=\{0,1,2\}$.
\item For a wire $f$ with first end $u$, second end $v$ and $q=q_f$, every connector $\kappa_i$ with $0\leq i<q$ has first color $\tau_u(\gamma_{\theta(u)}(s_u(f)))$, while $\kappa_q$ has first color $\tau_v(\gamma_{\theta(v)}(s_v(f)))$.
The equality indicator at $\kappa_0$ is that of its gate slot, and the indicators alternate along the wire.
An equal connector repeats its first color, and an unequal connector has the other color of $P_f$ in its second coordinate.
Each wire negator receives the unique extension of \cref{lem:w3}(a).
\end{enumerate}
\end{definition}

The last connector may have a different first color from the preceding connectors, also on a wire with only two negators.
The negator extension allows either choice; only opposite equality types are needed.

\begin{lemma}\label{lem:encoder-local}
For every legal source configuration $x$, $C_x$ is a proper coloring of $F^3_X$ and $x_{C_x}=x$.
\end{lemma}

\begin{proof}
We check agreement at every identification and properness at every triangle vertex.
The first coordinates at both ends of a wire belong to $P_f$, since $\gamma_\theta(j)\ne j$.
The initial slot state agrees by definition, and the parity calculation of \cref{lem:projection} makes the final equality indicator agree with the slot at the other end.
Its first coordinate was chosen from that end, so its full ordered state agrees as well.
Consecutive connector types are opposite, so every wire negator is properly colored by \cref{lem:w3}(a).
Its residual color is $\tau(f)$, the color absent from $P_f$.
At $S^z_j$ the two triangle colors are the two Tait colors other than $\tau_z(j)$, and the attached residual edge has color $\tau_z(j)$.
Hence, properness holds at every triangle vertex, and at the gates it holds by \cref{lem:gate-canonical}.
All slot equality types represent the local states of $x$, so the projection is $x$.
\end{proof}

\subsection{The Nine Local Transitions}\label{sec:direct-paths}

The canonical colorings in \cref{sec:encoder} give explicit endpoints for every local transition.
We now find a single Kempe path in each gate that changes the active slot and reaches exactly the canonical target.

The path starts at the second edge of the changing ordered slot, ends at a residual dangling edge, and avoids all other slot edges.
A \emph{boundary trace} lists the connector and residual edges in their order along the path; \cref{lem:w3}(b) or~(c) supplies the intervening internal negator edges.

\begin{table}[htbp]
\centering\small
\begin{adjustbox}{max width=\textwidth}
\begin{tabular}{ccl}
gate and transition & pair & boundary trace\\\hline
$\mathrm{OR}:001\leftrightarrow011$ & $\{0,2\}$ & $d(g_1),b(g_1)=a(g_2),e(g_2)$\\
$\mathrm{OR}:011\leftrightarrow010$ & $\{0,1\}$ & $d(g_2),b(g_2),e(T)$\\
$\mathrm{OR}:010\leftrightarrow110$ & $\{1,2\}$ & $c(g_0),a(g_0),e(T)$\\
$\mathrm{OR}:110\leftrightarrow100$ & $\{0,2\}$ & $d(g_1),b(g_1)=a(g_2),b(g_2),e(T)$\\
$\mathrm{AND}:011\leftrightarrow111$ & $\{1,2\}$ & $b(g_0),d(g_0),e(g_1),a(g_1)$\\
$\mathrm{AND}:100\leftrightarrow110$, $101\leftrightarrow111$ & $\{0,2\}$ & $c(g_0),e(g_0)$\\
$\mathrm{AND}:100\leftrightarrow101$, $110\leftrightarrow111$ & $\{0,1\}$ & $d(g_1),b(g_1),r_2$
\end{tabular}
\end{adjustbox}
\caption{Seven boundary traces for all nine local transitions.
Each path starts at the second edge of the changing ordered slot and ends at a residual dangling edge of the gate.
The opposite direction uses the same path and color pair.}
\label{tab:direct}
\end{table}

The table records whole paths in the bare gates, not exchanges on a disconnected union of edges.
We verify this and the exact target coloring in the next lemma.

\begin{lemma}\label{lem:direct-gates}
If two permitted local states $\beta,\beta'$ differ only at slot $j$, then the trace in \cref{tab:direct} gives a Kempe path of $\Sigma_\theta(\beta)$ for $\{0,1,2\}\setminus\{j\}$.
It contains exactly the second edge of that slot among all slot edges, and exchanging its colors gives $\Sigma_\theta(\beta')$ on every edge of the gate.
\end{lemma}

\begin{proof}
We follow the negator routes from the changing slot and compare the boundary colors with the canonical target.
The three-color extension is unique once both connector states are fixed, so checking the changed connector colors and all edges outside negators determines the entire resulting coloring.
It suffices to use the directions displayed in the table, since a Kempe exchange is an involution and preserves its bichromatic component.
\begin{itemize}
\item \textbf{Case 1: the four $\mathrm{OR}$ transitions.}
For $001\to011$, the second path of $g_1$ joins $d(g_1)$ to $b(g_1)=a(g_2)$.
In $g_2$, \cref{lem:w3}(c) joins this edge to $e(g_2)$.
For $011\to010$, the second path of $g_2$ reaches $b(g_2)$, then passes through $T$ to $e(T)$.
For $010\to110$, the first path of $g_0$ joins $c(g_0)$ to $a(g_0)$, then passes through $T$ to $e(T)$.
Here $c(g_0)$ is the second edge of the ordered slot zero.
For $110\to100$, the second path of $g_1$ reaches $a(g_2)$, and \cref{lem:w3}(c) joins the two equal input edges of $g_2$; the path then passes through $T$ to $e(T)$.
The other slot edges do not lie on these routes.
At every traversed new vertex the remaining edge has the color outside the pair, so each trace gives a whole path ending at the stated free edge.
In the four cases the changed ring entries are, respectively, the third entry, the fourth entry, the first entry, and the third and fourth entries.
Transposing the listed pair changes $2221$ to $2201$, $2201$ to $2200$, $2200$ to $1200$, and $1200$ to $1222$, respectively.
The changing slot has its second coordinate transposed, all other slot coordinates are fixed, and the color at $e(T)$ is the third color at $T$.
Thus, all boundary and additional edges agree with the target, and uniqueness gives agreement inside every negator.
\item \textbf{Case 2: the five $\mathrm{AND}$ transitions.}
For $011\to111$, the second path of $g_0$ joins $b(g_0)$ to $d(g_0)$.
At $x_1$ the path continues to $e(g_1)$ because $s_1$ has color $0$, and \cref{lem:w3}(c) joins $e(g_1)$ to $a(g_1)$.
For the two transitions at slot one, \cref{lem:w3}(c) joins $c(g_0)$ to $e(g_0)$.
For the two transitions at slot two, the second path of $g_1$ joins $d(g_1)$ to $b(g_1)$, then passes through $x_2$ to $r_2$ because $r'_2$ has color $2$.
These paths avoid every other slot edge and end at a dangling residual edge.
For slot zero, the changing displayed colors are $b(g_0),d(g_0),a(g_1),e(g_1)$, each transposed by $\{1,2\}$.
For slot one they are $c(g_0),e(g_0)$, each transposed by $\{0,2\}$, and for slot two they are $d(g_1),b(g_1),r_2$, each transposed by $\{0,1\}$.
The formulas of \cref{def:gate-colouring} leave every other displayed color fixed.
Uniqueness inside the negators therefore gives exactly the canonical target on all edges.
\end{itemize}
\end{proof}

\subsection{One Global Kempe Path}\label{sec:composition}

The local paths of \cref{sec:direct-paths} change only the active slot at each endpoint gate.
We now join them through the wire to simulate a legal source reversal by one Kempe change of the whole network.

Their free ends remain dangling in the whole network, and no edge of an inactive slot lies on them.
These two facts remove the need to normalize colorings near inactive wires.

\begin{theorem}\label{thm:simulation-three}
For every legal reversal $x\to y$ of the source graph, a single Kempe change of $F^3_X$ sends $C_x$ to $C_y$.
Consequently, $x\sim_X y$ implies $C_x\sim_{F^3_X}C_y$.
\end{theorem}

\begin{proof}
We concatenate the two gate paths with the second path through every negator of the reversed edge's wire, and check that their union is a whole bichromatic component.
Let the reversed edge be $f=uv$ and let $P=P_f$.
Since the source graph $X$ is protected, each endpoint transition is one of those in \cref{tab:direct}.
After applying its local permutation $\tau_z$, \cref{lem:direct-gates} supplies a $P$-path from the second edge of the active slot to a residual dangling edge at each endpoint gate.
Every connector of the wire is colored in $P$, so \cref{lem:w3}(b) supplies a second path through every wire negator.
These paths concatenate through their common second connector edges into a connected edge set $Y$.

At an internal vertex of a gate or wire negator that $Y$ visits, its local path contains both incident edges colored in $P$.
All internal vertices of distinct copies are disjoint.
An identification joins the internal ends of the common connector edge; both selected paths use its second edge and avoid its first edge.
The residual edges of the active wire have color $\tau(f)\notin P$, so $Y$ does not enter a triangle.
Its two gate endpoints remain dangling residual edges of the full network.
Thus, $Y$ contains both $P$-edges at every internal vertex it visits and has exactly two terminal ends.
It is one whole Kempe path.

Exchanging on $Y$ gives the canonical target on both gates by \cref{lem:direct-gates}.
On the active wire, every first coordinate stays fixed and every second coordinate is transposed; its connector equality types therefore all switch, precisely as required by $y$.
Uniqueness in \cref{lem:w3}(a) gives the canonical target inside each active negator.
Every inactive slot is unchanged, so every inactive wire keeps its canonical connector states and coloring.
Triangles and other gates are unchanged as well, hence the result is $C_y$ on the full network.

Parallel source edges cause no exception: their wires have distinct copies and slots, and the active gate paths avoid the slots of all other wires, even those with the same endpoints.
Looplessness ensures that the two active endpoint gates are distinct.
Concatenating these single moves proves the reachability statement.
\end{proof}
 \subsection{From Three Colors to \texorpdfstring{$k$}{k} Colors}\label{sec:colours}\label{sec:added}

The single-move simulation of \cref{sec:composition} uses three colors.
We now extend its canonical colorings to $k$ colors by choosing fixed added edges, whose deletion leaves the cubic network with digons inserted.

The idea comes from Leven and Galil~\cite{LevenGalil1983}. They combine a coloring with three colors and a coloring of a $(k-3)$-regular graph with the remaining colors.
In our network, the added edges carry the colors $3,\dots,k-1$, one edge of each color at every internal vertex.
Their colors stay fixed throughout the simulation.
Deleting these edges leaves the network for three colors with pairs of parallel edges inserted into some edges.
The difficulty is to choose the added edges so that every negator has precisely this form.
The first inverting component of each negator becomes the negator for three colors.
Each further component leaves four pairs of parallel edges on the two output edges.
The rule below makes these choices once and for all.

\begin{definition}\label{def:added-colours}
Let $k\geq 4$.
In every copy of $H_i$ in $F^k_X$ and for $2\leq j\leq k-2$, the color $j+1$ is placed on one edge of each of the following multiple edges: on $u$, $o$, $q$ and $c$ if $i<j$; on $t$, $v$, $r$ and $e$ if $i=j$; and on $u$, $p$, $s$ and $a$ if $i>j$.
In every $\mathrm{OR}$ gate, one of the $k-2$ dangling edges at $T$ is left for the colors $0$, $1$ and $2$, and the other $k-3$ of them get the colors $3,\dots,k-1$, one color each.
In every $\mathrm{AND}$ gate, the same is done with the bundles $r_0$ at $x_0$ and $r'_2$ at $x_2$.
The edges are chosen once and for all, with distinct colors on distinct edges of the same multiple edge, and the edges that get a color in this way are the \emph{added edges} of $F^k_X$.
For $k=3$, there are no added edges.
\end{definition}

\Cref{def:added-colours} places at most as many colors on each multiple edge as it has edges; this bound is part of the proof of the next lemma.
Its second part uses the following operation.

\begin{definition}\label{def:digon}
Inserting a \emph{digon}~\cite{MazakRajnikSkoviera2022} into an edge of a multigraph replaces the edge by a path of length three through two new vertices and adds a second edge joining the two new vertices. The two edges that join the two new vertices are the parallel edges of the digon, and we fix one of them as its \emph{first parallel edge}.
\end{definition}

A digon together with its two further edges is the simplest multipole of degree three with two dangling edges, and in every coloring with three colors its two further edges have the same color.
Digons are inserted into an edge $f$ \emph{in series}: each further digon is inserted into an edge of the path that replaces $f$ and is not a parallel edge of an earlier digon.
The edges of the resulting path that are not parallel edges of a digon lie \emph{over}~$f$.

\begin{lemma}\label{lem:reduct}
Let $k\geq 3$.
\begin{enumerate}[(a)]
\item Every internal vertex of $F^k_X$ is incident with exactly one added edge of each color $h\in\{3,\dots,k-1\}$, and no edge of a connector of a negator, no residual edge $e$ of a copy of $H_1$, no edge of a slot of a gate and no triangle edge is an added edge.
\item Deleting from $F^k_X$ the added edges, together with the terminal vertices of the deleted dangling edges, gives the multipole $F^3_X$ with digons inserted into edges.
More precisely, it is obtained from $F^3_X$ by inserting, for every copy of $W_3$ in $F^3_X$, $3(k-3)$ digons into the first edge and $k-3$ digons into the second edge of its output connector, and no other digons.
\end{enumerate}
\end{lemma}

\begin{proof}
We check~(a) vertex by vertex, and we obtain~(b) by counting the edges that remain in each multiple edge.
For $k=3$, there is nothing to prove, so let $k\geq 4$ and $2\leq j\leq k-2$.
\begin{enumerate}[(a)]
\item In a copy of $H_i$, the color $j+1$ lies on $u$, $o$, $q$ and $c$ if $i<j$, which meet the vertices $B$ and $D$, $M$ and $L$, $A$ and $R$, and $C$, respectively; on $t$, $v$, $r$ and $e$ if $i=j$, which meet $B$ and $L$, $D$ and $R$, $A$ and $C$, and $M$; and on $u$, $p$, $s$ and $a$ if $i>j$, which meet $B$ and $D$, $M$ and $R$, $L$ and $C$, and $A$.
In each case, every vertex of $H_i$ is incident with exactly one edge of color $j+1$ of $H_i$.
The vertex $Z_i$ is incident with the edges of $c$ in $H_i$, which contain the color $j+1$ exactly when $i<j$, and with the edges of $a$ in $H_{i+1}$, which contain it exactly when $i+1>j$, so exactly one edge at $Z_i$ has the color $j+1$.
At $T$, $x_0$ and $x_2$, the respective residual bundles receive one edge of each added color.
At $x_1$, these edges are the residual edges $e_i(g_1)$ for $i\geq2$, whose colors are $i+1$.
At a vertex $S^z_\ell$ of a triangle, the added edges are the residual edges $e$ of the copies $H_i$ with $i\geq 2$ in the negator of the adjacent wire, and the one of $H_i$ has the color $i+1$.
The remaining statements follow from \cref{def:added-colours}.
Indeed, the connector edges of $W_k$ are the single edges $a$ and $b$ of $H_1$ and $c$ and $d$ of $H_{k-2}$, and the color $j+1$ lies on $a$ only if $i>j\geq 2$, on $c$ only if $i<j\leq k-2$, and never on $b$ or $d$.
The slot edges of the $\mathrm{AND}$ gate are connector edges of its negators or $s_0,s_1$, and neither these nor the triangle edges receive an added color.
\item In $H_1$, each of $u$, $o$, $q$ and $c$ has $k-2$ edges and receives the $k-3$ colors $j+1$, since $1<j$ for all $j$, and the other multiple edges of $H_1$ are single edges that receive no color.
Hence, $H_1$ becomes the inverting component $H_1$ for $k=3$, which is $W_3$.
For $2\leq i\leq k-2$, exactly $k-2-i$ of the indices $j$ are greater than $i$ and $i-2$ are smaller, so $u$, $p$, $s$ and $a$ keep two edges each, $o$, $q$ and $c$ keep one edge each, and $t$, $v$, $r$ and $e$ keep no edge; for example, $u$ keeps $(k-2)-(k-2-i)-(i-2)=2$ edges, and $o$ keeps $(k-1-i)-(k-2-i)=1$ edge.
In particular, every multiple edge has at least as many edges as the colors placed on it.
The vertex $Z_{i-1}$ keeps one edge of $c$ of $H_{i-1}$ and the two edges of $a$ of $H_i$.
For $2\leq i\leq k-2$, let $A_i$, $B_i$, $C_i$, $D_i$, $L_i$, $M_i$ and $R_i$ denote the vertices $A$, $B$, $C$, $D$, $L$, $M$ and $R$ of $H_i$.
Hence, after the deletion, the pairs $\{Z_{i-1},A_i\}$, $\{R_i,M_i\}$ and $\{L_i,C_i\}$ are digons on the path formed by $c$ of $H_{i-1}$, by the vertices $Z_{i-1}$, $A_i$, $R_i$, $M_i$, $L_i$ and $C_i$ and by $c$ of $H_i$, and the pair $\{B_i,D_i\}$ is a digon on the path formed by the edge $d$ of $H_{i-1}$, which is $b$ of $H_i$, and by $d$ of $H_i$.
Induction on $i$ shows that the negator becomes $W_3$ with $3(k-3)$ digons inserted into the first edge of its output connector and $k-3$ digons inserted into the second one, and its only remaining residual edge is $e$ of $H_1$.
At $T$, $x_0$ and $x_2$, exactly one edge of the respective residual bundle remains.
At $x_1$, only $e_1(g_1)$ remains from the attached residual bundle, together with $d(g_0)$ and $s_1$.
Thus, each gate becomes its cubic gate with the digons just described in its negators.
Every vertex $S^z_\ell$ keeps its two triangle edges and the residual edge $e$ of $H_1$ of the adjacent negator.
Of the identifications of \cref{def:network}, those of connector edges and of the residual edges $e$ of $H_1$ involve no added edge by~(a), and they are exactly the identifications of $F^3_X$, while the attached residual edges $e$ of $H_i$ with $i\geq 2$ are deleted.
The digons stay on the physical output edges $c,d$ of each negator.
At the reversed $\mathrm{OR}$ slot zero, these are respectively the second and first ordered slot edges.
When two output edges are identified, their digons occur in series over the same cubic edge.
Thus, reversed slot order and parallel source edges do not change the deletion statement.
\end{enumerate}
\end{proof}

\subsection{Lifting the Moves}\label{sec:lift}

\Cref{sec:colours} identified the cubic network with inserted digons inside the network for $k$ colors.
We now lift its colorings and Kempe moves through those digons while keeping the added edges fixed.

A coloring extends over a digon by assigning its parallel edges the two colors missing from the replaced edge.
A fixed order for these colors makes the extension unique.
Every original Kempe chain extends through each digon on its route along one parallel edge.
The only other Kempe chains are the two parallel edges of individual digons.
An exchange along an extended chain may reverse the prescribed order of the colors on some pairs.
One move on each affected pair restores that order.
Hence, sequences of moves lift, and the extended canonical colorings still store the same legal configurations.

The next lemma shows that inserting digons does not change what moves can do.
For a coloring $c$ of a multipole $F$ of degree three and a multipole $F'$ obtained from $F$ by inserting digons into edges, let $\widetilde c$ be the map on $E(F')$ that gives every edge over an edge $f$ of $F$ the color $c(f)$ and, for every digon inserted into $f$, gives the first parallel edge the smaller and the other parallel edge the larger of the two colors other than $c(f)$.
At every vertex of a digon, the edge over $f$ and the two parallel edges have the three distinct colors, and every vertex of $F$ has in $F'$ edges over its edges in $F$, with their colors under $c$, so $\widetilde c$ is a coloring of $F'$.

\begin{lemma}\label{lem:digon}
Let $F'$ be obtained from a multipole $F$ of degree three by inserting digons into edges in series, let $c$ be a coloring of $F$, and let $P\subseteq\{0,1,2\}$ be a color pair.
\begin{enumerate}[(a)]
\item For every Kempe chain $Y$ of $c$ for $P$, the edges over the edges of $Y$, together with the parallel edge colored in $P$ of every digon inserted into an edge of $Y$, form a Kempe chain $\widetilde Y$ of $\widetilde c$ for $P$.
Every Kempe chain of $\widetilde c$ for $P$ that is not obtained in this way consists of the two parallel edges of one digon.
\item If $c'$ is obtained from $c$ by a move, then $\widetilde{c'}$ is obtained from $\widetilde c$ by a sequence of moves of $F'$.
\end{enumerate}
\end{lemma}

\begin{proof}
We look at one digon at a time, and then we exchange along $\widetilde Y$ and repair the order of the parallel edges.
Let $f$ be an edge of $F$ into which a digon is inserted.
If $c(f)\in P$, then the edges colored in $P$ at the digon are its two edges over $f$ and the parallel edge that has the other color of $P$, and they form a path through the digon.
If $c(f)\notin P$, then the edges colored in $P$ at the digon are its two parallel edges, which form a Kempe chain of $\widetilde c$ for $P$ on their own.
\begin{enumerate}[(a)]
\item At every vertex of $F$, the edges colored in $P$ under $\widetilde c$ are the edges over the edges colored in $P$ under $c$ at this vertex.
Together with the paths through the digons, this shows that $\widetilde Y$ is a component of the subgraph formed by the edges colored in $P$ under $\widetilde c$, and that the only components that contain no edge over an edge of $F$ are the pairs of parallel edges of digons inserted into edges $f$ with $c(f)\notin P$.
\item Suppose that $c'$ is obtained from $c$ by exchanging the colors of $P$ on a Kempe chain $Y$.
The exchange of the colors of $P$ on $\widetilde Y$ gives every edge over an edge $f$ of $Y$ the color $c'(f)$, and it replaces the color of the parallel edge colored in $P$ of every digon inserted into $f$ by the other color of $P$.
The two parallel edges of such a digon then have the two colors other than $c'(f)$, possibly in the wrong order.
In that case, these two parallel edges form a Kempe chain for the color pair of their two colors, since the two edges over $f$ at the digon have the third color, and exchanging this color pair on them puts the colors into the order of $\widetilde{c'}$.
All other edges have the same color under $\widetilde c$ and $\widetilde{c'}$, so these moves lead from $\widetilde c$ to $\widetilde{c'}$.
\end{enumerate}
\end{proof}

Let $F'$ be the multipole obtained from $F^k_X$ by deleting the added edges; by \cref{lem:reduct}(b), it is $F^3_X$ with digons inserted into edges.
For a legal configuration $x$ of $X$, let $\widehat C_x$ be the coloring of $F^k_X$ that gives every added edge its color from \cref{def:added-colours} and every other edge its color under the coloring $\widetilde{C_x}$ of $F'$ defined before \cref{lem:digon}.
For $k=3$, the coloring $\widehat C_x$ is $C_x$ itself.
With these colorings, the simulation with three colors carries over to $k$ colors.

\begin{theorem}\label{thm:simulation}
For every $k\geq 3$, the coloring $\widehat C_x$ is a coloring of $F^k_X$ whose projection is $x$, and $x\sim_X y$ implies $\widehat C_x\sim_{F^k_X}\widehat C_y$.
\end{theorem}

\begin{proof}
We check properness and the projection vertex by vertex, and we lift the moves of \cref{thm:simulation-three} through \cref{lem:digon}.
By \cref{lem:reduct}(a), every internal vertex of $F^k_X$ is incident with one added edge of each color $3,\dots,k-1$, and its other three edges are its edges in $F'$, which have three distinct colors of $\{0,1,2\}$ under $\widetilde{C_x}$.
Hence, $\widehat C_x$ is a coloring of $F^k_X$.
By \cref{lem:reduct}(a), the edges of the connectors are edges of $F'$, and each of them is an edge over an edge of a connector of $F^3_X$, so every connector of a gate has the same equality type under $\widehat C_x$ as the corresponding connector of $F^3_X$ under $C_x$.
Since the projection depends only on these equality types, the projection of $\widehat C_x$ is the projection of $C_x$, which is $x$ by \cref{lem:encoder-local}.

Now let $x\sim_X y$.
By \cref{thm:simulation-three}, $C_y$ is obtained from $C_x$ by a sequence of moves of $F^3_X$, each with a color pair $P\subseteq\{0,1,2\}$, and by \cref{lem:digon}(b), $\widetilde{C_y}$ is obtained from $\widetilde{C_x}$ by a sequence of moves of $F'$ with such color pairs.
The added edges have colors outside $\{0,1,2\}$, so for every color pair $P\subseteq\{0,1,2\}$, the edges colored in $P$ and the Kempe chains for $P$ are the same in $F'$ and in $F^k_X$ when the added edges keep their colors.
Hence, each of these moves is a move of $F^k_X$ that keeps the colors of the added edges, and $\widehat C_x\sim_{F^k_X}\widehat C_y$.
\end{proof}
 \section{Proof of \texorpdfstring{\cref{thm:main}}{Theorem 1.1}}\label{sec:closing}

The simulation of \cref{sec:three} now gives reachability between the chosen network colorings.
We turn the network into a regular multigraph by joining two copies in \cref{sec:double}, then remove parallel edges in \cref{sec:blocks}.
Finally, \cref{sec:main-proof} combines the simulation and projection to prove \cref{thm:main}.

The two copies are joined along their dangling edges.
Every vertex then has degree $k$, but parallel edges remain.
Each vertex is then replaced by a copy of $K_{k,k-1}$.
The edges at the old vertex receive distinct ends in this copy, which removes all parallel edges.
Both constructions preserve connectivity.

For reachability, the difficulty is to control every coloring of the new graph, including those outside our chosen family.
Each such coloring can be read back before the construction.
A move then becomes a sequence of moves, or changes nothing.
Conversely, our chosen colorings and their sequences of moves extend through both constructions.
The common observation is that an exchange on several Kempe chains can be performed one chain at a time.
The next lemma recognizes such unions by checking that both edges of the color pair occur together at each internal vertex.

\begin{lemma}\label{lem:closed-set}
Let $c$ be a coloring of a multipole of degree $k$, let $P$ be a color pair, and let $W$ be a set of edges colored in $P$ such that, for every edge of $W$ and every internal vertex incident with it, the other edge colored in $P$ at this vertex also lies in $W$.
Then $W$ is a union of Kempe chains of $c$ for $P$, and exchanging the colors of $P$ on all edges of $W$ is a sequence of moves, one for each of these Kempe chains.
\end{lemma}

\begin{proof}
We show that $W$ is closed under adjacency in the subgraph formed by the edges colored in $P$.
Two edges colored in $P$ are adjacent in this subgraph exactly when they have a common internal vertex, because a terminal vertex is incident with only one edge, and every internal vertex is incident with exactly two edges colored in $P$.
Hence, by the hypothesis, every edge colored in $P$ that has a common end with an edge of $W$ lies in $W$, so $W$ is a union of components of this subgraph, that is, of Kempe chains for $P$.
An exchange on one of these Kempe chains does not change the set of edges colored in $P$, so the other Kempe chains in $W$ remain Kempe chains, and the exchanges can be made one at a time.
\end{proof}

\subsection{Two Copies}\label{sec:double}

\Cref{lem:closed-set} lets us handle an exchange on several Kempe chains.
We now join two copies to remove the dangling edges while preserving reachability between the colorings chosen for the reduction.
We use the construction of Holyer~\cite{Holyer1981} and Leven and Galil~\cite{LevenGalil1983}. 

Two copies of each dangling edge merge into one edge joining the copies of its internal end.
Every vertex keeps its $k$ incident edges, so the result is $k$-regular.

Repeating a coloring on both copies gives a coloring of the new multigraph.
Repeating a move on both copies gives an exchange on a union of Kempe chains.
For the converse, the two copies may have different colorings, so we read only one copy.
A merged edge is read as the dangling edge of that copy.
Every move then exchanges the same two colors on a union of Kempe chains in the original multipole.

Let $F$ be a multipole of degree $k$.
The multigraph $D(F)$ has the vertices $(g,j)$ for the internal vertices $g$ of $F$ and $j\in\{0,1\}$, and the following edges.
For every internal edge $f=gg'$ of $F$ and every $j\in\{0,1\}$, it has an edge $\iota_j(f)$ that joins $(g,j)$ and $(g',j)$, and for every dangling edge $f$ of $F$ at an internal vertex $g$, it has one edge $\iota_0(f)=\iota_1(f)$ that joins $(g,0)$ and $(g,1)$, which is called a \emph{cross edge}.
Thus, $D(F)$ consists of two copies of the internal vertices and the internal edges of $F$, and every dangling edge of $F$ becomes a cross edge between its two copies; the dangling edges at one vertex give parallel cross edges.
For every $j\in\{0,1\}$, the map $\iota_j$ from $E(F)$ to $E(D(F))$ is injective, every edge of $D(F)$ is an image of $\iota_0$ or of $\iota_1$, and $\iota_0(f)=\iota_1(f')$ holds only when $f=f'$ is a dangling edge, so the diagonal coloring below is well defined.
For a coloring $c$ of $F$, the \emph{diagonal coloring} $\mathrm{diag}(c)$ of $D(F)$ gives the edge $\iota_j(f)$ the color $c(f)$, for every edge $f$ of $F$ and every $j\in\{0,1\}$.
For a map $\varphi$ from $E(D(F))$ to the colors and for $j\in\{0,1\}$, the map $\mathrm{pr}_j(\varphi)$ from $E(F)$ to the colors gives every edge $f$ the color $\varphi(\iota_j(f))$; it reads the internal edges of $F$ in the copy $j$ and every dangling edge on its cross edge.
The next lemma collects the elementary properties of these objects.

\begin{lemma}\label{lem:closure-graph}
Let $F$ be a multipole of degree $k$.
\begin{enumerate}[(a)]
\item The multigraph $D(F)$ is $k$-regular, and for every internal vertex $g$ of $F$ and every $j\in\{0,1\}$, the map $\iota_j$ is a bijection from the $k$ edges of $F$ at $g$ onto the $k$ edges of $D(F)$ at $(g,j)$.
\item For every coloring $c$ of $F$, the diagonal coloring $\mathrm{diag}(c)$ is a coloring of $D(F)$, and $\mathrm{pr}_0(\mathrm{diag}(c))=\mathrm{pr}_1(\mathrm{diag}(c))=c$.
\item For every coloring $\varphi$ of $D(F)$ and every $j\in\{0,1\}$, the map $\mathrm{pr}_j(\varphi)$ is a coloring of $F$.
\item If $F$ is connected and has at least one dangling edge, then $D(F)$ is connected.
\end{enumerate}
\end{lemma}

\begin{proof}
The first three parts follow from a description of the edges of $D(F)$ at a vertex $(g,j)$, and the last one from the connectivity of the two copies.
\begin{enumerate}[(a)]
\item An internal edge of $F$ has two distinct internal ends, and $(g,0)\neq(g,1)$, so $D(F)$ has no loops.
The edges of $D(F)$ at $(g,j)$ are the edges $\iota_j(f)$ for the internal edges $f$ of $F$ at $g$ and the cross edges $\iota_j(f)$ for the dangling edges $f$ of $F$ at $g$.
These are exactly the images under $\iota_j$ of the $k$ edges of $F$ at $g$, and they are distinct because $\iota_j$ is injective.
\item By~(a), $\mathrm{diag}(c)$ gives the $k$ edges at $(g,j)$ the colors of the $k$ edges of $F$ at $g$ under $c$, which are distinct.
Moreover, $\mathrm{pr}_j(\mathrm{diag}(c))(f)=\mathrm{diag}(c)(\iota_j(f))=c(f)$ for every edge $f$ of $F$.
\item By~(a), $\mathrm{pr}_j(\varphi)$ gives the $k$ edges of $F$ at $g$ the colors of the $k$ edges of $D(F)$ at $(g,j)$ under $\varphi$, which are distinct.
\item A terminal vertex has degree one, so no path of $F$ between two internal vertices passes through a terminal vertex, and the internal vertices and internal edges of $F$ form a connected multigraph.
Each copy in $D(F)$ is therefore connected, and a cross edge joins the two copies.
\end{enumerate}
\end{proof}

The next lemma is the statement on reachability that we need.
Its second part must handle colorings of $D(F)$ that are not diagonal, and it reads them through one copy.

\begin{lemma}\label{lem:diagonal-closure}
Let $F$ be a multipole of degree $k$.
\begin{enumerate}[(a)]
\item If $c\sim_F c'$, then $\mathrm{diag}(c)\sim_{D(F)}\mathrm{diag}(c')$.
\item If a coloring $\varphi'$ of $D(F)$ is obtained from a coloring $\varphi$ by a move with the color pair $P$, then, for each $j\in\{0,1\}$, the coloring $\mathrm{pr}_j(\varphi')$ is obtained from $\mathrm{pr}_j(\varphi)$ by exchanging the colors of $P$ on a set of Kempe chains of $\mathrm{pr}_j(\varphi)$ for $P$.
\end{enumerate}
\end{lemma}

\begin{proof}
We compare the edges colored in $P$ at corresponding vertices, and then we apply \cref{lem:closed-set} in $D(F)$ for~(a) and in $F$ for~(b).
By \cref{lem:closure-graph}(a), for every internal vertex $g$ of $F$ and every $i\in\{0,1\}$, the edges colored in $P$ under $\mathrm{diag}(c)$ at $(g,i)$ are the images under $\iota_i$ of the edges colored in $P$ under $c$ at $g$, and the edges colored in $P$ under $\varphi$ at $(g,i)$ are the images under $\iota_i$ of the edges colored in $P$ under $\mathrm{pr}_i(\varphi)$ at $g$.
\begin{enumerate}[(a)]
\item It suffices to treat one move, since the relations are generated by moves.
Suppose that $c'$ is obtained from $c$ by exchanging the colors of a color pair $P$ on a Kempe chain $Y$ of $c$, and let $W=\iota_0(Y)\cup\iota_1(Y)$, a set of edges colored in $P$ under $\mathrm{diag}(c)$.
Let $\iota_j(f)$ be an edge of $W$ with $f\in Y$, and let $(g,i)$ be an end of it.
Then $f$ is an edge of $F$ at $g$ with $\iota_i(f)=\iota_j(f)$, since either $f$ is an internal edge and $i=j$, or $f$ is a dangling edge and $\iota_0(f)=\iota_1(f)$.
The other edge colored in $P$ at $g$ lies in $Y$, because $Y$ is a Kempe chain, so the other edge colored in $P$ at $(g,i)$ lies in $\iota_i(Y)\subseteq W$.
By \cref{lem:closed-set}, exchanging the colors of $P$ on $W$ is a sequence of moves of $D(F)$.
An edge $\iota_j(f)$ lies in $W$ exactly when $f\in Y$, so this exchange turns $\mathrm{diag}(c)$ into $\mathrm{diag}(c')$.
\item Let $Z$ be the Kempe chain of $\varphi$ on which the move exchanges the colors of $P$, and let $Y_j=\{f\in E(F):\iota_j(f)\in Z\}$, a set of edges colored in $P$ under $\mathrm{pr}_j(\varphi)$.
Let $f\in Y_j$, let $g$ be an internal vertex of $F$ at $f$, and let $f'$ be the other edge colored in $P$ under $\mathrm{pr}_j(\varphi)$ at $g$.
Then $\iota_j(f')$ is the other edge colored in $P$ under $\varphi$ at the end $(g,j)$ of $\iota_j(f)$, so it lies in the Kempe chain $Z$, and $f'\in Y_j$.
By \cref{lem:closed-set}, $Y_j$ is a union of Kempe chains of $\mathrm{pr}_j(\varphi)$ for $P$, and since $\varphi'$ is obtained from $\varphi$ by exchanging the colors of $P$ on $Z$, the coloring $\mathrm{pr}_j(\varphi')$ is obtained from $\mathrm{pr}_j(\varphi)$ by exchanging them on $Y_j$.
\end{enumerate}
\end{proof}

\subsection{Removing Parallel Edges}\label{sec:blocks}

\Cref{sec:double} removed the dangling edges, but its regular multigraph can still have parallel edges.
We now replace each vertex by a copy of $K_{k,k-1}$ and prove that this preserves edge-Kempe reachability.
The original incident edges receive distinct ends on the larger side, so the result is simple and $k$-regular.

For $k\geq3$, Leven and Galil~\cite{LevenGalil1983} replaced $k$-regular multigraphs by simple $k$-regular graphs while preserving the existence of a proper $k$-edge-coloring, using a gadget they attribute to Izbicki. Their construction replaces bundles of parallel edges; here we replace vertices by copies of $K_{k,k-1}$ and prove preservation of edge-Kempe reachability between the prescribed extensions.

The Parity Lemma forces the $k$ edges leaving each copy to have distinct colors.
Thus, their colors give a proper edge-coloring of the original multigraph.
A chain entering a copy leaves through the other edge of its color pair.
Consequently, a move projects to one edge-Kempe change or changes no edge of the original multigraph.
Conversely, the prescribed extension of a coloring into every copy also lets sequences of moves extend.

Let $G$ be a $k$-regular multigraph.
At every vertex $g$ of $G$, the $k$ incident edges receive the labels $e_0(g),\dots,e_{k-1}(g)$.
For the graphs of our reduction, this numbering is chosen in an arbitrary but fixed way as part of the construction.
The graph $B(G)$ is obtained from $G$ by replacing every vertex $g$ with a copy $Q_g$ of the complete bipartite graph $K_{k,k-1}$, with the vertices $a^g_0,\dots,a^g_{k-1}$ on one side of the bipartition and $b^g_1,\dots,b^g_{k-1}$ on the other side.
For every edge $e$ of $G$ with ends $g$ and $g'$, the vertex $a^g_i$ with $e=e_i(g)$ is then joined to the vertex $a^{g'}_{i'}$ with $e=e_{i'}(g')$.
We call the new edge the \emph{link edge} of $e$; it leaves $Q_g$ and $Q_{g'}$, and every vertex $a^g_i$ is incident with exactly one link edge.
For a proper $k$-edge-coloring $\psi$ of $B(G)$, the map $\psi_G$ gives every edge of $G$ the color of its link edge.
For a proper $k$-edge-coloring $\varphi$ of $G$, the coloring $\overline{\varphi}$ of $B(G)$ gives every link edge the color of its edge of $G$ and, for $0\leq i\leq k-1$ and $1\leq j\leq k-1$, the edge $a^g_ib^g_j$ the color $\varphi(e_m(g))$, where $m\in\{0,\dots,k-1\}$ is the residue of $i+j$ modulo $k$.
The next lemma shows that this replacement keeps reachability in both directions between the colorings $\overline{\varphi}$.

\begin{lemma}\label{lem:blocks}
Let $G$ be a $k$-regular multigraph.
\begin{enumerate}[(a)]
\item The graph $B(G)$ is simple and $k$-regular, and it is connected if $G$ is connected.
\item For every proper $k$-edge-coloring $\psi$ of $B(G)$, the $k$ link edges that leave a copy $Q_g$ have distinct colors, so $\psi_G$ is a proper $k$-edge-coloring of $G$.
If $\psi'$ is obtained from $\psi$ by an edge-Kempe change with the color pair $P$, then either $\psi'_G=\psi_G$, or $\psi'_G$ is obtained from $\psi_G$ by an edge-Kempe change with the color pair $P$.
\item For every proper $k$-edge-coloring $\varphi$ of $G$, the coloring $\overline{\varphi}$ is a proper $k$-edge-coloring of $B(G)$ with $\overline{\varphi}_G=\varphi$, and if $\varphi\sim_G\varphi'$, then $\overline{\varphi}\sim_{B(G)}\overline{\varphi'}$.
\end{enumerate}
\end{lemma}

\begin{proof}
We use the Parity Lemma for~(b) and \cref{lem:closed-set} for~(c).
\begin{enumerate}[(a)]
\item The edges of a copy $Q_g$ form a simple graph, and the link edges join vertices of distinct copies, since $G$ has no loops.
Every vertex $a^g_i$ is incident with only one link edge, so no two link edges are parallel, and $B(G)$ is simple.
Every vertex $a^g_i$ has $k-1$ neighbors in $Q_g$ and one link edge, and every vertex $b^g_j$ has $k$ neighbors in $Q_g$, so $B(G)$ is $k$-regular.
Every copy $Q_g$ is connected, and the link edge of an edge $gg'$ of $G$ joins $Q_g$ and $Q_{g'}$, so $B(G)$ is connected if $G$ is connected.
\item The copy $Q_g$ with its $k$ link edges is a multipole of degree $k$ with $2k-1$ internal vertices, and $\psi$ restricts to a coloring of it.
By \cref{lem:parity}, every color occurs an odd number of times on these $k$ link edges, so every color occurs exactly once.
The link edges that leave $Q_g$ are the link edges of the edges of $G$ at $g$, so $\psi_G$ is a proper $k$-edge-coloring of $G$.

Let $Z$ be the Kempe chain of $\psi$ for $P$ on which the move exchanges the colors, and consider a copy $Q_g$.
Exactly two of its link edges are colored in $P$, and every vertex of $Q_g$ is incident with exactly two edges colored in $P$, so the edges of $Q_g$ colored in $P$ form a path joining the two ends of these two link edges, together with some cycles.
If $Z$ contains no link edge, then $\psi'_G=\psi_G$.
Otherwise, whenever $Z$ contains a link edge that leaves $Q_g$, it contains the path of $Q_g$ and the other link edge colored in $P$ that leaves $Q_g$.
Hence, the set $Y$ of edges of $G$ whose link edges lie in $Z$ contains, together with each of its edges and each end $g$ of this edge, the other edge colored in $P$ under $\psi_G$ at $g$.
By \cref{lem:closed-set}, $Y$ is a union of Kempe chains of $\psi_G$ for $P$, and since $Z$ is connected, so is the subgraph of $G$ formed by $Y$, so $Y$ is one Kempe chain.
The move exchanges the colors of $P$ on the link edges in $Z$, so $\psi'_G$ is obtained from $\psi_G$ by exchanging them on $Y$.
\item For every vertex $g$ of $G$, let $\sigma_g$ be the permutation of the colors that maps $m$ to $\varphi(e_m(g))$, which is a bijection since the edges at $g$ have distinct colors.
At the vertex $a^g_i$, the $k-1$ edges of $Q_g$ have the colors $\sigma_g(m)$ for the residues $m\neq i$, and the link edge has the color $\varphi(e_i(g))=\sigma_g(i)$.
At the vertex $b^g_j$, the $k$ edges have the colors $\sigma_g(m)$ for all residues $m$, since $i\mapsto i+j$ permutes the residues modulo $k$.
Hence, $\overline{\varphi}$ is proper, and $\overline{\varphi}_G=\varphi$ by definition.

It suffices to treat one edge-Kempe change.
Suppose that $\varphi'$ is obtained from $\varphi$ by exchanging the colors of a color pair $P$ on a Kempe chain $Y$ of $\varphi$, and let $\pi$ be the transposition of the two colors of $P$.
If $g$ is an end of an edge of $Y$, then both edges colored in $P$ at $g$ lie in $Y$, so $\varphi'(e_m(g))=\pi(\varphi(e_m(g)))$ for every $m$, and the edges of $Q_g$ satisfy $\overline{\varphi'}=\pi\circ\overline{\varphi}$.
For every other vertex $g$, the colorings $\overline{\varphi'}$ and $\overline{\varphi}$ agree on $Q_g$.
Let $W$ be the set of edges colored in $P$ under $\overline{\varphi}$ that lie in a copy $Q_g$ for an end $g$ of an edge of $Y$, or that are link edges of edges of $Y$.
Then $\overline{\varphi'}$ is obtained from $\overline{\varphi}$ by exchanging the colors of $P$ on $W$.
Every vertex of such a copy $Q_g$ has both of its edges colored in $P$ in $W$, since its link edge colored in $P$, if any, is the link edge of an edge colored in $P$ at $g$, which lies in $Y$.
A vertex of any other copy $Q_{g'}$ is incident with no edge of $W$, because a link edge colored in $P$ that leaves $Q_{g'}$ belongs to an edge at $g'$, which does not lie in $Y$.
By \cref{lem:closed-set}, applied in $B(G)$, we have $\overline{\varphi}\sim_{B(G)}\overline{\varphi'}$.
\end{enumerate}
\end{proof}

For $k\geq 3$ and a source graph $X$, let $G^k_X=B(D(F^k_X))$, and for a legal configuration $x$ of $X$, let $\alpha_x=\overline{\mathrm{diag}(\widehat C_x)}$.
The next lemma shows that $G^k_X$ has the properties required by \cref{thm:main}.

\begin{lemma}\label{lem:connected}
The graph $G^k_X$ is a connected simple $k$-regular graph.
\end{lemma}

\begin{proof}
We show that $F^k_X$ is connected and has a dangling edge, and then we apply \cref{lem:closure-graph} and \cref{lem:blocks}(a).
Every inverting component is connected, since its multiple edges $u$, $t$, $v$, $o$, $p$, $q$, $s$ and $r$ have positive multiplicities and join its seven vertices into a connected multigraph.
The negator $W_k$ is connected, because the vertex $Z_i$ and the edge $d$ of $H_i$, which is the edge $b$ of $H_{i+1}$, join $H_i$ to $H_{i+1}$.
Every gate is connected by \cref{def:gates}: the $\mathrm{OR}$ ring joins its three negators, while $x_1$ joins the two $\mathrm{AND}$ negators through its nonempty residual bundle.
Every wire is connected, and it is joined to the gates at both ends of its edge, and every triangle is joined to the wires at its vertex $z$, because each vertex $S^z_j$ is incident with the $k-2\geq 1$ residual edges of the negator next to $z$ of the wire at slot $j$.
Since $X$ is connected, $F^k_X$ is connected, and it has a dangling edge, for example at $T$ in an $\mathrm{OR}$ gate or the edge $a(g_1)$ in an $\mathrm{AND}$ gate.
By \cref{lem:closure-graph}(a) and~(d), $D(F^k_X)$ is a connected $k$-regular multigraph, and by \cref{lem:blocks}(a), $G^k_X$ is a connected simple $k$-regular graph.
\end{proof}

\subsection{The Main Result}\label{sec:main-proof}

By \cref{sec:double,sec:blocks}, the network now gives a connected simple regular graph with colorings that preserve the simulation.
We combine those constructions with the projection from \cref{sec:reflection} to prove reachability in both directions.
We then count vertices, include the source preprocessing costs, and pass to the line graph to prove \cref{thm:main}.

A legal configuration gives its canonical coloring with $k$ colors, which is repeated on the two copies and extended into every copy of $K_{k,k-1}$.
Conversely, a sequence of edge-Kempe changes can be read back on one copy of the network.
Its moves project to sequences of legal reversals, which proves the required equivalence.

For fixed $k$, the number of vertices grows only linearly with that of the source graph.
Finally, passing to the line graph turns the edge-colorings into vertex-colorings of a connected claw-free graph.
Together with membership in $\PSPACE$, this reduction proves the main theorem.

The two directions of the reduction now combine into the following theorem, which is the combinatorial content of \cref{thm:main}.

\begin{theorem}\label{thm:equivalence}
Let $k\geq 3$, and let $X$ be a source graph as in \cref{cor:tait-source}.
For all legal configurations $x$ and $y$ of $X$, we have $x\sim_X y$ if and only if $\alpha_x\sim_{G^k_X}\alpha_y$.
\end{theorem}

\begin{proof}
We lift the simulation of \cref{thm:simulation} through the two constructions for the forward direction, and we project every edge-Kempe change back to the network for the backward direction.
\begin{itemize}
\item[($\Rightarrow$)] By \cref{thm:simulation}, $x\sim_X y$ implies $\widehat C_x\sim_{F^k_X}\widehat C_y$.
By \cref{lem:diagonal-closure}(a), $\mathrm{diag}(\widehat C_x)\sim_{D(F^k_X)}\mathrm{diag}(\widehat C_y)$, and by \cref{lem:blocks}(c), $\alpha_x\sim_{G^k_X}\alpha_y$.
\item[($\Leftarrow$)] Let $\alpha_x=\psi_0,\psi_1,\dots,\psi_m=\alpha_y$ be proper $k$-edge-colorings of $G^k_X$ in which each coloring is obtained from the previous one by an edge-Kempe change, and let $c_i=\mathrm{pr}_0((\psi_i)_{D(F^k_X)})$ for $0\leq i\leq m$.
By \cref{lem:blocks}(b) and \cref{lem:closure-graph}(c), every $c_i$ is a coloring of $F^k_X$.
By \cref{lem:blocks}(b) and \cref{lem:diagonal-closure}(b), for $1\leq i\leq m$, either $c_i=c_{i-1}$, or $c_i$ is obtained from $c_{i-1}$ by exchanging the colors of a color pair on a set of Kempe chains, so $x_{c_{i-1}}\sim_X x_{c_i}$ by \cref{lem:reflection}.
By \cref{lem:blocks}(c) and \cref{lem:closure-graph}(b), we have $c_0=\widehat C_x$ and $c_m=\widehat C_y$, whose projections are $x$ and $y$ by \cref{thm:simulation}.
Hence, $x\sim_X y$.
\end{itemize}
\end{proof}

We now account separately for the target construction and for the source preprocessing.
Let $N(F)$ denote the number of internal vertices of a multipole $F$.
The first bound below is measured on the protected, Tait-colored source supplied to the network.

\begin{proposition}\label{prop:size}
For every $k\geq3$ and every source graph $X$ of order $n$, we have
\[
N(F^k_X)\leq\frac{120k-247}{2}\,n,
\qquad |V(G^k_X)|\leq(2k-1)(120k-247)n.
\]
For fixed $k$, the graph $G^k_X$ and its endpoint coloring $\alpha_x$ are computable in linear time from the incidence list of $X$, its supplied data, and the legal configuration $x$.
\end{proposition}

\begin{proof}
We count the gates, triangles and individual wires, including parallel source edges separately.
Since $X$ is loopless and cubic, $|E(X)|=3n/2$.
The $\mathrm{OR}$ gate has $24k-50$ internal vertices, and the $\mathrm{AND}$ gate has $16k-31$.
Since the difference is $8k-19>0$, every gate has at most $24k-50$ internal vertices.
The triangles contribute $3n$ vertices, and every wire has at most three negators of $8k-17$ vertices each.
Therefore,
\[
N(F^k_X)\leq(24k-50)n+3n+\frac{3n}{2}\,3(8k-17)
=\frac{120k-247}{2}\,n.
\]
Doubling multiplies this count by two, and each simple block then has $2k-1$ vertices, which gives the second bound.
For fixed $k$, each source incidence introduces boundedly many vertices and edges.
The gate formulas, the negator extension, the added-color pattern, the diagonal extension and the modular block formula specify every endpoint color by bounded local work.
Hence, both outputs have the stated construction time.
\end{proof}

The next two consequences include the cost of protection and Tait preprocessing.
Their input $X_0$ is a simple connected source before protection; if it was obtained by Hearn and Demaine's preprocessing~\cite{HearnDemaine2005}, that earlier polynomial transformation is additional. 

\begin{corollary}\label{cor:edge-equivalence}
For every $k\geq3$ and every nonempty connected simple cubic $\mathrm{AND}/\mathrm{OR}$ constraint graph $X_0$, there are a connected simple $k$-regular graph $G$ with at most
\[
130(2k-1)(120k-247)|V(X_0)|
\]
vertices and a map $x\mapsto\varphi_x$ on the legal configurations of $X_0$ such that $x\sim_{X_0}y$ if and only if $\varphi_x\sim_G\varphi_y$.
The map is computable in polynomial time.
\end{corollary}

\begin{proof}
The protected replacement of Hearn and Demaine~\cite{HearnDemaine2009}, with its paired red--blue converter, gives a protected loopless multigraph with at most $13|V(X_0)|$ vertices. Its endpoint map preserves and reflects reachability by \cref{lem:protected-endpoint-completion}.
By \cref{prop:tait-source}, a second map gives the Tait-colored source $X$ with at most $130|V(X_0)|$ vertices.
Let $x'$ be the composite image of $x$, let $G=G^k_X$, and let $\varphi_x=\alpha_{x'}$.
The source maps preserve and reflect reachability, as does \cref{thm:equivalence} for the target map.
The order bound follows from \cref{prop:size}, and all the constructions and endpoint maps are polynomial-time.
\end{proof}

Reading the resulting edge-colorings on the line graph gives the same reduction for vertex colorings.

\begin{corollary}\label{cor:source-equivalence}
For the graph and endpoint map of \cref{cor:edge-equivalence}, $L(G)$ is connected, claw-free and $(2k-2)$-regular, and $x\sim_{X_0}y$ if and only if $\varphi_x\sim_{L(G)}\varphi_y$ when the edge-colorings are read as vertex colorings of $L(G)$.
\end{corollary}

\begin{proof}
The structural properties follow from \cref{lem:line-graphs}, since $G$ is connected, simple, $k$-regular and nonempty.
The equality of the two move relations follows from \cref{lem:transfer}; combining it with \cref{cor:edge-equivalence} gives the equivalence.
\end{proof}

\begin{proof}[Proof of \cref{thm:main}]
We combine source hardness, the constructed endpoint equivalence, and polynomial-space membership.
Membership for vertex colorings is \cref{prop:membership}.
For edge-colorings, \cref{lem:transfer} reduces the problem to vertex colorings of the line graph, whose order is $|E(G)|$ and whose construction takes polynomial time.
These upper bounds also apply to every graph restriction in the theorem.

For hardness, an instance of \cref{cor:tait-source} supplies $X,\tau,x_s,x_t$.
By \cref{prop:size}, $G^k_X$, $\alpha_{x_s}$ and $\alpha_{x_t}$ are computable in polynomial time.
By \cref{lem:connected,thm:equivalence}, this is an equivalent edge-Kempe instance on a connected simple $k$-regular graph.
Passing to its line graph is polynomial-time and preserves every move by \cref{lem:transfer}.
By \cref{lem:line-graphs}, that line graph is connected, claw-free and $(2k-2)$-regular.
The reductions prove all three hardness statements and hence the theorem.
\end{proof}

\subsection{Concluding Remarks}\label{sec:conclusion}

The preceding proof establishes $\PSPACE$-completeness for every fixed $k\geq3$ on line graphs of connected simple $k$-regular graphs, and hence on connected $(2k-2)$-regular claw-free graphs.
We finish with the smaller degree bounds and the structure of Kempe classes that this construction leaves unresolved.

For a fixed $k\geq 4$, every instance with $\Delta(G)\leq k$ is a yes-instance, and a claw-free graph with a proper $k$-coloring has maximum degree at most $2k-2$.
Our hard instances attain degree $2k-2$, so they do not settle the restrictions with maximum degree at most $d$ for $k+1\leq d\leq2k-3$.
In particular, we leave open whether \KR{} with four colors is $\PSPACE$-complete on claw-free graphs of maximum degree five.

The structure of the Kempe classes of Tait colorings is a second direction.
A complete invariant of these Kempe classes that is computable in polynomial time would decide in polynomial time whether two Tait colorings of a connected simple cubic graph are edge-Kempe equivalent, so by \cref{thm:main} for $k=3$ its existence would imply $\PClass=\PSPACE$.
Complete invariants without such a time bound exist: for every graph, Ohsugi and Tsuchiya~\cite{OhsugiTsuchiya2026} characterized Kempe equivalence by the membership of a binomial in an ideal built from the $2$-colorings of induced subgraphs. They decide it by normal forms with respect to a Gr\"obner basis of a related ideal~\cite{OhsugiTsuchiya2026}, and through \cref{lem:transfer} these normal forms separate the Kempe classes of Tait colorings. Their algorithm computes the basis from all stable sets of the graph, and the authors note that it is not feasible for large graphs~\cite{OhsugiTsuchiya2026}. Our result does not exclude useful invariants that give only necessary conditions for Kempe equivalence, and the corresponding questions for cubic bipartite graphs need a separate analysis, since the cubic graphs of our reduction need not be bipartite.
Some results of both kinds are known.
For a cubic graph embedded in the plane, Kotzig~\cite{Kotzig1977} showed that the sum over all vertices of a sign that records the cyclic order of the three colors around the vertex is invariant modulo $8$ under edge-Kempe changes. Every planar bipartite cubic graph has a single Kempe class of Tait colorings~\cite{belcastroHaas2014}, which for $3$-connected graphs follows from a result of Fisk, as noted in~\cite{belcastroHaas2014}. Ozeki~\cite{Ozeki2022} characterized the bipartite cubic graphs embeddable in the projective plane that have a single Kempe class of Tait colorings.
Such a graph has a single class exactly when it admits an embedding in the projective plane whose faces are homeomorphic to open disks and whose dual triangulation has chromatic number at least five. For every $n\geq 1$ there are bipartite cubic graphs of order $4n+2$ with $2^n$ such classes~\cite{belcastroHaas2014} (see also~\cite{GoedgebeurOstergard2022}). The reduction uses the global protection of the source to limit the local transitions that its gates must simulate.
A reduction for the smaller degree bounds would need to keep every target vertex below the degree used by our construction.

\section*{Acknowledgments}
Duc A.\ Hoang's research was supported by the Vietnam National University, Hanoi under the project QG.25.07 ``A study on reconfiguration problems from algorithmic and graph-theoretic perspectives''.

\section*{AI Declaration}
ChatGPT 6 Pro and Claude Fable 5.1 were used to assist with developing and reviewing mathematical arguments and drafting and revising the paper.
The author takes full responsibility for the correctness and content of the paper.

\printbibliography[title={References}]

\end{document}